\documentclass[a4paper, 11pt]{article} 
\usepackage[T1]{fontenc}  
\usepackage[english]{babel}
\usepackage{float}
\usepackage{floatpag}
\usepackage[document]{ragged2e}

\usepackage[normalem]{ulem}

\usepackage{setspace} 
\usepackage{ amsmath, amssymb, amscd, amsthm, array, arydshln,
                       etex, etoolbox, euscript, enumerate,
                       float,
                       graphicx,graphics,
                       hhline, 
                       lscape,
                       multicol, multirow,  
                       natbib,       
                       tabularx, lscape, tikz,      
                       rotating, setspace                                  
                       }	
\usepackage[justification=justified,singlelinecheck=false,skip=0pt]{caption}	
\usepackage{ragged2e}
\usepackage{titling}
\usepackage{xpatch}
\usepackage{titlesec}
\let\endtitlepage\relax
\newenvironment{mytitlepage}%
  {\begin{titlepage}\def\@thanks{}}%
  {\end{titlepage}}
\xpatchcmd\titlepage{\setcounter{page}\@ne}{}{}{}
\xpatchcmd\endtitlepage{\setcounter{page}\@ne}{}{}{}
\usepackage[lofdepth,lotdepth]{subfig}
\usepackage{xfrac}
\usepackage{booktabs}
\usepackage{authblk}
\usepackage{algorithm}
\usepackage{algpseudocode}

\usepackage{color} 
\definecolor{darkblue}{rgb}{0,0.08,0.45}
\definecolor{bred}{rgb}{0.8, 0.0, 0.0}
\definecolor{darkred}{rgb}{0.53, 0.08, 0.0}
\definecolor{darkgreen}{rgb}{0,128,0}
\usepackage[colorlinks=true]{hyperref}
\hypersetup{linkcolor=blue} 
\hypersetup{citecolor=blue}

\hypersetup{
    colorlinks=true,%
    citecolor=black,%
    filecolor=black,%
    linkcolor=black,%
    urlcolor=black
}

 \usetikzlibrary{arrows, patterns,decorations.pathreplacing }
\makeatletter
\renewcommand*{\@fnsymbol}[1]{\ensuremath{\ifcase#1\or \dagger \or \ddagger  
\or \mathsection \or * \or \mathparagraph\or \|\or **\or \dagger\dagger
\or \ddagger\ddagger \else\@ctrerr\fi}}
\makeatother

\makeatother 
\newtheorem{theo}{Proposition} 
\newtheorem{Proof}{Proof} 
  
\newtheorem{lem}{Lemma} 

\theoremstyle{remark}
\newtheorem{remark}{Remark}
\newcolumntype{K}[1]{>{\centering\arraybackslash}p{#1}}
\newtheorem{assumption}{} 

\newtheoremstyle{stylemodel}
{}
{}
{}
{}
{\bfseries}
{}
{\newline}
{}
\theoremstyle{stylemodel}

\DeclareMathOperator*{\argmin}{arg\,min}    
\newcommand{\1}[1]{\mathbf 1_{\{ {#1} \}}} 

\def\bx{\mathbf x}
\def\by{\mathbf y}
\def\bX{\mathbf X}
\def\bM{\mathbf M}
\def\bmu{\boldsymbol{\mu}}
\def\bu{\boldsymbol{u}}
\def\bepsilon{\boldsymbol{\epsilon}}
\def\argmax{\text{argmax}}
\def\bbeta{\boldsymbol{\beta}}
\def\btau{\boldsymbol{\tau}}
\def\btheta{\boldsymbol{\theta}}

\def\bSigma{\boldsymbol{\Sigma}}

\def\bomega{\boldsymbol{\omega}}
\def\bB{\boldsymbol{B}}
\def\bpsi{\boldsymbol{\psi}}
\def\ba{\boldsymbol{a}}
\def\NORM{\mathcal{N}}

\def\REAL{\mathbb{R}}
\def\ra{\rightarrow}
\newcommand{\pen}{\mathcal{P}_{\mathbf{SELO}}} 
\def\btau{\boldsymbol{\tau}}
\def\bz{\mathbf z}

\usepackage{etoolbox} 
\makeatletter
\patchcmd\maketitle{\@makefntext}{\@@@ddt}{}{}
\patchcmd\maketitle{\rlap}{\mbox}{}{}
\makeatother
\makeatletter
\def\blfootnote{\xdef\@thefnmark{}\@footnotetext}
\makeatother

\title{
 {\bf Selective linear segmentation for detecting relevant parameter changes}
}
\author[a,b,c]{Arnaud Dufays\footnote[1]{Corresponding author.}}
\author[b]{Elysee Aristide Houndetoungan}
\author[d]{Alain Co\"{e}n}
\affil[a]
{\small\emph{D\'{e}partement des sciences de gestion, Universit\'{e} Namur, Belgium}\vspace{-1pt}}
\affil[b]
{\small\emph{D\'{e}partement d'\'{e}conomique, Universit\'{e} Laval, Canada}\vspace{-1pt}}

\affil[c]{\small\emph{CRREP and CeReFiM associate researcher}\vspace{-1pt}}

\affil[d]
{\small\emph{Department of Finance, UQAM, Montr\'{e}al, QC, Canada}\vspace{-4pt}}
\date{June 2020}

\begin{document}
\onehalfspacing
\justifying
\begin{mytitlepage}
\maketitle
\blfootnote{
\emph{Email addresses}:
\texttt{arnaud.dufays@unamur.be} (Arnaud Dufays), \texttt{elysee-aristide.houndetoungan.1@ulaval.ca} (Aristide Houndetoungan), \texttt{coen.alain@uqam.ca} (Alain Co\"{e}n).
}

\begin{abstract}
\begin{spacing}{1.15} 
\noindent Change-point processes are one flexible approach to model long time series. We propose a method to uncover which model parameter truly vary when a change-point is detected. Given a set of breakpoints, we use a penalized likelihood approach to select the best set of parameters that changes over time and we prove that the penalty function leads to a consistent selection of the true model. Estimation is carried out via the deterministic annealing expectation-maximization algorithm. Our method accounts for model selection uncertainty and associates a probability to all the possible time-varying parameter specifications. Monte Carlo simulations highlight that the method works well for many time series models including heteroskedastic processes. For a sample of 14 Hedge funds (HF) strategies, using an asset based style pricing model, we shed light on the promising ability of our method to detect the time-varying dynamics of risk exposures as well as to forecast HF returns.
\end{spacing}

\vspace{1cm}


\thispagestyle{empty}	
	
\end{abstract}

\end{mytitlepage}

\newpage

\clearpage
\pagenumbering{arabic}

\section{Introduction}

Long time series are standard in this period of large publicly available datasets. Care is required when modeling such a time series since many of them span over critical events that may change the series dynamic. At least two statistical solutions exist to take into account these changes. On the one hand, a process with fixed parameters can be used but it needs to exhibit a rich and complex dynamic. This complexity often makes the model difficult to estimate and to interpret (see, for instance, long memory processes such as \cite{GPH83}). On the other hand,  one can rely on time-varying parameter (TVP) models and in particular Markov-switching and change-point (CP) processes since they allow for abrupt changes in the model parameters when a critical event affects the series dynamic \citep[see][]{Hamilton89,Bauwens2013}. This paper deals with CP linear regression models where we allow the mean parameters to change over time.   \\
The CP literature dates back to \cite{chernoff1964estimating} and is nowadays vast. Just focusing on linear regressions, \cite{andrews1993tests}, \cite{BaiPerron1998}, \cite{killick2012optimal}, \cite{fryzlewicz2014wild} and \cite{yau2016inference} develop prominent procedures to detect breakpoints. On the Bayesian side, there also exist many ways to estimate structural breaks and important contributions can be found in \cite{Stephens1994}, \cite{Chib98}, \cite{Fearnhead2007}, \cite{rigaill2012exact} and \cite{Maheu13}. While all these methods differ in the criterion or in the algorithm used to detect the changes, most of them rely on the assumption that, when a break is detected (that may be triggered by the change in only one model parameter), a new segment is created and a new set of parameters needs to be estimated. Although the assumption seems harmless, it creates two important drawbacks:
\begin{enumerate}
	\item From an interpretation perspective, if all the parameters have to change when a break is detected, it is difficult to assess which parameters have indeed abruptly varied and so it complicates the economic interpretation of the structural break. 
	\item Forecasting wise, when a parameter does not vary from one regime to another, its estimation is more accurate than if two parameters were considered over these two regimes. This feature could improve the predictions of the model.
\end{enumerate}

In this paper, we propose a method to relax the assumption that a break triggers a change in all the model parameters. To do so, we first estimate the potential break dates exhibited by the series and then we use a penalized likelihood approach to detect which parameters change. Because some segments in the CP regression can be small, we opt for a (nearly) unbiased penalty function, called the seamless-L0 (SELO) penalty function, recently proposed by \cite{dicker2013variable}. We prove the consistency of the SELO estimator in detecting which parameters indeed vary over time and we suggest using a deterministic annealing expectation-maximisation (DAEM) algorithm to deal with the multimodality of the objective function \citep[see][]{ueda1998deterministic}. Since the SELO penalty function depends on two tuning parameters, we use a criterion (new in this literature) to choose the best tuning parameters and as a result the best model. This new criterion exhibits a Bayesian interpretation which makes possible to assess the parameters' uncertainty as well as the model's uncertainty. This last feature is determinant when predicting a time series since the Bayesian model averaging technique, that typically improves forecast accuracy, is readily applicable \citep[see, e.g.,][]{raftery2010online,koop2012forecasting}. 

We are aware of five other papers that also relax the assumption on the number of parameters that changes when a break is detected. In the frequentist literature, the influential paper of \cite{BaiPerron1998} proposes a method that also operates when only a subset of parameters can break. However, the number of possibilities grows exponentially with the number of breaks as well as with the number of parameters that can break. From a Bayesian perspective, \cite{GiordaniKohn2008}, \cite{Eo2012}, \cite{huber2019should} and \cite{dufays2019relevant} propose flexible state-space models to capture which parameters vary over time. However, all these estimation procedures break down when the number of parameters is large (see Supplementary Appendix \ref{App:Bayesian} for more details).

We believe that our method exhibits several advantages over the existing alternatives. Firstly, it operates for small and large dimensions. Secondly, the estimation is fast compared to the Bayesian alternatives. 
As a final advantage, we relax the assumption on breakpoints \textit{once} the structural breaks have been detected which makes our approach operating in combination with any existing CP methods. In this paper, we illustrate our approach with the CP procedure of \cite{yau2016inference} but any other CP method could have been used. 

A final reference close to our framework is \cite{chan2014group} who propose a penalized regression for segmenting time series in piecewise linear models. The paper uses a group Lasso penalty function \citep[see][]{yuan2006model} to get an overestimated number of segments and in a second phase, an information criterion is used to improve the estimation. 
Nevertheless, we differ from their methods in many aspects. First, we use an almost unbiased penalty function and from a theoretical perspective, as we use the penalized regression on a potential break date set, our assumptions for a consistent estimator are different and in line with the standard penalized regression literature. We also use a Bayesian criterion to select among the promising models uncovered by the penalty function which allows for model uncertainty and for Bayesian model averaging. Also, our estimation procedure is fast compared to \cite{chan2014group} since we iterate on closed-form expressions and because our model exhibits fewer parameters. As a final difference, we provide break uncertainty.

Eventually, we apply our method on Hedge funds (HF) returns. As highlighted  by  \cite{Fung2008hedgefunds}, by \cite{meligkotsidou2008detecting}, by \cite{Bollen2009hedgefundriskdynamics}, and more recently by \cite{Patton2015change}, the dynamics of HF risk exposures and the nonlinear generating process of HF returns should be associated with market events and structural breaks. For a sample of 14 monthly Credit Suisse HF indices spanning from March 1994 to March 2016, and using the asset based style pricing model introduced by \cite{fung2001risk}, we show that our modeling is particularly appealing to detect time-varying exposures in HF tradings. In particular, our results report the relative role played by static and dynamic parameters and factors in the decomposition of HF returns. We also investigate the prediction performance of our approach and it turns out that the selective segmentation approach compares favorably in terms of root mean squared forecast errors and cumulative log-predictive densities with respect to other CP processes. In particular, it almost systematically dominates the CP model which assumes that all the parameters vary when a break is detected.

The paper is organized as follows. Section \ref{sec:spec} documents the model specification and the SELO penalty function. Section \ref{sec:estimation} explains how the DAEM algorithm is applied to our framework. In Section \ref{sec:modelSelection}, we detail the criterion used to select the SELO tuning parameters and we relate it to the Bayesian paradigm. Section \ref{sec:break} documents the CP method of \cite{yau2016inference} and discusses how it can be slightly improved. An extensive Monte Carlo study is proposed in Section \ref{sec:MonteCarlo}. We end the paper by applying the method on HF returns in Section \ref{s:empirics}. All the proofs are given in the Supplementary Appendix (SA).

\section{Model specification \label{sec:spec}}
We consider a standard linear regression specified as
\begin{align} \label{eq:linreg}
\begin{split}
y_t  & = \beta_1 + \beta_2 x_{t,2} + \ldots + \beta_{K} x_{t,K} + \epsilon_t,\\
 & =  \bx_t'\bbeta_1 + \epsilon_t \,,
\end{split}
\end{align}
where $  \epsilon_t \sim MDS(0,\sigma^2)$ (in which $MDS$ stands for the martingale difference sequence), $\mathbf x_t = (1, x_{t,2}, \ldots,  x_{t,K})'$ and $\bbeta_1 = (\beta_1, \beta_2, \ldots, \beta_{K})'$. Typically, if a linear model is estimated over a long period, the parameters are subject to abrupt changes over time. To take this time-varying dynamic into account, we allow for $m-1$ structural breaks in the model parameters as follows,
\begin{align} \label{eq:linreg_break}
\begin{split}
y_t  & =  \bx_t'\bbeta_i^* + \epsilon_t \,, \text{ for } \tau_{i-1}< t \leq \tau_{i}, 
\end{split}
\end{align}
in which $\bbeta_i^*$, is the true parameter of the explanatory variables over the regime $i$,  $\btau_0 = \{\tau_0,\ldots,\tau_m\} \in \mathbb{N}^{m+1}$ where $\tau_0=0$, $\tau_m = T$ and $\tau_i < \tau_{i+1}$ $\forall i \in [0,m-1]$. In this paper, we are interested in capturing which parameters are subject to breaks and which do not vary over time. To do so, we reframe the model \eqref{eq:linreg_break} as follows,
\begin{align} \label{eq:linreg_break2}
\begin{split}
y_t  & = \bx_t'\bbeta_1^* +  \bx_t'( \sum_{j=2}^{m} \Delta \bbeta_j^* \1{t> \tau_{j-1}}) + \epsilon_t \,,\\
\by & =  \bX_{\btau} \bbeta^* + \bepsilon,
\end{split}
\end{align}
where $\1{x>a}=1$ if $x>a$ and zero otherwise, $\Delta \bbeta_j^* = \bbeta_{j}^* - \bbeta_{j-1}^*$, for $j\in [2,m]$, stands for the model parameters in first-difference, $\by = (y_1,\ldots,y_T)'$, $\bX_{\btau} = (\tilde{\bX}_{\tau_0}, \tilde{\bX}_{\tau_1},\ldots,\tilde{\bX}_{\tau_{m-1}})$ with $\tilde{\bX}_{\tau_i} = (\mathbf 0,\mathbf 0,\ldots,\mathbf 0,\bx_{\tau_i+1},\ldots,\bx_{T})'$,  $\bepsilon = (\epsilon_1,\epsilon_2,\ldots,\epsilon_T)'$ and $\bbeta^*=(\bbeta_1^{*\prime},\Delta \bbeta_2^{*\prime},\ldots,\Delta \bbeta^{*\prime}_m)^{\prime} \in \Re^{Km\times 1}$. Note that the matrix $\tilde{\bX}_{\tau_0}$ stands for the standard regressors since $\tau_0=0$. Regarding the notations, the first-difference parameter in regime $j$ is a K-dimensional vector $\Delta \bbeta_{j}^*$ such that $\Delta \bbeta_{j}^* =(\Delta \beta_{j1}^*,\dots, \Delta \beta_{jK}^*)^{\prime}$. Let us also denote $A = \{(j,k); \Delta \beta_{jk}^*\neq 0, \text{ for } j\in[2,m] \text{ and for } k\in[1,K]\}$, the set of indices defining the true model. 

Our strategy to uncover which parameters truly vary over time consists in first finding where are the potential break dates $\btau$, then, in a second phase, in detecting which parameters evolve. Note that even when we know the true break dates $\btau$, the problem of finding which parameters vary when a break occurs is not straightforward as the number of models to consider amounts to $2^{(m-1)K}$. Consequently, it is infeasible to carry out an exhaustive model selection when $K$ or $m$ is large. We propose a penalized likelihood approach to explore this large model space and to select which parameters experience breaks. To focus on our selective segmentation approach, we shall first assume that we have obtained a set of potential break dates $\btau$. We discuss how we estimate this set in Section \ref{sec:break}.
\begin{remark}\label{rem:exact}
In the situation where all the models can be considered (i.e., $(m-1)K \leq 10$), we do not need to rely on the penalized likelihood approach explained in Section \ref{sec:penalty}. In particular, we could directly estimate all the model combinations and select the best one according to the marginal likelihood criterion given in Section \ref{sec:modelSelection}. 
\end{remark}

\subsection{Penalized likelihood and choice of the penalty function \label{sec:penalty}}
 As emphasized by Equation \eqref{eq:linreg_break2}, given a set of break dates $\btau$, the problem of finding which parameters abruptly change when a break occurs boils down to a penalized linear regression problem. Specifically, one can solve the following optimization problem
\begin{align} \label{eq:optim}
\begin{split}
\hat{\bbeta} & = \argmin_{\bbeta} ||\by- \bX_{\btau} \bbeta||_2^2 + T \sum_{j=2}^{m} \sum_{k=1}^{K}\text{pen}(\Delta \beta_{jk}),
\end{split}
\end{align}

\noindent where $||~ . ~||_p$ denotes the $L_p$ norm and $\text{pen}(\Delta \beta_{jk})$ stands for a penalty function. Popular choices of $\text{pen}(\Delta \beta_{jk})$ are the Lasso penalty function (i.e., $\text{pen}(\Delta \beta_{jk}) = \lambda ||\Delta \bbeta_{jk}||_1$, see \cite{Tibshirani94}) or the rigde function (i.e., $\text{pen}(\Delta \beta_{jk}) = \lambda \Delta ||\beta_{jk}||_2^2$, see, for instance, \cite{Ishwaran2005}).

Following \cite{Fan2001}, standard desirable properties induced by a penalty function are i) unbiasdness, ii) sparsity and iii) continuity. For instance, the ridge function is only continuous while the Lasso penalty function achieves sparsity and continuity (beside at zero). However one standard issue with these popular penalty functions is that they provide biased (but typically consistent) estimators. In our framework, this drawback is problematic since a segment can sometimes contain a small amount of observations that makes consistency results not sufficient. Recently, \cite{dicker2013variable} propose a penalty function, called seamless-$L_0$ (SELO), that exhibits all the desirable properties. For a model parameter denoted $\omega$, the penalty function reads as
\begin{eqnarray*}
\pen(\omega|\zeta,\lambda) & = & \frac{\lambda}{\ln 2}  \ln (\frac{2|\omega| + \zeta}{|\omega| + \zeta}),
\end{eqnarray*}
where the parameter $\zeta$ controls for the concavity of the function and $\lambda$ stands for the penalty imposed when $\omega \neq 0$. We slightly modify their function to end up with parameters that are directly interpretable. In fact, we use the following penalty function,
\begin{eqnarray}\label{eq:SELO}
\pen(\omega|a,\lambda) & = & \frac{\lambda}{\ln 2} \ln (\frac{2(\frac{|\omega|}{a}) + \zeta}{(\frac{|\omega|}{a}) + \zeta}),
\end{eqnarray}
where $\zeta = \frac{2^y-2}{1-2^y}$ with $y \in (0,1)$ and we set $y=0.99$. In most cases, the parameter $a$ can be interpreted as an interval $\omega \in [-a,a]$ in which $\omega$ will be biased with respect to the OLS estimate since $\pen(a) = \lambda y$. Intuitively, when $|\omega|>a$, we have $\pen(\omega) \approx \lambda$ and $\frac{d\pen(\omega)}{d\omega}|_{|\omega|\geq a} \approx 0$ for large values of $a$. Figure \ref{fig:SELO} shows the SELO penalty function with $\{a,\lambda\}=\{1,0.9\}$ and illustrates that the function is almost flat for absolute values greater than $a$. To be more precise about how large $a$ must be, when $|\bar{\beta}|\geq a$ with $a\geq \frac{\zeta}{\ln 2[\zeta^2 + 3\zeta + 2]} = 0.0099$, we have that $\frac{d\pen(\omega)}{d\omega}|_{|\omega|\geq a} \leq \lambda$  which implies that the bias imposed by the SELO penalty function is smaller than the one of the Lasso function (i.e. $\lambda |\omega|)$.

\renewcommand{\baselinestretch}{1}
\begin{figure}[h!]
\centering
\includegraphics[scale=.3]{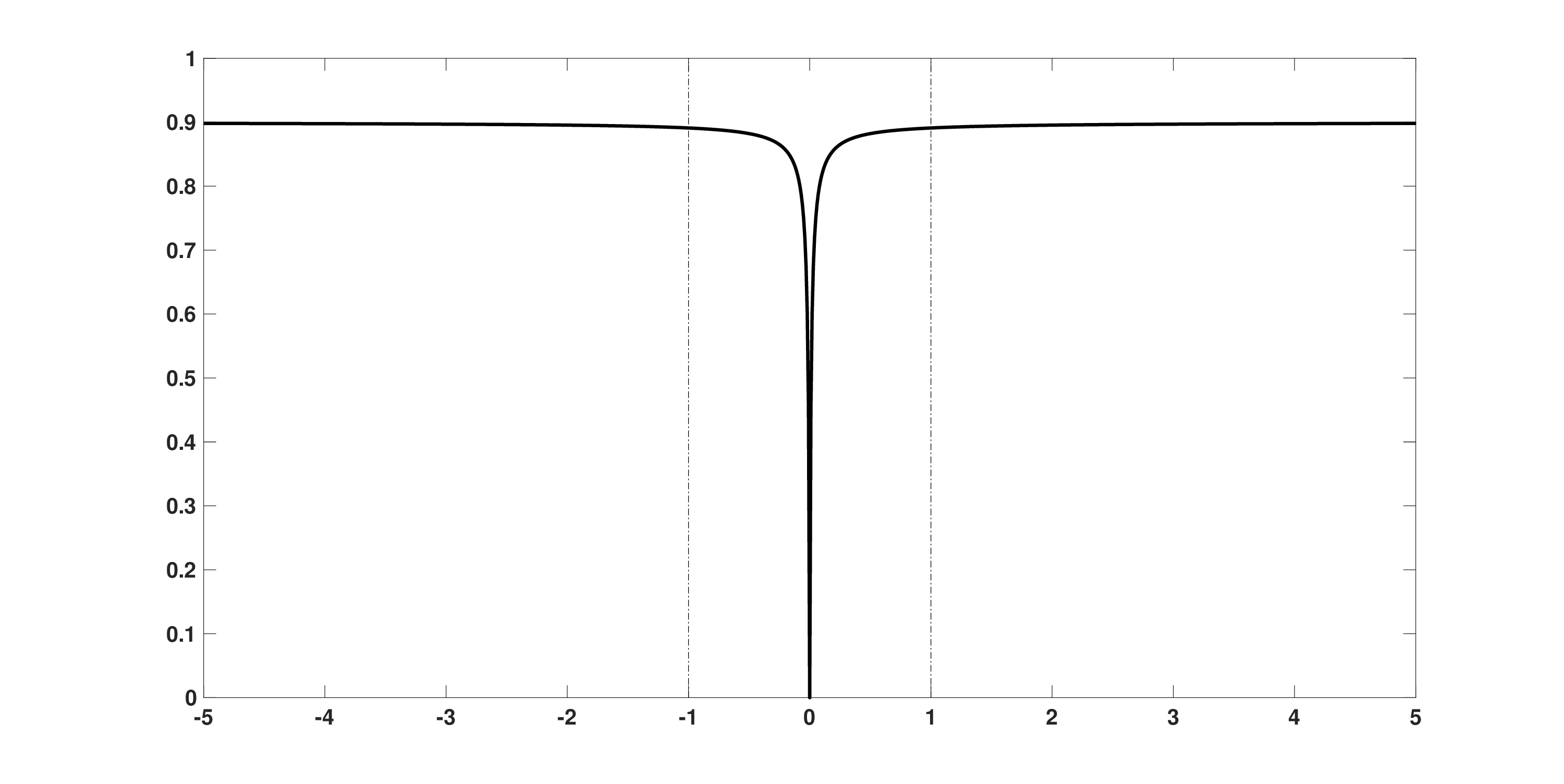}
\caption{\justifying \textbf{SELO penalty function.} Penalty function is shown in solid black lines while vertical dotted lines highlight the interval [-a,a]. The SELO parameters are set to $\lambda = 0.9$ and $a=1$.  \label{fig:SELO}}
\end{figure}
\renewcommand{\baselinestretch}{1.5}

\subsection{Consistency of the SELO estimator}
 The interval $[-a,a]$ in which a parameter is biased is likely to change with the variable to which it refers. Furthermore, if we assume that this interval is fixed over time, we should set a new parameter $a$ for each variable on the $m$ segments. Unlike \cite{dicker2013variable} who define a single parameter for all the variables, we use $K$ parameters $a_1,\dots,a_K$, that is, one per explanatory variable. Thus, the objective function to minimize is given by
\begin{align} 
\begin{split}
f(\bbeta) & = ||\by- \bX_{\btau} \bbeta||_2^2 + T \frac{\lambda}{\ln(2)}  \sum_{j=2}^{m} \sum_{k=1}^{K} \ln \left(\frac{2(\frac{|\Delta \beta_{jk}|}{a_k}) + \zeta}{(\frac{|\Delta \beta_{jk}|}{a_k}) + \zeta}\right).
\end{split}
\end{align} 

\noindent Before discussing how to maximize the objective function, we present the main results about the modified SELO estimator. As highlighted in \cite{dicker2013variable}, the SELO estimator is consistent under reasonable conditions. Proposition \ref{theoconsislimitdist} shows that this consistency result also applies in our framework. To do so, we consider the following assumptions (in which a sequence $\omega_T \rightarrow \omega $ is understood as $\text{lim}_{T \rightarrow  \infty} \omega_T = \omega$):
\begin{assumption}
	\label{tauistruetau}
  $\btau = \btau_0$ and $\forall j\in [1,m]$, we have $\tau_j - \tau_{j-1} = T\delta_{\tau_j} \rightarrow \infty$, with $\displaystyle \sum_{j=1}^m \delta_{\tau_j} =1$.
\end{assumption} 
\begin{assumption}
	\label{rhopositive}
	$\displaystyle\rho\sqrt{T} \rightarrow \infty$, where $\rho = \min_{r,k \in A}\left(|\Delta\beta^*_{r,k}|\right)$. 
\end{assumption} 
\begin{assumption}
	\label{Xfullrank}
	There exist $r_0$, $R_0 > 0$ such that $r_0 \leq \lambda_{T,\min} < \lambda_{T,\max} \leq R_0$, where $\lambda_{T,\min}$ and $\lambda_{T,\max}$ are the smallest and  largest eigenvalues of $\left(T^{-1}\bX_{\btau}'\bX_{\btau}\right)$ respectively. 
\end{assumption} 
\begin{assumption}
	\label{exogeneity}
	The process $\{\epsilon_t,\bx_t\}_{t \in (\tau_{j-1},\tau_j]}$ is ergodic and stationary for any $j=1,\dots,m$. Moreover, $\forall ~ t \in [1,T]$, $\mathbb{E}(\epsilon_t|\mathbf x_t) = 0$, $\mathbb{E}(\epsilon_t^2) = \sigma^2$ and the process $\{g_t\} = \{\bx_t \epsilon_t\}$ is a martingale difference sequence with finite second moments.
\end{assumption}
\begin{assumption}
	\label{tauisO}
	$\lambda = \mathcal{O}_p(1)$ and $a_{k}=\mathcal{O}_p\left(T^{-\frac{3}{2}}\right)$, $\forall k \in [1,K]$. 
\end{assumption}

We first discuss the assumptions before detailing our consistency result. Assumptions \ref{tauistruetau} to \ref{tauisO} are similar to those found in the variable selection literature \cite[see][]{fan2004nonconcave,dicker2013variable} and in the CP literature \cite[see][]{BaiPerron1998,yau2016inference}. Condition \ref{tauistruetau} assumes that the estimated CPs are the true locations. However, the SELO estimator maintains the same asymptotic properties with a set of potential breakpoints as long as it contains the true break dates (see the adapted assumption \ref{tauistruetau2} below). In such case, Proposition \ref{theoconsislimitdist} also ensures that the number of breakpoints is consistently estimated. Note that condition \ref{tauistruetau} implies that the length of each segment increases linearly with $T$. Although unattractive, this condition is generally made in the CP literature \citep[see, e.g.,][]{perron2006dealing,yau2016inference}. For interested readers, \cite{perron2006dealing} motivate this assumption in details. Assumption \ref{rhopositive} allows the minimum break size to decrease with the sample size but at a slower rate than $T^{-\frac{1}{2}}$. Conditions \ref{Xfullrank} are related to the eigenvalues and are standard in the variable selection literature \citep[see, e.g.,][]{zhang2010regularization}. However, we show in SA \ref{App:eig} that this condition is not innocuous and that it implies a fixed number of regimes as well as $\min_j\{\delta_{\tau_j}\}>0$; that is $\delta_{\tau_j}$ does not drift to 0 as $T \rightarrow \infty$. Avoiding this assumption would imply stronger conditions on the process $\{y_t,\bx_t\}$ \citep[see, e.g.,][]{chan2014group}. The assumption \ref{exogeneity} refers to ergodicity and stationarity of each segment and imposes the standard exogeneity hypothesis. This assumption ensures that sampled counterparts of the first two moments of $\{\bx_t\epsilon_t\}$  are converging to finite values. Importantly, it does not rule out conditional heteroskedasticity. Eventually, condition \ref{tauisO}  defines restrictions on the tuning parameters rate. The same condition applies in \cite{dicker2013variable}. The consistency of SELO estimator is given by the following Proposition.
 
\begin{theo}
	\label{theoconsislimitdist}
	Assume that \ref{tauistruetau}-\ref{tauisO} hold and let, 
	\begin{align}
\label{ojecfun}
f_{T}(\bbeta) = \dfrac{1}{T}||\by- \bX_{\btau} \bbeta||_2^2 + \sum_{j=2}^m\sum_{k=1}^{K}\pen(\Delta\beta_{jk}|a_k,\lambda). 
\end{align}
There exists a sequence of $\sqrt{T}$-consistent local minima  $\hat{\bbeta}$ of $f_{T}(\bbeta)$ as defined by Equation \eqref{ojecfun} such that:
	\begin{enumerate}[(i)]
		\item $\displaystyle \lim_{T\to\infty}  \mathbf{P}\left(\left\{ (j,k); \hat{\beta}_{jk} \ne 0 \right\}= A \right) = 1$
		\label{theoconsistency}
		\item $\forall~ \delta > 0$,  $\displaystyle \lim_{T\to\infty} \mathbf{P}\left(||\hat{\bbeta}_{A} - \bbeta_{A}^*|| > \delta \right) = 0$ 
		\label{theolimitdist}
	\end{enumerate}
\end{theo}

\begin{proof}
The proof is given in SA \ref{App:theo1}
\end{proof}

\begin{remark} \label{rem1}
Proposition \ref{theoconsislimitdist} also applies when the set of breakpoints contains additional spurious break dates. In particular, Proposition \ref{theoconsislimitdist} holds if we relax assumption \ref{tauistruetau} by the less restrictive assumption:
\begin{assumption}
	\label{tauistruetau2}
	$\btau = \{\tau_1,\ldots,\tau_{\hat{m}}\}$ with $\hat{m}\geq m$, $\btau_0\subseteq \btau$ and $\forall j\in [1,\hat{m}]$, we have $\tau_j - \tau_{j-1} = T\delta_{\tau_j} \rightarrow \infty$, with $\displaystyle \sum_{j=1}^{\hat{m}} \delta_{\tau_j} =1$. 
\end{assumption}
\end{remark}

\section{Estimation \label{sec:estimation}}
The objective function to minimize is given by
\begin{align} \label{eq:optimSELO}
\begin{split}
f(\bbeta) & = ||\by- \bX_{\btau} \bbeta||_2^2 + T \frac{\lambda}{\ln(2)}  \sum_{j=2}^{m} \sum_{k=1}^{K} \ln \left(\frac{2(\frac{|\Delta \beta_{jk}|}{a_k}) + \zeta}{(\frac{|\Delta \beta_{jk}|}{a_k}) + \zeta}\right),\\
& = ||\by- \bX_{\btau} \bbeta||_2^2 + \sum_{j=2}^{m} \sum_{k=1}^{K} \ln q_{k}(\Delta \beta_{jk}),
\end{split}
\end{align}
in which $q_{k}(\Delta \beta_{jk}) = \left(\frac{2(\frac{|\Delta \beta_{jk}|}{a_k}) + \zeta}{(\frac{|\Delta \beta_{jk}|}{a_k}) + \zeta}\right)^{\left(\frac{T \lambda}{\ln(2)}\right)}$. Due to the penalty function, we cannot find any analytical expression of the minimizer. In addition to that, the function likely exhibits many local modes which complicates the optimization. We address the problem of finding the global mode by using a deterministic annealing expectation-minimization (DAEM) algorithm \cite[see][]{ueda1998deterministic}. To do so, we first approximate the penalty function by a mixture of two Normal components (to take into account the large tail of the SELO penalty function), the details of it are given in SA \ref{App:mixture}. Secondly, since minimizing the sum of squared residuals is identical to maximizing a likelihood function when the error term is normally distributed, we work with the following model, $\by =  \bX_{\btau} \bbeta + \eta$, where $\eta \sim \NORM(\mathbf 0, \sigma^2 I_T)$. The modified model implied the following objective function to maximize with respect to $(\bbeta,\sigma^2)$:
\begin{align} \label{eq:optimEM}
\begin{split}
f(\by|\bbeta,\sigma^2) & = -\frac{T}{2}\ln \sigma^2 - \frac{1}{2\sigma^2}||\by- \bX_{\btau} \bbeta||_2^2 - \sum_{j=2}^{m} \sum_{k=1}^{K} \ln g_k(\Delta \beta_{jk}),
\end{split}
\end{align}
where $g_k(\Delta \beta_{jk}) = \sum_{i=1}^{2}\omega_{i}^{(k)} f_N(\Delta \beta_{jk}|\mu_{i}^{(k)},s_{i}^{(k)})$, $f_N(x|\mu,s)$ stands for the normal density function evaluated at $x$ with expectation and variance given by $\mu$ and $s$ respectively and $\omega_{i}^{(k)} \in (0,1)$ with $\sum_{i=1}^{2}\omega_{i}^{(k)} = 1$. Note that the function $f(\by|\bbeta,\sigma^2)$ in Equation \eqref{eq:optimEM} is proportional to the posterior log-density of the parameter distribution $\bbeta,\sigma^2|\by$ from a Bayesian perspective with prior distributions given by $f(\sigma^2,\bbeta_1) \propto 1$ and $f(\Delta \beta_{jk}) = g_k(\Delta \beta_{jk})$ for $j\in[2,m]$ and $k\in [1,K]$. In particular, the distribution $f(\Delta \beta_{jk})$ can be understood as a spike and slab prior \citep[e.g.][]{George_1993} and our optimization procedure fits into the framework of \cite{rovckova2014emvs} which proposes tackling the linear variable selection problem with the EM algorithm and its DAEM variant. The optimization is therefore equivalent to finding the mode of $\bbeta,\sigma^2|\by$. Using a data augmentation approach, we add latent variables $\bz = (z_{21},z_{22},\ldots,z_{mK})'$ such that $f(z_{jk}=i) = \omega_i^{(k)}$, $\forall j \in [2,m]$, $\forall k \in [1,K]$ and  $\forall i \in [1,2]$. With these latent variables, we can write the prior distribution of $\Delta \beta_{jk}$ in a convenient hierarchical way as follows,
\begin{eqnarray*}
f(\Delta \beta_{jk}|z_{jk}=i) & = & f_N(\Delta \beta_{jk}|\mu_i^{(k)},s_i^{(k)}), \text{ and }  f(z_{jk}=i) = \omega_i^{(k)}.
\end{eqnarray*}
By fixing $\btheta = \{\bbeta,\sigma^2\}$, the EM algorithm (and its DAEM variant) solves the following optimization at iteration $n$,
\begin{eqnarray*}
\text{argmax}_{\btheta_n } Q(\btheta_{n}|\btheta_{n-1}) & = & \text{argmax}_{\btheta_n} \mathbb{E}_{\bz|\by,\btheta_{n-1}}(\ln f(\btheta_n,\bz|\by) |\by,\btheta_{n-1}).
\end{eqnarray*}
One can easily show that maximizing $Q(\btheta_{n}|\btheta_{n-1})$ implies that $f(\btheta_n|\by) \geq f(\btheta_{n-1}|\by)$.

\subsection{Derivation of the DAEM algorithm \label{sec:DAEM}}
To apply the DAEM algorithm, we need to find an expression of $Q(\btheta|\btheta_{n-1})$. Given a set of parameter $\btheta_{n-1}$, we have that 
\begin{eqnarray*}
Q(\btheta|\btheta_{n-1}) & = & \mathbb{E}_{\bz|\by,\btheta_{n-1}}(\ln f(\btheta|\by,\bz)f(\bz|\by) |\by,\btheta_{n-1})\\
 & \propto & \ln f(\by|\bbeta,\sigma^2) + \ln f(\bbeta_1,\sigma^2) + \sum_{j=2}^{m} \sum_{k=1}^{K} \sum_{i=1}^{2} -\frac{(\Delta \beta_{kj}-\mu_{i}^{(k)})^2}{2s_{i}^{(k)}} f(z_{kj}=i|\by,\btheta_{n-1}),\\
 & \propto & \ln f(\by|\bbeta,\sigma^2)  - \frac{1}{2}\sum_{i=1}^{2} (\bbeta-\bmu_i)'\bSigma_i(\bbeta-\bmu_i),
\end{eqnarray*}
where 
\begin{eqnarray*}
\bmu_i &=& (\underbrace{0,0,\ldots,0}_{\text{K-dimensional}},\mu_i^{(1)},\mu_i^{(2)},\ldots,\mu_i^{(K)},\mu_i^{(1)},\ldots)' \in \Re^{mK \times 1},\\
\bSigma_i &= &\text{diag}(\underbrace{0,0,\ldots,0}_{\text{K-dimensional}},\frac{p_{21}^{(i)}}{s_i^{(1)}},\frac{p_{22}^{(i)}}{s_i^{(2)}},\ldots,\frac{p_{2K}^{(i)}}{s_i^{(K)}},\frac{p_{31}^{(i)}}{s_i^{(1)}},\ldots,\frac{p_{mK}^{(i)}}{s_i^{(K)}}),
\end{eqnarray*}
with $p_{jk}^{(i)} = f(z_{jk}=i|\by,\btheta_{n-1})$ $\forall i \in [1,2]$,$\forall j \in [2,m]$ and $\forall k \in [1,K]$  . Importantly, the difference between the EM algorithm and its DA version only appears in the quantities $p_{jk}^{(i)}$. In fact, the DAEM algorithm introduces an increasing function $\phi(r):[1,N]\rightarrow (0,1]$ such that $0<\phi(1) \leq 1$ and $\phi(N) = 1$. For each value $r=1, \ldots, N$, it applies recursively the EM algorithm (that starts with the final estimate of the previous EM algorithm) where the posterior probabilities $p_{jk}^{(i)}$ are denoted $p_{jk}^{(i,\phi(r))}$ and are modified as follows,
\begin{eqnarray} \label{eq:postprob_p}
p_{jk}^{(i,\phi(r))} & \propto & (f_N(\Delta \beta_{jk}|\mu_i^{(k)},s_i^{(k)})\omega_i^{(k)})^{\phi(r)}.
\end{eqnarray}
When $r=N$, the increasing function $\phi(r)=1$ and the standard EM algorithm is run (but with a promising starting point). 
To find the maximum of $Q(\btheta|\btheta_{n-1})$, we sequentially maximize $\bbeta$ given $\sigma^2$ and then $\sigma^2$ with respect to $\bbeta$. This approach, called coordinate iterative ascent, operates in two steps:
\begin{enumerate}
	\item Compute $\bbeta_{n} = \argmax_{\bbeta} Q(\bbeta,\sigma_{n-1}^2|\btheta_{n-1})$. 
	\item Compute $\sigma_{n}^2 = \argmax_{\sigma^2} Q(\bbeta_{n},\sigma^2|\btheta_{n-1})$. 
\end{enumerate}
At the end of the two steps, we necessarily have $Q(\bbeta_{n-1},\sigma_{n-1}^2|\btheta_{n-1}) \leq Q(\bbeta_{n},\sigma_{n-1}^2|\btheta_{n-1}) \leq Q(\bbeta_{n},\sigma_{n}^2|\btheta_{n-1})$. The maximisation of $\bbeta$ given $\sigma_{n-1}^2$ leads to
	\begin{eqnarray*}
\bbeta_n  & = & [\sigma_{n-1}^{-2} X_{\btau}'X_{\btau} + \sum_{i=1}^{2} \bSigma_i]^{-1}[\sigma_{n-1}^{-2}X_{\btau}'\by  + \sum_{i=1}^{2}\bSigma_i \bmu_i].
\end{eqnarray*}

The update of $\sigma^2$ conditional to $\bbeta_n$ is given by
	\begin{eqnarray*}
\sigma_{n}^2  & = & \frac{[(\by-X_{\tau}\bbeta_n)'(\by-X_{\tau}\bbeta_n)}{T}.
\end{eqnarray*}

We summarize the DAEM procedure in Algorithm \ref{algo:DAEM}. In practice, the minimum distance $e$ indicating a convergence of the algorithm is set to $10^{-5}$ and the number of DAEM iteration $N$ is fixed to 10.

\begin{algorithm}[ht]
	\caption{\justifying DAEM algorithm}\label{algo:DAEM}
	\begin{algorithmic}
		\State \text{Initialize $\bbeta_0$ using Algorithm \ref{algo:init}}
		\State \text{Set $\sigma_0^2=\frac{[(\by-\bX_{\tau}\bbeta_0)'(\by-\bX_{\tau}\bbeta_0)}{T}$, $\phi(1)=(\frac{1}{N})^2$, $r=1$ and $\text{dist}=\infty$.}
		\While{$r<= N$} 
		\State \text{Set $n=0$ and $\btheta_n = (\bbeta_0',\sigma_0^2)'$.}
		\While{$\text{dist} > e$} 
		\State \text{Increment $n= n+1$.}
		\State \text{Compute the posterior probabilities $p_{jk}^{(i,\phi(r))}$ given in Equation \eqref{eq:postprob_p} for i=1,2}
		\State \text{Compute the mean parameters }
		$$\bbeta_n=[\sigma_{n-1}^{-2} \bX_{\btau}'\bX_{\btau} + \sum_{i=1}^{2} \bSigma_i]^{-1}[\sigma_{n-1}^{-2}\bX_{\btau}'\by  + \sum_{i=1}^{2}\bSigma_i \bmu_i].$$
		\State \text{Compute the variance parameter }
		$$\sigma_{n}^2 = \frac{[(\by-\bX_{\tau}\bbeta_n)'(\by-\bX_{\tau}\bbeta_n)}{T}.$$
		\State \text{Set $\btheta_n = (\bbeta_n',\sigma_n^2)'$ and compute the distance value $\text{dist} = || \btheta_n-\btheta_{n-1}||_2$.}
		\EndWhile 
		\State \text{Increment $r = r+1$ and set $\phi(r) = (\frac{r}{N})^2$. Set $\bbeta_0 = \bbeta_n$ and $\sigma_0^2 = \sigma_n^2$.}
		\EndWhile 
	\end{algorithmic}
\end{algorithm}

The EM and the DAEM algorithms are sensitive to starting values. Inspired by \cite{zhao2012random}, we mitigate this issue by randomly exploring the model space using a swapping approach before applying the DAEM algorithm. To be specific, we generate $N_{\text{init}}$ values as explained in Algorithm \ref{algo:init} and we initialize the DAEM algorithm with the parameter estimates that minimize the penalized function given in Equation \eqref{eq:optimSELO}. In practice, we set $N_{\text{init}} = \text{min}(2^{(m-1)K-1},3000)$.
\begin{algorithm}[ht]
	\caption{\justifying Initialization of the DAEM algorithm}\label{algo:init}
	\begin{algorithmic}
	\For{$n=1$ to $N_{\text{init}}$} 
  \State \text{Set $\hat{A} = \emptyset$ and sample $p\sim U[0,1]$.}
	\State \text{For $j=2,...,m$ and for $k=1,...,K$, do $\hat{A} = \hat{A} \cup (j,k)$ with probability $p$
	.}
	\State \text{$(f_n,\bbeta_n) =$ Swap($\hat{A}$) (see Algorithm \ref{algo:swap})}
	\EndFor
	\State \text{Return the OLS estimates $\bbeta_{\hat{n}}$ such that $\hat{n} = \argmin_{n\in [1,N_{\text{init}}]} f_n$.}
	\end{algorithmic}
\end{algorithm}

\begin{algorithm}[ht]
	\caption{\justifying Swap the set of indices - Swap($\hat{A}$)}\label{algo:swap}
	\footnotesize
	\begin{algorithmic}
	\State \text{Given a set of indices $\hat{A}$ defining the parameters $\Delta \bbeta \neq 0$, for $j=2,...,m$ and for $k=1,...,K$ \textbf{do}}
	\State \text{Build the sets $\tilde{A}_{jk} = \hat{A}\cup (j,k)$ if $\hat{A}\cap (j,k) = \emptyset$ or the set $\tilde{A}_{jk} = \hat{A}\backslash (j,k)$ otherwise.}
	\State \text{For each set $\tilde{A}_{jk}$, compute the OLS estimates ($\hat{\bbeta}_{jk}$) and the penalized function $f_{jk} = f(\hat{\bbeta}_{jk})$ (see \eqref{eq:optimSELO}).}
	\State \text{For the set $\hat{A}$, compute the OLS estimates ($\hat{\bbeta}_{\hat{A}}$) and the penalized function $f_{\hat{A}} = f(\hat{\bbeta}_{\hat{A}})$ (see \eqref{eq:optimSELO}).}
	\State \text{Find $(\hat{j},\hat{k}) = \argmin_{j,k}f_{jk}$}
	\If{$f_{\hat{j}\hat{k}}<f_{\hat{A}}$}
	\State \text{Return $\hat{\bbeta}_{\hat{j}\hat{k}}$ and $f_{\hat{j}\hat{k}}$}
	\Else
	\State \text{Return $\hat{\bbeta}_{\hat{A}}$ and $f_{\hat{A}}$}
	\EndIf
	\end{algorithmic}
\end{algorithm}

\section{Selection of the penalty parameters and parameter uncertainties \label{sec:modelSelection}}
The SELO penalty function exhibits two tuning parameters $\boldsymbol{a}$ and $\lambda$. The standard approach to fix them consists in considering a grid of values of these parameters and in selecting the parameters that maximize a (generally consistent) information criterion \citep[e.g.,][]{zhang2010regularization}. Instead of relying on a standard information criterion and select the tuning parameters $\boldsymbol{a}$ and $\lambda$ that maximize it, we consider each pair $(\boldsymbol{a},\lambda)$ as a model to take into account the model uncertainty. For a given value of $(\boldsymbol{a},\lambda)$, the DAEM algorithm exposed in Section \ref{sec:DAEM} provides an estimate $\hat{\Delta \bbeta}$ of $\Delta \bbeta$ which delivers an estimate of $\hat{A}$, i.e., the set of indices with $\hat{\Delta \beta_{jk}} \neq 0$ for $ j\in [2,m]$ and for $k \in [1,K]$. This set tells us which covariates should be included in the linear regression and which should not. Let us denote by $\tilde{\bX}_{\btau}^{\hat{A}}$ the covariates related to the first-difference estimates that are different from zero. We use the following criterion for selecting $\boldsymbol{a}$ and $\lambda$:
\begin{equation} \label{eq:ML}
\begin{split}
f(\by|\ba,\lambda,\btau) & = (\frac{g_{\hat{A}}}{1+g_{\hat{A}}})^{k_{\hat{A}}/2} [\frac{g_{\hat{A}}}{1+g_{\hat{A}}}s_{\tilde{\bX}_{\tau_0}}+\frac{1}{(1+g_{\hat{A}})}s_{\tilde{\bX}_{\tau_0},\tilde{\bX}_{\btau}^{\hat{A}}}]^{-\frac{T-K}{2}},
\end{split}
\end{equation}
where $s_{\tilde{\bX}_{\tau_0}}$ stands for the residual sum of squares (RSS) from the ordinary least squares (OLS) with $\bX=\tilde{\bX}_{\tau_0}$ (i.e., a regression without break), $s_{\tilde{\bX}_{\tau_0},\tilde{\bX}_{\btau}^{\hat{A}}}$ is the RSS from the OLS with $\bX=(\tilde{\bX}_{\tau_0},\tilde{\bX}_{\btau}^{\hat{A}})$, the value $k_{\hat{A}} = |\hat{A}|$ denotes the number of first-difference parameters different from zero in the model and $g_{\hat{A}}$ is a user parameter. We properly derive the criterion in SA \ref{app:linreg}. \cite{fernandez2001benchmark} show that the criterion \eqref{eq:ML} is consistent in the sense that it selects asymptotically the true subset of regressors when $g_{\hat{A}} = w(T)^{-1}$ as stated in proposition \ref{prop:consistebyt}. \\

\begin{theo} \label{prop:consistebyt}
\citep[Adaption of][]{fernandez2001benchmark}. Conditional on the true break dates, the criterion \eqref{eq:ML} is asymptotically maximized for the true subset of covariate $A$ if the following conditions on the parameter $g_{\hat{A}}=w(T)^{-1}$ holds i)  $\lim_{T\ra \infty} w(T) = \infty$, ii) $\lim_{T \ra \infty} \frac{w'(T)}{w(T)} = 0$, iii) $\lim_{T \ra \infty} \frac{T}{w(T)} \in [0,\infty)$.
\end{theo}
\begin{proof}
See SA \ref{app:consistency}.
\end{proof}

\begin{remark} \label{rem2}
Proposition \ref{prop:consistebyt} can be readily adapted when the conditioning set is a potential break date set complying with Assumption \ref{tauistruetau2}.
\end{remark}

In \cite{fernandez2001benchmark}, they advocate for setting $g_{\hat{A}} = \text{min}(T^{-1},(k_{\hat{A}}+K)^{-2})$ as this prior empirically delivers good results for selecting the true covariates in standard linear regressions. However, we deviate from this benchmark prior by fixing $g_{\hat{A}} = \frac{1}{T^{\alpha}-1}$ with $\alpha=1$ when $k_{\hat{A}}=0$ and $\alpha = \frac{k_{\hat{A}} + \hat{m}_{\hat{A}}-1}{k_{\hat{A}}}>1$ when $k_{\hat{A}}>0$ in which $\hat{m}_{\hat{A}}$ denotes the number of active segments. When $\alpha>1$, we show in SA \ref{App:BIC} that the criterion in Equation \eqref{eq:ML} asymptotically converges in probability to
\begin{align} 
\begin{split}
\ln f(\by|\ba,\lambda,\boldsymbol{\tau}) -  \left(-\frac{T}{2} \ln s_{\tilde{\mathbf{X}}_{\tau_0},\tilde{\mathbf{X}}_{\boldsymbol{\tau}}^{\hat{A}}}	-\frac{\alpha k_{\hat{A}}}{2} \ln T\right) & \overset{p}{\rightarrow} 0.
\end{split}
\end{align}
The asymptotic value is equivalent to the Bayesian information criterion (BIC) of a linear regression model exhibiting a number of parameters of $\alpha k_{\hat{A}}$.\footnote{The BIC of a linear regression model with K parameters is given by $-\frac{T}{2} \ln(\frac{s_{\tilde{\bX}_{\tau_0},\tilde{\bX}_{\btau}^{\hat{A}}}}{T})	-\frac{K}{2} \ln T$. So the marginal likelihood criterion of Equation \eqref{eq:ML} converges to the BIC up to an additive constant (that is $\frac{T}{2} \ln T$).} Consequently, the model penalty takes additionally into account the number of active breakpoints when $\alpha = \frac{k_{\hat{A}} + \hat{m}_{\hat{A}}-1}{k_{\hat{A}}}$. This stronger penalty works empirically well and is motivated by several CP papers advocating for stronger penalties than the BIC as it tends to overfit the number of regimes in finite sample \citep[see, e.g., ][]{LiuAl1997,zhang2007modified,kim2016consistent}.

 Interestingly, criterion \eqref{eq:ML} stands for a marginal likelihood in the Bayesian paradigm under $\epsilon \sim \NORM(0,\sigma^2 I_T)$ and the following prior,
\begin{equation} \label{eq:prior}
\begin{split}
f(\bbeta_1,\sigma^2) & \propto \sigma^{-2},\\
f(\Delta \bbeta_{\hat{A}}|\sigma^2,\btau) & \sim \NORM(\mathbf 0, \sigma^2 (g_{\hat{A}} (\tilde{\bX}_{\btau}^{\hat{A}})' \bM_{\tilde{\bX}_{\tau_0}} \tilde{\bX}_{\btau}^{\hat{A}})^{-1}), \text{ and } f(\Delta \bbeta_{\hat{A}^c}) \sim \text{Dirac}{(\mathbf 0)},
\end{split}
\end{equation}
where $\bM_{\tilde{\bX}_{\tau_0}} = I_T - \tilde{\bX}_{\tau_0}((\tilde{\bX}_{\tau_0})'\tilde{\bX}_{\tau_0})^{-1}(\tilde{\bX}_{\tau_0})'$. The prior distributions given by Equations \eqref{eq:prior} lead to simple posterior inference. The posterior distribution of the model parameters are given by, see SA \ref{App:post} for derivations,
\begin{eqnarray*}
\sigma^2|\by,\btau & \sim & \mathcal{IG}(\frac{T-K}{2},\frac{\frac{g_{\hat{A}}}{1+g_{\hat{A}}}s_{\tilde{\bX}_{\tau_0}}+\frac{1}{(1+g_{\hat{A}})}s_{\tilde{\bX}_{\tau_0},\tilde{\bX}_{\btau}^{\hat{A}}}}{2}),\\
\bbeta_1|\by,\sigma^2,\Delta \bbeta,\btau & \sim & \NORM((\tilde{\bX}_{\tau_0}'\tilde{\bX}_{\tau_0})^{-1}\tilde{\bX}_{\tau_0}' (\by- \tilde{\bX}_{\btau}^{\hat{A}}\Delta \bbeta), \sigma^2(\tilde{\bX}_{\tau_0}'\tilde{\bX}_{\tau_0})^{-1}), \\
\Delta \bbeta_{\hat{A}}|\by,\sigma^2,\btau & \sim & \NORM((1+g_{\hat{A}})^{-1}[(\tilde{\bX}_{\btau}^{\hat{A}})'\bM_{\tilde{\bX}_{\tau_0}}\tilde{\bX}_{\btau}^{\hat{A}}]^{-1}(\tilde{\bX}_{\btau}^{\hat{A}})'\bM_{\tilde{\bX}_{\tau_0}}\by, \frac{\sigma^2}{(1+g_{\hat{A}})}[(\tilde{\bX}_{\btau}^{\hat{A}})'\bM_{\tilde{\bX}_{\tau_0}}\tilde{\bX}_{\btau}^{\hat{A}}]^{-1}), \\
\Delta \bbeta_{\hat{A}^c}|\by,\btau & = & \mathbf 0,
\end{eqnarray*}
in which $\mathcal{IG}(-,-)$ denotes the Inverse-Gamma distribution.
Consequently, we can go beyond selecting the best pair $(\boldsymbol{a}_p,\lambda_p)$ (i.e., the pair that maximizes the criterion \eqref{eq:ML}) and can take the uncertainty of this selection into account. Given a set of models $M_z = (\boldsymbol{a}_z,\lambda_z)$, with $z=1,...,Z$, we can directly assess the posterior probability of a specific model as follows
\begin{equation} \label{eq:postprob}
\begin{split}
f(M_p|\by,\btau) & = \frac{f(\by|\boldsymbol{a}_p,\lambda_p,\btau)f(M_p|\btau)}{\sum_{z=1}^{Z}f(\by|\boldsymbol{a}_z,\lambda_z,\btau)f(M_z|\btau)}, \forall p \in [1,Z],
\end{split}
\end{equation}
where $f(M_z|\btau)$ denotes the prior probability of model $M_z$. In this paper, we assume uninformative prior, so $f(M_z|\btau) = Z^{-1}$. The posterior probability can be used to account for uncertainty on the selected regressors. In fact, we have
\begin{equation} \label{eq:postprob1}
\begin{split}
f(\bbeta_1,\Delta \bbeta,\sigma^2,M|\by,\btau) & =  f(\bbeta_1|\by,\btau,\sigma^2,\Delta \bbeta,M) f(\Delta \bbeta|\by,\btau,\sigma^2,M) \\
 & \quad f(\sigma^2|\by,\btau,M) f(M|\by,\btau)
\end{split}
\end{equation}
It is worth emphasizing that the consistent property of the criterion \eqref{eq:ML} does not depend on the normality assumption. Only, the posterior distribution of the model parameters does. We do not see this as a limitation since one can easily extend the model with another distributional assumption and compute the posterior distribution by numerical integrations.

\subsection{Prediction using Bayesian model averaging}
Equation \eqref{eq:postprob} shows how to take into account the uncertainty of the model parameters with respect to the selection of the SELO parameters. The Bayesian paradigm also provides a simple tool to forecast the series taking this uncertainty into account. In particular, the predictive density $f(y_{T+1:T+h}|\by)$, for $h\geq 1$, is related to the posterior density as follows
\begin{eqnarray}
f(y_{T+1:T+h}|\by,\btau) & = & \sum_{z=1}^{Z} \int f(y_{T+1:T+h}|\by,\btau,\bbeta_1,\Delta \bbeta,\sigma^2,M_z)f(\bbeta_1,\Delta \bbeta,\sigma^2,M_z|\by,\btau) d\bbeta_1 d\Delta \bbeta d\sigma^2,\nonumber \\
 & \approx & \frac{1}{N} \sum_{i=1}^{N} f(y_{T+1:T+h}|\by,\btau,\bbeta_1^{(i)},\Delta \bbeta^{(i)},(\sigma^2)^{(i)},M^{(i)}),\label{eq:predBMA}
\end{eqnarray}
where $\{\bbeta_1^{(i)},\Delta \bbeta^{(i)},(\sigma^2)^{(i)},M^{(i)}\}_{i=1}^{N}$ are independent draws from the posterior distribution (i.e., $\bbeta_1,\Delta \bbeta,\sigma^2,M|\by,\btau$). From \eqref{eq:predBMA}, it is apparent that the predictive density takes the model uncertainty into account.\footnote{Using the full marginal likelihood for weighting the models' predictions could raise concerns as only the last segment matters in CP processes. However, as marginal likelihood is frequently used for selecting the number of regimes in the literature and because it is also informative about the fit of the last regime, this average should give large weights to the models exhibiting a good fit at the end of the sample. Nevertheless, we could also weight the models' predictions using the predictive marginal likelihood $f(y_{t_1+1:T}|y_{1:t_1},\btau)$ in which $t_1$ is a user-defined value.} This feature should be contrasted with the standard penalized regression literature in which forecasting is performed using one unique set of parameter estimates; i.e., the estimates given by one penalty parameter selected, for instance, by cross-validation or by an information criterion. \\
In practice, simulations from the posterior distribution are not required for evaluating the predictive density. Assuming that the future covariates $\bx_{T+1:T+h}$ are observed at time $T$, the predictive distribution of $y_{T+1:T+h}$ given a model $M_z$ turns out to be a multivariate student distribution. Supplementary Appendix \ref{App:predictive} documents the analytical expression of $f(y_{T+1:T+h}|\by,M_z)$. Therefore, we can efficiently take into account model uncertainty in the predictive density since Equation \eqref{eq:predBMA} simplifies into 
\begin{eqnarray}
f(y_{T+1:T+h}|\by,\btau) & = & \sum_{z=1}^{Z} f(y_{T+1:T+h}|\by,\btau,M_z)f(M_z|\by,\btau).
\end{eqnarray}


\subsection{How to choose the values of \texorpdfstring{$\lambda$}{TEXT} and\texorpdfstring{ $\boldsymbol{a}$}{TEXT}}
When the number of models to consider is too large to directly explore the model space using the criterion \eqref{eq:ML} (i.e., when $(m-1)K>10$, see remark \ref{rem:exact}), we rely on the SELO penalty function to uncover the promising explanatory variables. While the asymptotic result of Proposition \ref{theoconsislimitdist} is reassuring, it only applies if the parameters $\lambda$ and $\boldsymbol{a}$ are adequately chosen. Similar to what is generally done in the penalized regression literature, we propose to explore many values of $\lambda$ and $\boldsymbol{a}$ and consider each couple as a model that would be ultimately discriminated via criterion \eqref{eq:ML}. For the parameter $\boldsymbol{a}$, we use a value of $a_i = \kappa \times \text{std}(\hat{\beta}_{j1})$ for each $j$ of the $K$ parameters per regime where $\text{std}(\hat{\beta}_{j1})$ stands for the standard deviation of the OLS estimate $\hat{\beta}_{j1}$ when we assume no break in the linear regression (i.e., $\bX = \bX_{\tau_0}$). We test several values for the parameter $\kappa$, namely $\kappa \in \{0.1,1\}$. Regarding the penalty parameter $\lambda$, we test 50 different values uniformly spaced in the interval $(0,\bar{\lambda}]$ in which $\bar{\lambda} = 2\ln T$. The penalty imposes by the upper bound $\bar{\lambda}$ is conservative enough as it is stronger than standard information criteria such as the BIC (that corresponds to a penalty of $\frac{1}{2}\ln T$) and the modified BIC.

\section{Break date detection \label{sec:break}}
In this Section, we present one approach to obtain a set of potential break dates. Before going into details, it is worth emphasizing that our method for detecting which parameters vary when a break occurs is independent of the segmentation detection procedure used in the first phase. To build the break date set, we could, for instance, adapt the dynamic programming method of \cite{bai2003computation} for the marginal likelihood given by Equation \eqref{eq:ML} and therefore propose our own CP detection method. We could also detect the locations of the segments using one of the standard segmentation approaches such as \cite{BaiPerron1998}, \cite{killick2012optimal} or \cite{korkas2017multiple}. Even better, we could apply several CP detection algorithms and discriminate between the sets of breakpoints by comparing their marginal likelihoods \textit{once} the SELO optimization has been carried out on each set. However, as the emphasis of the paper is not on the break detection, we prefer relying on one break detection procedure, the one documented in \cite{yau2016inference}, because i) it delivers a set of potential break dates with a computational complexity of $\mathcal{O}(T \left(\log(T)\right)^2)$ (which is faster than $\mathcal{O}(T^2)$, i.e., the complexity of the dynamic programming method of \cite{bai2003computation}) and because ii) we slightly improve their CP detection procedure. In particular, their estimated breakpoints depend on one tuning parameter, the radius $h$. Instead of fixing it, we use multiple values of $h$ and we also adapt their approach to end up with a potential breakpoint set. \\
It is worth noting that, as the paper combines model selection and CP detection methods, our approach only requires a set of potential break dates that includes the correct break dates. By penalizing the parameter variation between two consecutive regimes, the spurious break dates are consistently deleted (see remarks \ref{rem1} and \ref{rem2}).

\subsection{Segmentation procedure \label{sec:breakYau}}
 \cite{yau2016inference} propose a likelihood ratio scan method in three steps for estimating multiple break dates in piecewise stationnary processes. They also establish the consistency of the estimated number and location fractions of the CPs. We apply their three steps to detect the break dates but we modify them to reduce the computational burden and to keep at the end of the procedure a potential break date set (that could overestimate the true number of regimes). We now detail the three steps that we use to segment the data.\\
\noindent \textbf{First step.} Fix a window radius $h\in [K+1,T-K]$. For $t=h$ to $T-h$, compute the likelihood ratio scan statistic given by,
\begin{equation}
\label{LRSS}
S_h(t) = \frac{1}{h}L_{t-h+1:t}(\hat{\bbeta},\hat{\sigma}) + \frac{1}{h}L_{t+1:t+h}(\hat{\bbeta},\hat{\sigma}) - \frac{1}{h}L_{t-h+1:t+h}(\hat{\bbeta},\hat{\sigma}),
\end{equation} 
\noindent where $\hat{L}_{t_1:t_2}(\hat{\bbeta},\hat{\sigma})$ denotes the maximum value of the log-likelihood of model  (\ref{eq:linreg}) over the segment $t\in [t_1,t_2]$, assuming that $\epsilon_t \sim \NORM(0,\sigma^2)$. Then, the set $\Gamma(h)$ of potential break dates is given by, 
\begin{equation}
\label{potentbreak}
\Gamma(h) = \left\{j  \in \{h+h+1,\dots,T-h\}; S_h(j) = \max_{t \in [j-h,j+h]} S_h(t)\right\},
\end{equation}
\noindent where $S_h(t) = 0$ for $t<h$ and $t>T-h$. As the window radius $h$ is crucial, we differ from \cite{yau2016inference} by using a grid of M values uniformly-spaced in the interval $[\frac{h_{\text{YZ}}}{2},2h_{\text{YZ}}]$ in which $h_{\text{YZ}}$ denotes their advocated value that is $h_{\text{YZ}} =\max\left\{25, \left(\log(T)\right)^2\right\}$ when $T < 800$ and  $h_{\text{YZ}} = \max\left\{50, 2\left(\log(T)\right)^2\right\}$ otherwise. So, at the end of the first step, we end up with M potential break date sets, i.e., $\Gamma(h_1),\ldots,\Gamma(h_M)$. \\
\noindent \textbf{Second step.} For every $z\in [1,M]$ and $i\in [1,m_{h_z}-1]$ where $m_{h_z} = |\Gamma(h_z)|+1$, we re-estimate each break date location $\tau_i^{(z)} \in \Gamma(h_z)$ as follows
\begin{eqnarray*}
\hat{\tau}_i^{(z)} & = & \text{argmax}_{t\in [\tau_i^{(z)}-h_z,\tau_i^{(z)}+h_z]} L_{\tau_i^{(z)}-\left\lfloor1.5h_z\right\rceil:t}(\hat{\bbeta},\hat{\sigma}) + L_{t+1:\tau_i^{(z)}+\left\lfloor 1.5 h_z\right\rceil}(\hat{\bbeta},\hat{\sigma}),
\end{eqnarray*}
in which $\left\lfloor x \right\rceil$ stands for the nearest integer to $x$. Gathering all the new locations in the set $\hat{\Gamma}(h_z) = \{\hat{\tau}_1^{(z)},\ldots,\hat{\tau}_{m_{h_z}}^{(z)}\}$, it is clear from Theorems 1 to 3 in \cite{yau2016inference} that for any $j \in \{1,\dots,m-1\}$, there exist $\hat{\tau}_i^{(z)} \in \hat{\Gamma}(h_z)$ with $i\in [1,m_{h_z}-1]$ such that $\hat{\tau}_i^{(z)} - \tau_j = \mathcal{O}_p(1)$. 

\noindent \textbf{Third step.} We select the best breakpoints among the M potential break date sets by minimizing the Minimum Description Length (MDL) defined by, for $z\in [1,M]$, 
\begin{align}
\begin{split}
\label{MDL}
\text{MDL}(h_z) &= \ln^{\mathcal{+}} \left(m_{h_z}-1\right) + m_{h_z}\ln \left(T\right) \, \\
 & \quad + \sum_{j=1}^{m_{h_z}}\left(\frac{K + 1}{2}\log(\hat{\tau}_j^{(z)} - \hat{\tau}_{j-1}^{(z)}) - L_{\hat{\tau}_{j-1}^{(z)}+1:\hat{\tau}_{j}^{(z)}}(\hat{\bbeta},\hat{\sigma}) \right),
\end{split}
\end{align}
where $\hat{\tau}_{0}=0$, $\hat{\tau}_{m_{h_z}}=T$ and $\{\hat{\tau}_j\}_{j=2,\dots,m_{h_z}-1}=\hat{\Gamma}(h_z)$. In practice, we fix $M=30$.\\ 


\subsection{Break uncertainty}
\label{sec:breakuncertainty}
Given a set of break dates obtained either from the procedure described in Section \ref{sec:breakYau} or from any other existing break detection method such as the one of  \cite{BaiPerron1998}, our method to uncover the partial structural changes can be undertaken. Let us denote by $M_*=(\boldsymbol{a}_*,\lambda_*)$ the SELO parameters maximizing the marginal likelihood criterion \eqref{eq:ML} and their corresponding break dates $J = \{\bar{\tau}_0 =0, \bar{\tau}_1, \ldots, \bar{\tau}_{\hat{m}-1}, \bar{\tau}_{\hat{m}} = T\}$. To provide break uncertainty, we shall infer the posterior distribution of the structural breaks; i.e., $\btau \equiv \tau_1,\ldots,\tau_{\hat{m}-1}|\by,M_*$. To do so, we first assume uninformative priors for the break dates using the set $J$. For $i=1,...,\hat{m}-1$, the break parameter $\tau_i$ is driven by a Uniform distribution as follows 
\begin{eqnarray*}
\tau_i &\sim & \mathcal{U}[\left\lfloor \frac{\bar{\tau}_{i-1}+\bar{\tau}_{i}}{2}\right\rfloor + \gamma, \left\lfloor \frac{\bar{\tau}_{i-1}+\bar{\tau}_{i}}{2}\right\rfloor - \gamma], 
\end{eqnarray*}
in which $\left\lfloor x \right\rfloor$ stands for the nearest integer less than or equal to $x$ and $\gamma=(K+1)$ is a minimum duration parameter ensuring that the marginal likelihood criterion \eqref{eq:ML} can be computed for any break parameters complying with the prior distributions given by Equation \eqref{eq:prior}. The posterior density is proportional to 
\begin{align}
\begin{split}
f(\btau|\by,M_*) & \propto  f(\by|M_*,\btau) f(\by|M_*,\btau) \left(\prod_{i=1}^{\hat{m}-1} \1{\btau_i \in [\left\lfloor \frac{\bar{\tau}_{i-1}+\bar{\tau}_{i}}{2}\right\rfloor - \gamma, \left\lfloor \frac{\bar{\tau}_{i-1}+\bar{\tau}_{i}}{2}\right\rfloor + \gamma]}\right).
\end{split}
\end{align}
As shown in SA \ref{app:linreg}, the marginal likelihood $f(\by|M_*,\btau)=f(\by|\boldsymbol{a}_*,\lambda_*,\btau)$ exhibits a closed form expression. Several solutions exist to sample the break parameters \citep[see, e.g.,][]{Stephens1994,Liao08}. In this paper, we use the D-DREAM algorithm developed in \cite{DREAM11}. It builds a symmetric proposal distribution inspired by the Differential Evolution optimization literature and draws from this proposal distribution are accepted or rejected through a Metropolis step in a Markov-chain Monte Carlo (MCMC) algorithm. As shown in \cite{DREAM11}, the D-DREAM algorithm complexity is $\mathcal{O}(T)$ and leads to a rapidly mixing MCMC algorithm since the break parameters are jointly sampled from the proposal distribution. To infer the break parameters, we apply the following steps:
\begin{itemize}
	\item Sample $R=2m$ initial structural break vectors $\{\btau_i\}_{i=1}^{R}$ from the prior distribution.
	\item At each MCMC iteration, for each $j=1,...,R$, apply the D-DREAM Metropolis move:
	
	\begin{enumerate}
		\item Propose a new draw of the break parameter as follows
		\begin{equation} \label{propDREAM}
\hat{\btau}_j = \btau_j + \left\lfloor\gamma(\delta, m) (\sum_{g=1}^{\delta} \btau_{r_1(g)}- \sum_{h=1}^{\delta} \btau_{r_2(h)}) + \xi \right\rceil,
\end{equation}
with  $\xi \sim \NORM(0,(0.0001) I)$ and $\forall g, h = 1, 2, ..., \delta$, $j \neq r_1(g)$, $r_2(h)$; $r_1(.)$ and $r_2(.)$ stand for random integers uniformly distributed on the support $[1,R]$. We set $\gamma(\delta,m) = \frac{2.38}{\sqrt{2\delta m}}$ and $\delta \sim \mathcal{U}[1,3]$.\footnote{When the posterior distribution is a multivariate normal one, \cite{terBraak2006} proves that choosing $\gamma = \frac{2.38}{\sqrt{\delta m}}$ leads to the optimal acceptance rate of the Metropolis ratio. As shown in \cite{terBraak2006}, the proposal distribution works when the number of chains, i.e. $\delta$, is equal to one. However, the mixing of the MCMC algorithm can be improved by increasing $\delta$ as illustrated with simulation exercises in \cite{Vrugt2009}.}
 \item Accept the proposal $\hat{\btau}_j$ according to the probability $\alpha(\btau_j , \hat{\btau}_j) = \min \{\frac{f(\by|M_*,\hat{\btau}_j)}{f(\by|M_*,\btau_j)}, 1\}.$
	\end{enumerate}
\end{itemize}
In practice, we set the number of MCMC iterations to $4000$ and start collecting the draws after $\text{round}[\frac{M}{2}]$ MCMC iterations. In addition, we assess the convergence of the MCMC algorithm using the multivariate Potential Scale Reduction Factor test proposed in \cite{brooks1998general}. For the two in-sample applications below in which credible intervals of the breakpoints are computed, the convergence statistics amount to 1.014 and 1.052, respectively. These values meet the threshold of 1.1 commonly used to validate the convergence of MCMCs.

\section{Monte Carlo study} \label{sec:MonteCarlo}
In this Section, we document a Monte Carlo study to assess the accuracy of the SELO approach. We first rely on nine different data generating processes (DGPs) that are documented in Table \ref{DGP::MC}. For each DGP, we simulate 1000 series with a sample size equal to $T=1024$ and we investigate i) the performance of detecting the break dates using the approach in Section \ref{sec:breakYau} and ii) the performance of the SELO method for detecting which parameter truly varies when a break occurs. The nine DGPs differ in their mean parameter specifications. For each of them, we study the SELO performance when the innovation is either homoskedastic or driven by a GARCH process. \\
Regarding the DGPs, the first six DGPs are piecewise stationary AR models directly taken from \cite{yau2016inference} while the others cover situations with exogenous explanatory variables. DGP A and E do not exhibit any breakpoint. They aim at showing the performance of the SELO approach when only spurious break dates are detected. DGPs B and C are weakly persistent piecewise stationary AR models exhibiting three regimes. Simulated series from DGP D experience a break after 50 observations. This DGP should highlight the performance in a short regime context. DGPs E and F are highly persistent piecewise stationary AR models but DGP F differs by exhibiting breaks in the mean parameters. Eventually, DGPs G, H and I include exogenous variables. While DGPs G only exhibits exogeneous regressors, DGPs H and I stand for ARX processes by mixing the parameters of the DGPs B and G.

\begin{table}[h!]
\centering
\singlespacing
\caption{\justifying \textbf{Data Generating Processes of sample size amounting to $T=1024$.}\\
This Table summarizes the DGPs from which 1000 series are simulated for the Monte Carlo study. The variables $\text{V}$ and $\text{W}$ stand for exogenous variables such that, $V_t \sim \NORM(0,3^2)$ and $W_t \sim \NORM(0,4^2)$. For instance, DGP B is an AR(2) model that exhibits two breakpoints at $t=512$ and $t=768$. The true values of the first AR term for the first two regimes are equal to 0.9 and 1.69, respectively. The dynamic of the variance is either homoskedastic ('Constant') or heteroskedastic ('GARCH'). \label{DGP::MC}}
\scalebox{0.95}{
	\begin{tabular}{lccc}
		\toprule
		 & \textbf{DGP A} & \textbf{DGP B} & \textbf{DGP C}  \\
		Breaks& -     & [512, 768] & [400, 612]   \\
		\midrule
		Intercept & [0] & [0, 0, 0] & [0, 0, 0]  \\
		$\text{AR}_1$ & [- 0.7] & [0.9, 1.69, 1.32] & [0.4, - 0.6, 0.5]   \\
		$\text{AR}_2$ & -     & [0, - 0.81, - 0.81] & -      \\[0.2cm]
		\hline
		\toprule
		 & \textbf{DGP D} & \textbf{DGP E} & \textbf{DGP F}   \\
		Breaks & [50] & - & [400, 750] \\
		\midrule
		Intercept & [0, 0] & [0] & [0, 0, 0]  \\
		$\text{AR}_1$ & [0.75, - 0.5] & [0.999] & [1.399, 0.999, 0.699]  \\
		$\text{AR}_2$ & -     & -     & [- 0.4, 0, 0.3]  \\[0.2cm]
		\hline
		\toprule
		 & \textbf{DGP G} & \textbf{DGP H} & \textbf{DGP I} \\
		Breaks & [400, 750] & [400,750] & [512, 768] \\
		\midrule
		Intercept & [1, 0, 0] & [0, 0, 0] & [0,0,0] \\
		$\text{AR}_1$ & -     & [0.9, 1.69, 1.32]  & [0.9, 1.69, 1.32] \\
		$\text{AR}_2$ & -     & [0, -0.81, -0.81] & [0, -0.81, -0.81] \\
		$\text{V}$ & [1.5, 0.9, 2.2] & [1.5, 0.9, 2.2] & [1.5, 0.9, 2.2] \\
		$\text{W}$ & [- 0.6, - 0.6, - 1] & [- 0.6, - 0.6, - 1] & [- 0.6, - 0.6, - 1] \\
			\hline
		\toprule
		& \multicolumn{3}{c}{Dynamic of the variance of $\epsilon_t \sim \mathcal{N}(0,\sigma_t^2)$}\\
		\midrule
		\textbf{Constant} & \multicolumn{3}{c}{$\sigma_t^2 = 1$, $\forall t\in [1,T]$} \\
		\textbf{GARCH}    & \multicolumn{3}{c}{$\sigma_t^2 = 0.05 + 0.05 \epsilon_{t-1}^2 + 0.9 \sigma_{t-1}^2$, $\forall t\in [1,T]$ and $\sigma_{0}^2 = \frac{0.05}{1-0.95} = 1$} \\
		\bottomrule
	\end{tabular}}
\end{table}%

\renewcommand{\arraystretch}{1.2}

\begin{table}[h!]
\centering
\singlespacing
\caption{\justifying \textbf{Break estimates : SELO approach.}\\
Based on 1000 replications, this Table presents several metrics for assessing the performance of the SELO method on DGPs detailed in Table \ref{DGP::MC}. \textbf{Number of regimes} is the rate of detecting a specific number of regimes per model parameter. Bold values correspond to the true number of regimes. \textbf{Break} documents the rate of having at least one breakpoint in the potential CP set located in the neighborhood of 50 observations of the true breakpoints. We use '---' when the DGP exhibits no breakpoint.  \textbf{Exact} denotes the rate of detecting the true number of breakpoints for all the model parameters with a posterior probability of at least 10\%.\label{MC::bdatesselo}}
\scalebox{0.65}{
	\begin{tabular}{ll cccccccccccccccc}
		\toprule
		&  & \multicolumn{8}{c}{\textbf{Constant Variance}} & \multicolumn{8}{c}{\textbf{GARCH Variance}} \\
		\cmidrule(lr){3-10}\cmidrule(lr){11-18}
	 	&  &  \multicolumn{6}{c}{\textbf{Number of regimes}} & Break & Exact &  \multicolumn{6}{c}{\textbf{Number of regimes}} &  Break & Exact\\
		\cmidrule(lr){3-8}\cmidrule(lr){11-16}	
		\textbf{DGP} &  &   1 & 2 & 3 & 4 & 5 & 6 & &  &   1 & 2 & 3 & 4 & 5 & 6 & & \\ 
		\midrule
		\hline
A & 	Intercept & 	\textbf{99.4} & 	0.6 & 	0& 	0& 	0& 	0& 	\multirow{2}{*}{ --- }  & 	\multirow{2}{*}{99.9}  & 	\textbf{99.2} & 	0.8 & 	0& 	0& 	0& 	0& 	\multirow{2}{*}{ --- }  & 	\multirow{2}{*}{99.2}  \\ 
 & 	AR1 & 	\textbf{99.5} & 	0.5 & 	0& 	0& 	0& 	0& 	 & 	 & 	\textbf{99.4} & 	0.6 & 	0& 	0& 	0& 	0& 	 & 	 \\[0.3cm] 
B & 	Intercept & 	\textbf{98.6} & 	1.4 & 	0& 	0& 	0& 	0& 	\multirow{3}{*}{100 }   & 	\multirow{3}{*}{99.7}  & 	\textbf{97.3} & 	2.7 & 	0& 	0& 	0& 	0& 	\multirow{3}{*}{99.3}  & 	\multirow{3}{*}{99.5}  \\ 
 & 	AR1 & 	0& 	0& 	\textbf{100 }  & 	0& 	0& 	0& 	 & 	 & 	0& 	0.2 & 	\textbf{99.4} & 	0.4 & 	0& 	0& 	 & 	 \\ 
 & 	AR2 & 	0& 	\textbf{98.8} & 	1.2 & 	0& 	0& 	0& 	 & 	 & 	0& 	\textbf{98.3} & 	1.7 & 	0& 	0& 	0& 	 & 	 \\[0.3cm] 
C & 	Intercept & 	\textbf{97.9} & 	2& 	0.1 & 	0& 	0& 	0& 	\multirow{2}{*}{99.8}  & 	\multirow{2}{*}{99.7}  & 	\textbf{97.6} & 	2.4 & 	0& 	0& 	0& 	0& 	\multirow{2}{*}{99.8}  & 	\multirow{2}{*}{99.1}  \\ 
 & 	AR1 & 	0& 	0& 	\textbf{100 }  & 	0& 	0& 	0& 	 & 	 & 	0& 	0& 	\textbf{99.7} & 	0.3 & 	0& 	0& 	 & 	 \\[0.3cm] 
D & 	Intercept & 	\textbf{97.4} & 	2.6 & 	0& 	0& 	0& 	0& 	\multirow{2}{*}{99.8}  & 	\multirow{2}{*}{99.5}  & 	\textbf{97.6} & 	2.2 & 	0.2 & 	0& 	0& 	0& 	\multirow{2}{*}{99.7}  & 	\multirow{2}{*}{99.1}  \\ 
 & 	AR1 & 	0.1 & 	\textbf{99.4} & 	0.5 & 	0& 	0& 	0& 	 & 	 & 	0.2 & 	\textbf{99.3} & 	0.4 & 	0.1 & 	0& 	0& 	 & 	 \\[0.3cm] 
E & 	Intercept & 	\textbf{86.4} & 	12.4 & 	1.2 & 	0& 	0& 	0& 	\multirow{2}{*}{ --- }  & 	\multirow{2}{*}{94.6}  & 	\textbf{84.8} & 	12.5 & 	2.5 & 	0.2 & 	0& 	0& 	\multirow{2}{*}{ --- }  & 	\multirow{2}{*}{91.5}  \\ 
 & 	AR1 & 	\textbf{93.6} & 	6.1 & 	0.3 & 	0& 	0& 	0& 	 & 	 & 	\textbf{91 }  & 	8.2 & 	0.5 & 	0.3 & 	0& 	0& 	 & 	 \\[0.3cm] 
F & 	Intercept & 	\textbf{69.7} & 	23.9 & 	6.2 & 	0.1 & 	0.1 & 	0& 	\multirow{3}{*}{25.5}  & 	\multirow{3}{*}{23.2}  & 	\textbf{65.3} & 	27.5 & 	6.7 & 	0.5 & 	0& 	0& 	\multirow{3}{*}{22.4}  & 	\multirow{3}{*}{22.1}  \\ 
 & 	AR1 & 	0& 	68.7 & 	\textbf{31 }  & 	0.1 & 	0.2 & 	0& 	 & 	 & 	0& 	69.5 & 	\textbf{29.6} & 	0.8 & 	0.1 & 	0& 	 & 	 \\ 
 & 	AR2 & 	0& 	71.5 & 	\textbf{28.3} & 	0.2 & 	0& 	0& 	 & 	 & 	0& 	73.2 & 	\textbf{26.4} & 	0.4 & 	0& 	0& 	 & 	 \\[0.3cm] 
G & 	Intercept & 	0& 	\textbf{99.3} & 	0.7 & 	0& 	0& 	0& 	\multirow{3}{*}{100 }   & 	\multirow{3}{*}{99.8}  & 	0& 	\textbf{99.2} & 	0.8 & 	0& 	0& 	0& 	\multirow{3}{*}{100 }   & 	\multirow{3}{*}{99.8}  \\ 
 & 	V & 	0& 	0& 	\textbf{99.8} & 	0.2 & 	0& 	0& 	 & 	 & 	0& 	0& 	\textbf{99.7} & 	0.3 & 	0& 	0& 	 & 	 \\ 
 & 	W & 	0& 	\textbf{99.2} & 	0.8 & 	0& 	0& 	0& 	 & 	 & 	0& 	\textbf{99 }  & 	0.9 & 	0.1 & 	0& 	0& 	 & 	 \\[0.3cm] 
H & 	Intercept & 	\textbf{88.9} & 	11& 	0.1 & 	0& 	0& 	0& 	\multirow{5}{*}{100 }   & 	\multirow{5}{*}{83.1}  & 	\textbf{92.9} & 	6.9 & 	0.2 & 	0& 	0& 	0& 	\multirow{5}{*}{100 }   & 	\multirow{5}{*}{86.8}  \\ 
 & 	AR1 & 	0& 	0& 	\textbf{92.7} & 	7.3 & 	0& 	0& 	 & 	 & 	0& 	0& 	\textbf{94.7} & 	5.3 & 	0& 	0& 	 & 	 \\ 
 & 	AR2 & 	0& 	\textbf{92.6} & 	7.4 & 	0& 	0& 	0& 	 & 	 & 	0& 	\textbf{94.1} & 	5.9 & 	0& 	0& 	0& 	 & 	 \\ 
 & 	V & 	0& 	0& 	\textbf{87.7} & 	12.3 & 	0& 	0& 	 & 	 & 	0& 	0& 	\textbf{89.6} & 	10.4 & 	0& 	0& 	 & 	 \\ 
 & 	W & 	0& 	\textbf{88 }  & 	12& 	0& 	0& 	0& 	 & 	 & 	0& 	\textbf{90.4} & 	9.6 & 	0& 	0& 	0& 	 & 	 \\[0.3cm] 
I & 	Intercept & 	\textbf{91.6} & 	8.4 & 	0& 	0& 	0& 	0& 	\multirow{5}{*}{100 }   & 	\multirow{5}{*}{85.7}  & 	\textbf{91 }  & 	8.9 & 	0.1 & 	0& 	0& 	0& 	\multirow{5}{*}{100 }   & 	\multirow{5}{*}{85.1}  \\ 
 & 	AR1 & 	0& 	0& 	\textbf{94.3} & 	5.7 & 	0& 	0& 	 & 	 & 	0& 	0& 	\textbf{95 }  & 	4.9 & 	0.1 & 	0& 	 & 	 \\ 
 & 	AR2 & 	0& 	\textbf{94.6} & 	5.4 & 	0& 	0& 	0& 	 & 	 & 	0& 	\textbf{94.9} & 	5& 	0.1 & 	0& 	0& 	 & 	 \\ 
 & 	V & 	0& 	0& 	\textbf{89.8} & 	10.2 & 	0& 	0& 	 & 	 & 	0& 	0& 	\textbf{89.4} & 	10.6 & 	0& 	0& 	 & 	 \\ 
 & 	W & 	0& 	\textbf{88.7} & 	11.1 & 	0.2 & 	0& 	0& 	 & 	 & 	0& 	\textbf{90 }  & 	10& 	0& 	0& 	0& 	 & 	 \\[0.3cm] 
		\hline
		 \hline
		\bottomrule 
	\end{tabular}}
\end{table}

\noindent Table \ref{MC::bdatesselo} documents the percentage of detecting a number of regimes per model parameter over the 1000 simulated series per DGP for the SELO method. Overall, the detection rates of identifying the true number of regimes per parameter are excellent and besides DGP F, they are at least equal to 86.4\%. Interestingly, this detection rate does not deteriorate when the innovation is driven by a GARCH process. The worst detection rates arise for the DGP F. Even though this DGP is highly persistent with an autocorrelation structure that barely varies over time, the SELO method correctly identifies that the intercept does not experience abrupt switches 69.7\% of the times. Note that the potential breakpoint sets for this DGP poorly identify the true breakpoints since only 25.5\% of the sets exhibit at least one potential CP close to every true breakpoints. Therefore, the SELO detection rate could hardly exceed this bound. As exemplified by DGPs G, H and I, the detection rates of the SELO method remain excellent when exogenous variables kick in even in the presence of heteroscedasticity. The Table also documents the rate of detecting the true model (i.e. jointly the correct number of regimes) with a posterior probability of at least 10\%.\footnote{For this simulation study, the number of explanatory variables ranges from 2 to 5 and the maximum number of potential regimes observed for each DGP is as follows: DGP A (5), DGP B (5), DGP C (6), DGP D (6), DGP E (6), DGP F (6), DGP G (5), DGP H (9), DGP I (9). It can thus lead to a number of models amounting to $2^{40}$. In such a case, the set of the models exhibiting a probability equal or greater than 10\% has a prior probability of containing the true model that is approximately equal to $\frac{1000}{2^{40}}\%$.} For all the DGPs but DGP F, the correct detection amounts to at least 83.1\% and 85.1\% for the constant and the GARCH innovation dynamics, respectively. These excellent results highlight that model uncertainty should be taken into account since several models often exhibit high posterior probabilities.

\noindent DGPs from Table \ref{DGP::MC} are frequently used in the CP literature to assess the performance of a new segmentation method (see, e.g., \cite{cho2015multiple}, \cite{yau2016inference} and \cite{korkas2017multiple}). Nevertheless, our empirical exercise implies more explanatory variables and a smaller sample size. To assess the SELO performance in such environment, we also consider fourteen variants of an 'empirical DGP' given by
\begin{eqnarray}
y_t & = & \begin{cases}\bx_t'\bbeta_{1} + \sigma_t \epsilon_t, & \quad\text{if } 1\leq t \leq 132,\\ 	
  \bx_t'\bbeta_{2}  + \sigma_t \epsilon_t, & \quad\text{if }133 \leq T,\end{cases} \label{DGP:Emp}
\end{eqnarray}
where $T=256$ as in the application, $\epsilon_t \sim \NORM(0,1)$ and $\bx_t=(1,x_{t,1},\ldots,x_{t,12})'$. The explanatory variables are close to the risk factors used in our empirical exercise. In particular, they are generated from AR models whose coefficients and AR orders are estimated using the risk factors of the application. The parameter values of $\bbeta_1$ are equal to the OLS estimates of the Hedge fund Index (HFI) regression without breakpoints (see Table \ref{tab:OLS_SELO1} in the empirical application). We consider 14 variants of the DGP given by Equation \eqref{DGP:Emp} that differ by the number of parameters experiencing a breakpoint. Defining $\bbeta_{2}=(\beta_{2,1},\ldots,\beta_{2,13})'$ and considering the $ith$ DGP, with $i=1,\ldots,14$, we have $\beta_{2,j} = \beta_{1,j} + 3\omega_j \text{sign}(\beta_{1,j})$ for $j< i$ and $\beta_{2,j} = \beta_{1,j}$ for $j\geq i$. The size of the break given by $\omega_j$ is equal to the standard deviation of the $jth$ OLS estimate of the Hedge fund Index (HFI) regression without breakpoints (see Table \ref{tab:OLS_SELO1}). To summarize, the first DGP does not exhibit a breakpoint while the 14th one exhibits a structural change in all its parameters. As before, we consider homoskedastic errors with $\sigma_t^2\equiv \bar{\omega}^2 = 1.7$ $\forall t$ and heteroskedastic ones with $\sigma_t^2= 0.05\bar{\omega}^2 + 0.05 \epsilon_{t-1}^2 + 0.9\sigma_{t-1}^2$ for $t>1$ in which $\bar{\omega}^2$ stands for the OLS variance estimate of the Hedge fund Index (HFI) regression without breakpoints (see Table 7 of the paper). Figure \ref{fig:DGPs_emp1} displays one simulated series drawn from some of the fourteen variants with heteroskedastic errors.

\begin{figure}[h!]
	\begin{center}
		\subfloat[\# CPs = 0]{\includegraphics[width=7.5cm,height=5cm]{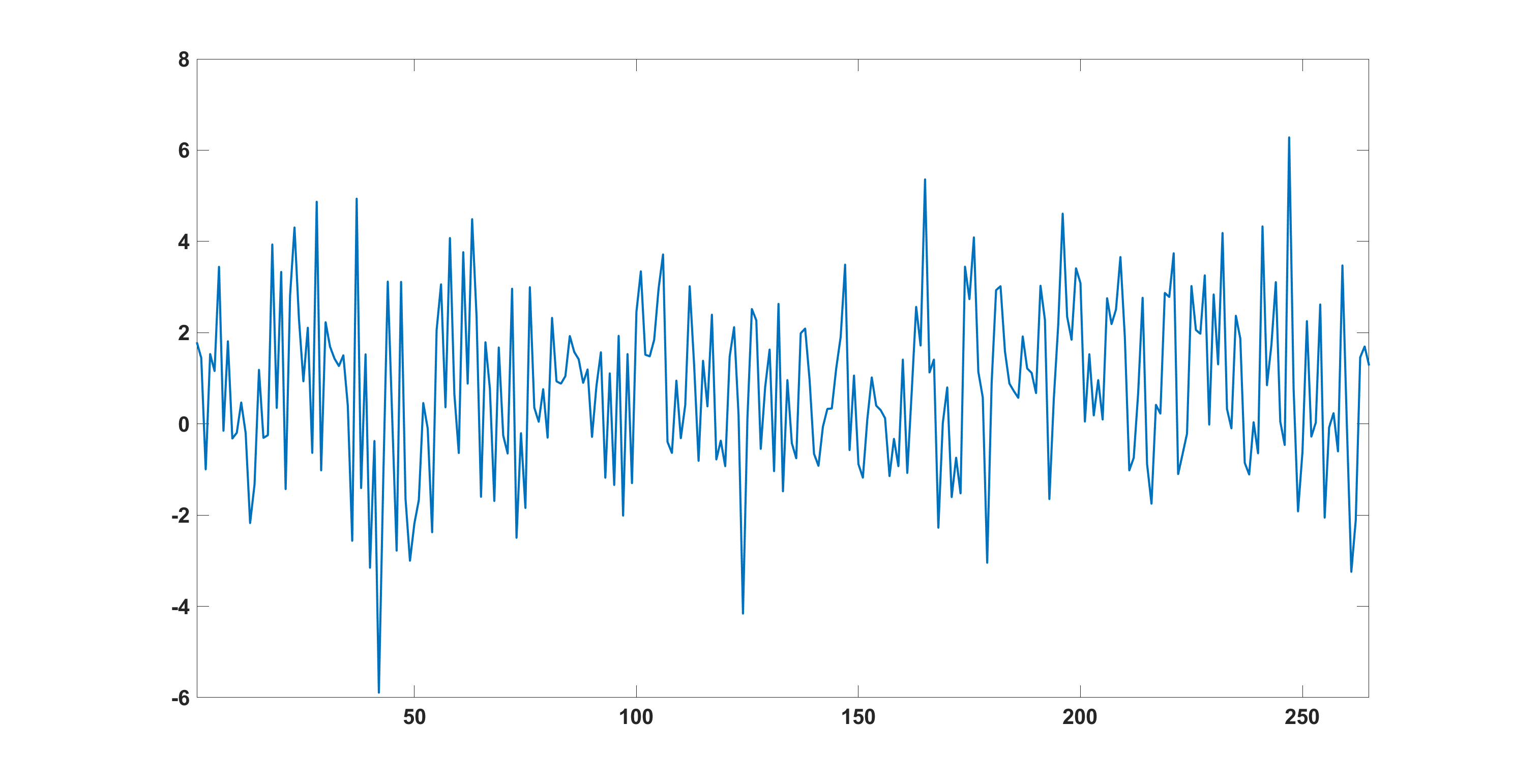}}
		\subfloat[\# CPs = 5]{\includegraphics[width=7.5cm,height=5cm]{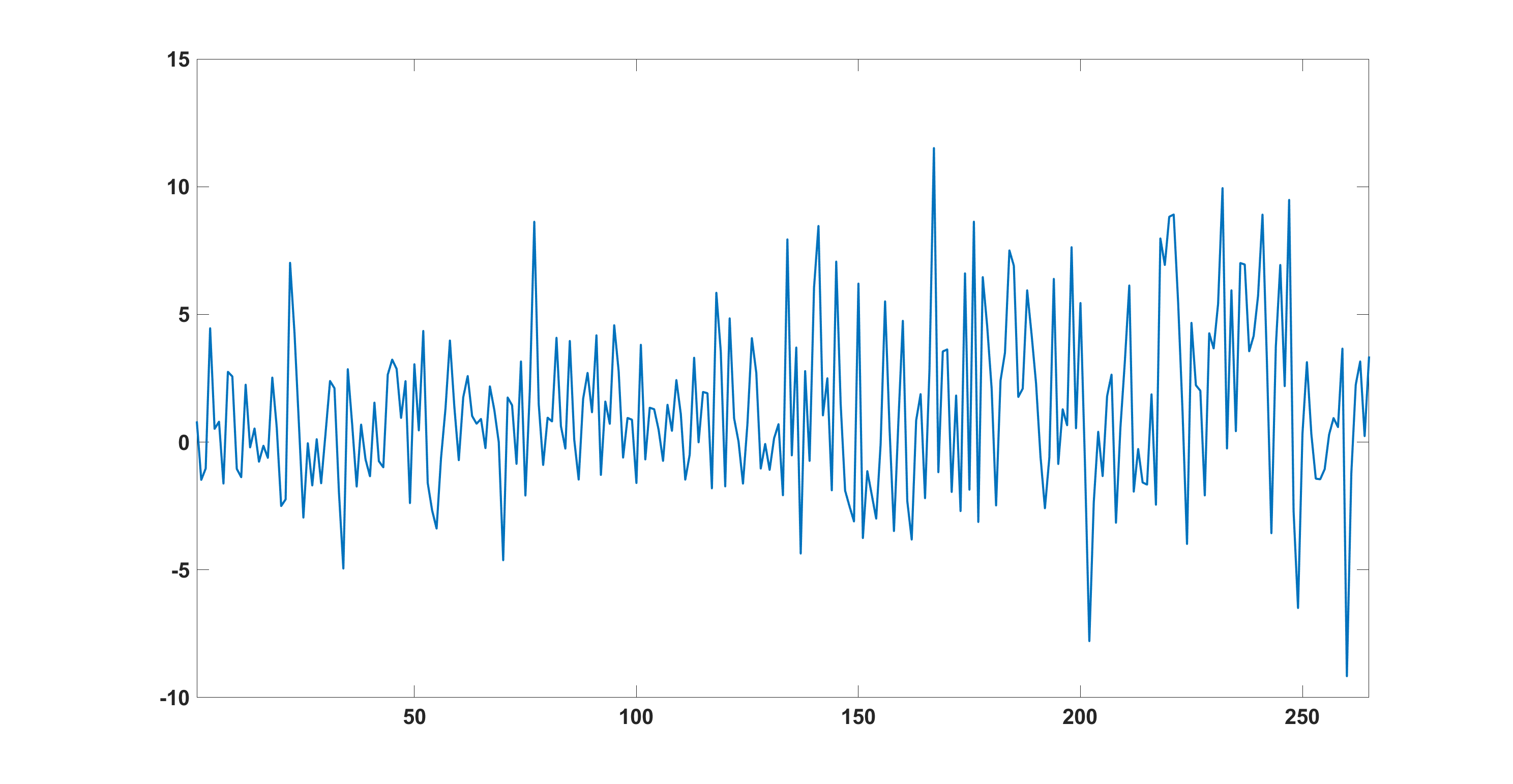}}\\
		\subfloat[\# CPs = 10]{\includegraphics[width=7.5cm,height=5cm]{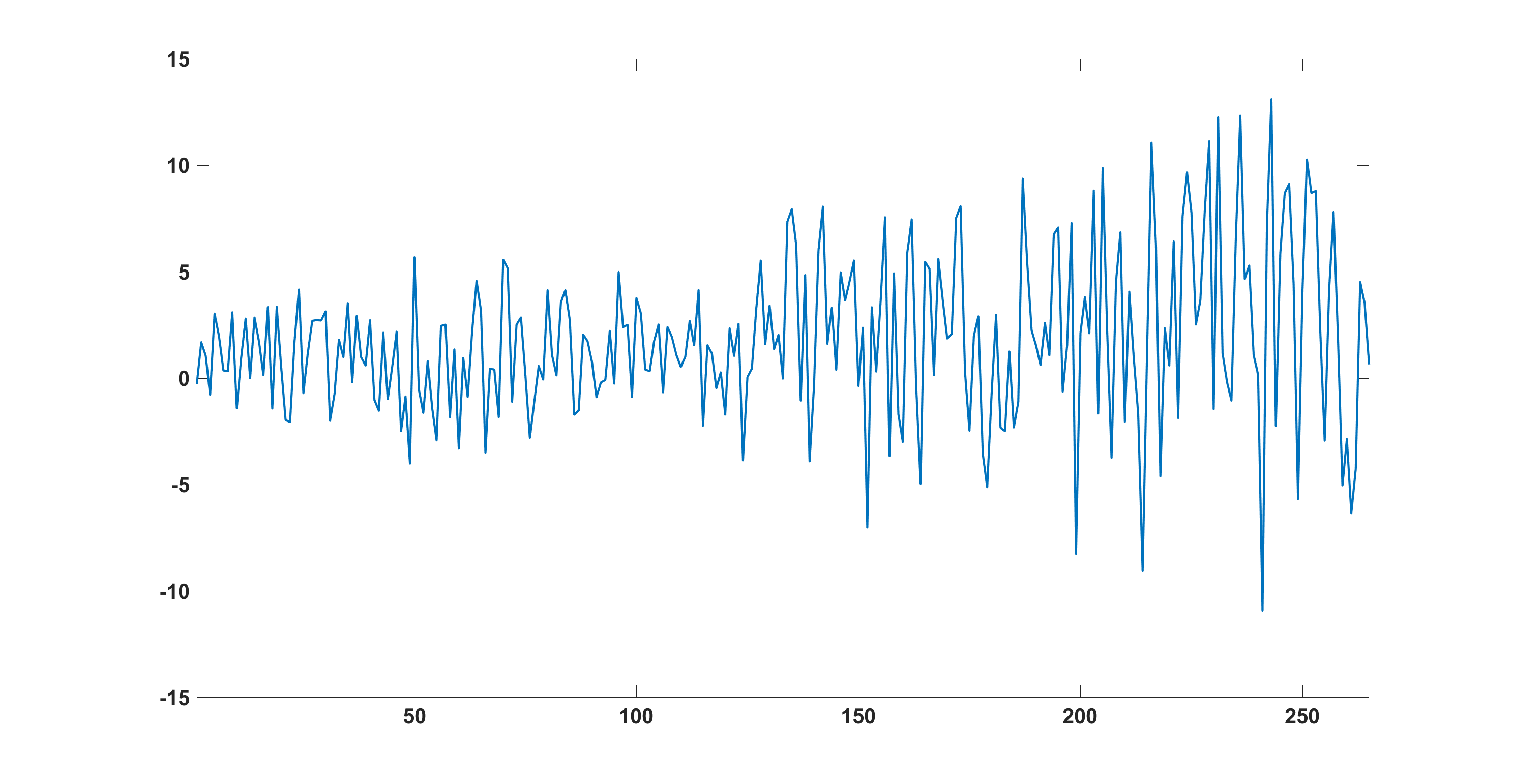}}
		\subfloat[\# CPs = 13]{\includegraphics[width=7.5cm,height=5cm]{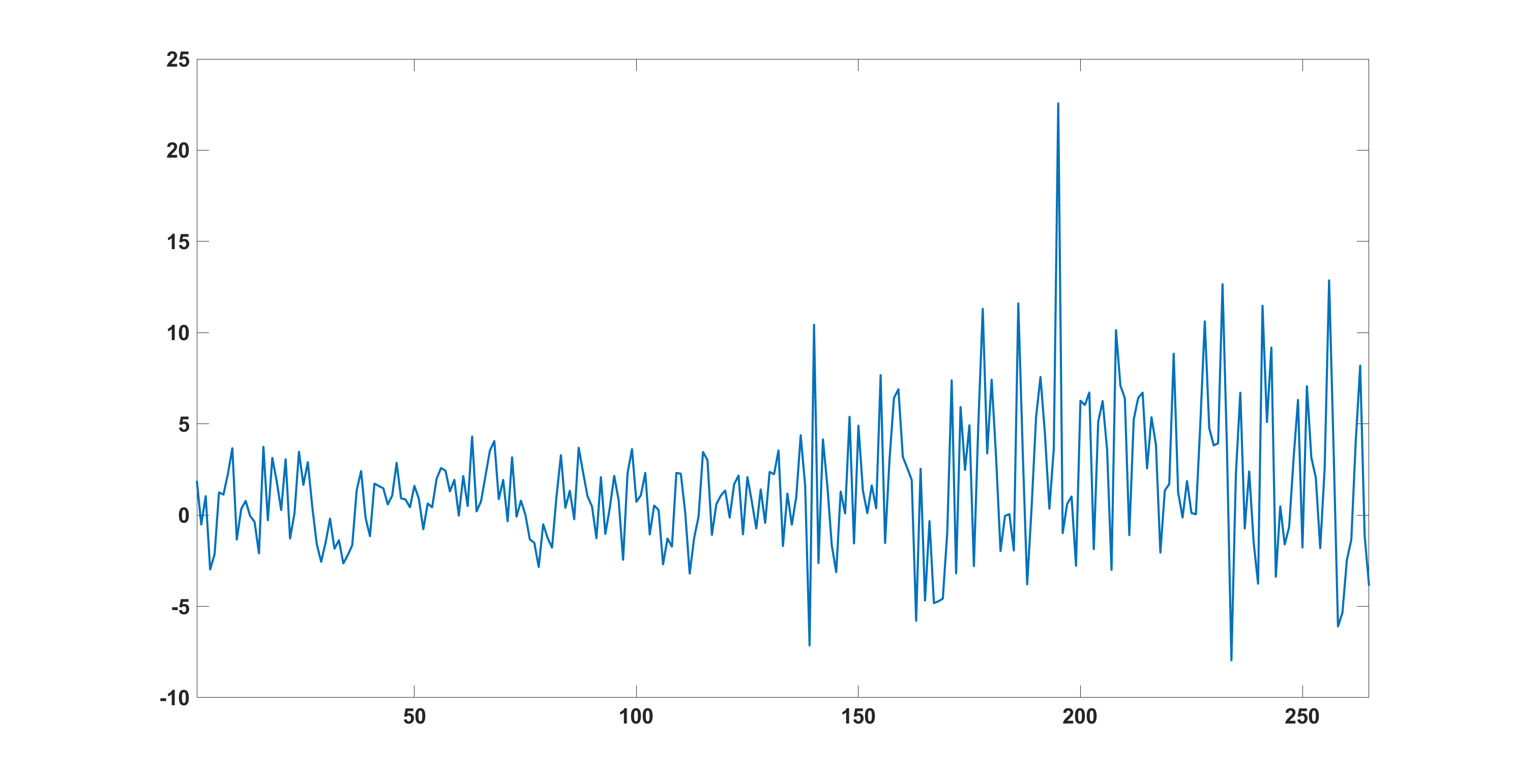}}\\
		\caption{\justifying One simulated series with a GARCH dynamic from the DGP based on the empirical data. '\# CPs' stands for the number of parameters that are experiencing a breakpoint at time $t=133$. \label{fig:DGPs_emp1}}
	\end{center}
\end{figure}

To carry out the Monte Carlo study, we have drawn 100 series from the 14 variants of the empirical DGP. For each simulated series, we have also generated the explanatory variables using the AR models. Table \ref{DGPEmpirical::MC} provides the percentage of the number of regimes detected by the selective segmentation and by the Lasso methods. The Lasso approach consists in using a Lasso penalty function instead of the SELO penalty function and in choosing the best Lasso penalty value and the corresponding model using the marginal likelihood proposed in Section \ref{sec:modelSelection}.\footnote{We replace the SELO with the Lasso penalty function. In particular, conditional on the potential breakpoints $\btau$, we minimize the following objective function:
\begin{align} \label{eq:optim1}
\begin{split}
\hat{\bbeta} & = \argmin_{\bbeta} ||\by- \bX_{\btau} \bbeta||_2^2 + \lambda \sum_{j=2}^{m} \sum_{k=1}^{K}|\Delta \beta_{jk}|.
\end{split}
\end{align}
Using the Matlab Lasso toolbox of \cite{mcilhagga2016penalized}. We also tested the matlab glmnet toolbox of \cite{LassoMatlab} which leads to similar results. In particular, we also observe that the Lasso estimates over-estimate the number of regimes.} First, we observe that the selective segmentation approach delivers high detection rates whatever the dynamic of the variance. In addition, the detection rate does not  deteriorate when the parameter experiences a break. Note also that all the average detections are above 89\%. In contrast, the Lasso method does not provide good results when only a subset of the parameters exhibits a CP. In fact, the average detection rates follow a 'U-shape' function meaning that the Lasso method is good at detecting no breakpoint or when all the parameters are experiencing a break. When partial breaks occur, the Lasso approach typically over-estimates the number of regimes.

\begin{table}[h!]
\centering
\singlespacing
\caption{\justifying \textbf{Break detection rates - Selective segmentation and Lasso approaches.}\\ 
Based on 100 replications, this Table assesses the break detection performance of the selective segmentation and the Lasso methods on the 14 variants of the empirical DGP detailed in Equation \eqref{DGP:Emp}. \textbf{Correct detection rate} is the rate of detecting the true number of regimes per model parameter. Underlined values correspond to the detection rates when the parameter experiences a breakpoint at $t=133$. \textbf{Avg.} documents the average rate of detecting the true number of regimes for each variant.\label{DGPEmpirical::MC}}
\scalebox{0.55}{
	\begin{tabular}{l cccc cccc cccc c|c}
		\toprule
\multicolumn{1}{c}{\underline{\textbf{DGP}}} & \multicolumn{13}{c}{\textbf{Correct detection rate}}  & \textbf{Avg.}\\	
\midrule
  & 	Inter & 	PMKT & 	SMB & 	TERM & 	DEF & 	SBD & 	SFX & 	SCOM & 	UMD & 	SIR & 	STK & 	CPI & 	NAREIT & 	Avg. detection \\
& \multicolumn{14}{c}{ \textbf{Selective segmentation - Constant variance}} \\														
\# CPs = 0 & 	100 & 	99 & 	99 & 	100 & 	98 & 	100 & 	100 & 	100 & 	100 & 	100 & 	99 & 	99 & 	100 & 	99.5 \\ 
\# CPs = 1 & 	\underline{47} & 	100 & 	98 & 	99 & 	100 & 	99 & 	98 & 	98 & 	98 & 	98 & 	97 & 	99 & 	95 & 	94.3 \\ 
\# CPs = 2 & 	\underline{98} & 	\underline{100} & 	98 & 	100 & 	96 & 	100 & 	100 & 	96 & 	99 & 	98 & 	97 & 	97 & 	96 & 	98.1 \\ 
\# CPs = 3 & 	\underline{98} & 	\underline{100} & 	\underline{88} & 	100 & 	96 & 	98 & 	100 & 	97 & 	98 & 	99 & 	94 & 	98 & 	98 & 	97.2 \\ 
\# CPs = 4 & 	\underline{99} & 	\underline{100} & 	\underline{97} & 	\underline{98} & 	98 & 	98 & 	97 & 	99 & 	100 & 	98 & 	96 & 	97 & 	98 & 	98.1 \\ 
\# CPs = 5 & 	\underline{98} & 	\underline{100} & 	\underline{93} & 	\underline{94} & 	\underline{98} & 	97 & 	100 & 	98 & 	98 & 	96 & 	98 & 	94 & 	95 & 	96.8 \\ 
\# CPs = 6 & 	\underline{99} & 	\underline{100} & 	\underline{95} & 	\underline{90} & 	\underline{100} & 	\underline{91} & 	97 & 	100 & 	98 & 	97 & 	93 & 	99 & 	95 & 	96.5 \\ 
\# CPs = 7 & 	\underline{97} & 	\underline{100} & 	\underline{89} & 	\underline{94} & 	\underline{99} & 	\underline{91} & 	\underline{97} & 	97 & 	100 & 	97 & 	97 & 	97 & 	99 & 	96.5 \\ 
\# CPs = 8 & 	\underline{95} & 	\underline{100} & 	\underline{91} & 	\underline{96} & 	\underline{100} & 	\underline{88} & 	\underline{97} & 	\underline{66} & 	98 & 	98 & 	98 & 	95 & 	96 & 	93.7 \\ 
\# CPs = 9 & 	\underline{99} & 	\underline{98} & 	\underline{95} & 	\underline{93} & 	\underline{99} & 	\underline{95} & 	\underline{97} & 	\underline{68} & 	\underline{100} & 	98 & 	99 & 	96 & 	100 & 	95.2 \\ 
\# CPs = 10 & 	\underline{97} & 	\underline{99} & 	\underline{94} & 	\underline{96} & 	\underline{100} & 	\underline{91} & 	\underline{97} & 	\underline{71} & 	\underline{100} & 	\underline{76} & 	96 & 	98 & 	93 & 	92.9 \\ 
\# CPs = 11 & 	\underline{96} & 	\underline{100} & 	\underline{90} & 	\underline{97} & 	\underline{100} & 	\underline{98} & 	\underline{98} & 	\underline{81} & 	\underline{100} & 	\underline{79} & 	\underline{95} & 	99 & 	98 & 	94.7 \\ 
\# CPs = 12 & 	\underline{93} & 	\underline{99} & 	\underline{90} & 	\underline{98} & 	\underline{98} & 	\underline{92} & 	\underline{96} & 	\underline{68} & 	\underline{100} & 	\underline{78} & 	\underline{95} & 	\underline{96} & 	98 & 	92.4 \\ 
\# CPs = 13 & 	\underline{94} & 	\underline{99} & 	\underline{93} & 	\underline{93} & 	\underline{98} & 	\underline{93} & 	\underline{96} & 	\underline{70} & 	\underline{99} & 	\underline{85} & 	\underline{96} & 	\underline{98} & 	\underline{92} & 	92.8 \\[0.2cm]
& \multicolumn{14}{c}{ \textbf{Lasso - Constant variance}} \\														
\# CPs = 0 & 	100 & 	100 & 	100 & 	100 & 	100 & 	100 & 	100 & 	100 & 	100 & 	100 & 	99 & 	100 & 	100 & 	99.9 \\ 
\# CPs = 1 & 	\underline{10} & 	99 & 	99 & 	100 & 	100 & 	93 & 	91 & 	91 & 	95 & 	91 & 	89 & 	100 & 	95 & 	88.7 \\ 
\# CPs = 2 & 	\underline{28} & 	\underline{99} & 	96 & 	99 & 	100 & 	77 & 	70 & 	77 & 	90 & 	59 & 	61 & 	100 & 	83 & 	79.9 \\ 
\# CPs = 3 & 	\underline{51} & 	\underline{100} & 	\underline{71} & 	100 & 	100 & 	52 & 	51 & 	55 & 	76 & 	47 & 	37 & 	100 & 	75 & 	70.4 \\ 
\# CPs = 4 & 	\underline{71} & 	\underline{100} & 	\underline{84} & 	\underline{52} & 	99 & 	37 & 	33 & 	35 & 	56 & 	30 & 	27 & 	97 & 	47 & 	59.1 \\ 
\# CPs = 5 & 	\underline{85} & 	\underline{92} & 	\underline{84} & 	\underline{78} & 	\underline{87} & 	20 & 	17 & 	20 & 	28 & 	20 & 	19 & 	89 & 	27 & 	51.2 \\ 
\# CPs = 6 & 	\underline{88} & 	\underline{91} & 	\underline{81} & 	\underline{75} & 	\underline{85} & 	\underline{86} & 	19 & 	18 & 	25 & 	15 & 	14 & 	92 & 	30 & 	55.3 \\ 
\# CPs = 7 & 	\underline{89} & 	\underline{83} & 	\underline{84} & 	\underline{80} & 	\underline{93} & 	\underline{78} & 	\underline{78} & 	14 & 	19 & 	14 & 	16 & 	90 & 	19 & 	58.2 \\ 
\# CPs = 8 & 	\underline{87} & 	\underline{86} & 	\underline{82} & 	\underline{81} & 	\underline{87} & 	\underline{81} & 	\underline{83} & 	\underline{77} & 	28 & 	18 & 	21 & 	92 & 	23 & 	65.1 \\ 
\# CPs = 9 & 	\underline{87} & 	\underline{89} & 	\underline{84} & 	\underline{84} & 	\underline{90} & 	\underline{84} & 	\underline{84} & 	\underline{83} & 	\underline{86} & 	19 & 	15 & 	92 & 	25 & 	70.9 \\ 
\# CPs = 10 & 	\underline{93} & 	\underline{88} & 	\underline{90} & 	\underline{82} & 	\underline{92} & 	\underline{86} & 	\underline{86} & 	\underline{82} & 	\underline{92} & 	\underline{81} & 	25 & 	92 & 	22 & 	77.8 \\ 
\# CPs = 11 & 	\underline{95} & 	\underline{87} & 	\underline{85} & 	\underline{84} & 	\underline{93} & 	\underline{83} & 	\underline{81} & 	\underline{85} & 	\underline{90} & 	\underline{85} & 	\underline{84} & 	97 & 	27 & 	82.8 \\ 
\# CPs = 12 & 	\underline{95} & 	\underline{86} & 	\underline{88} & 	\underline{93} & 	\underline{99} & 	\underline{86} & 	\underline{83} & 	\underline{84} & 	\underline{87} & 	\underline{83} & 	\underline{83} & 	\underline{94} & 	17 & 	82.9 \\ 
\# CPs = 13 & 	\underline{93} & 	\underline{85} & 	\underline{81} & 	\underline{90} & 	\underline{97} & 	\underline{74} & 	\underline{79} & 	\underline{77} & 	\underline{81} & 	\underline{77} & 	\underline{77} & 	\underline{94} & 	\underline{79} & 	83.4 \\[0.2cm]
& \multicolumn{14}{c}{ \textbf{Selective segmentation - GARCH variance}} \\														
\# CPs = 0 & 	100 & 	99 & 	99 & 	98 & 	98 & 	98 & 	98 & 	100 & 	99 & 	99 & 	98 & 	100 & 	98 & 	98.8 \\ 
\# CPs = 1 & 	\underline{54} & 	95 & 	96 & 	98 & 	98 & 	96 & 	96 & 	97 & 	94 & 	95 & 	96 & 	96 & 	98 & 	93.0 \\ 
\# CPs = 2 & 	\underline{99} & 	\underline{100} & 	95 & 	100 & 	97 & 	99 & 	98 & 	99 & 	99 & 	99 & 	97 & 	98 & 	99 & 	98.4 \\ 
\# CPs = 3 & 	\underline{99} & 	\underline{100} & 	\underline{93} & 	99 & 	99 & 	98 & 	99 & 	96 & 	97 & 	98 & 	92 & 	96 & 	100 & 	97.4 \\ 
\# CPs = 4 & 	\underline{98} & 	\underline{99} & 	\underline{93} & 	\underline{98} & 	99 & 	99 & 	96 & 	95 & 	98 & 	98 & 	97 & 	95 & 	95 & 	96.9 \\ 
\# CPs = 5 & 	\underline{96} & 	\underline{100} & 	\underline{92} & 	\underline{97} & 	\underline{100} & 	97 & 	96 & 	97 & 	97 & 	97 & 	94 & 	98 & 	98 & 	96.8 \\ 
\# CPs = 6 & 	\underline{95} & 	\underline{100} & 	\underline{87} & 	\underline{93} & 	\underline{99} & 	\underline{96} & 	96 & 	99 & 	99 & 	96 & 	99 & 	96 & 	96 & 	96.2 \\ 
\# CPs = 7 & 	\underline{97} & 	\underline{100} & 	\underline{92} & 	\underline{98} & 	\underline{100} & 	\underline{91} & 	\underline{92} & 	97 & 	99 & 	97 & 	98 & 	98 & 	100 & 	96.8 \\ 
\# CPs = 8 & 	\underline{98} & 	\underline{99} & 	\underline{90} & 	\underline{87} & 	\underline{99} & 	\underline{91} & 	\underline{92} & 	\underline{76} & 	98 & 	98 & 	98 & 	97 & 	93 & 	93.5 \\ 
\# CPs = 9 & 	\underline{98} & 	\underline{100} & 	\underline{89} & 	\underline{94} & 	\underline{100} & 	\underline{91} & 	\underline{95} & 	\underline{67} & 	\underline{100} & 	99 & 	98 & 	99 & 	99 & 	94.5 \\ 
\# CPs = 10 & 	\underline{96} & 	\underline{99} & 	\underline{96} & 	\underline{97} & 	\underline{98} & 	\underline{90} & 	\underline{97} & 	\underline{77} & 	\underline{100} & 	\underline{84} & 	96 & 	95 & 	99 & 	94.2 \\ 
\# CPs = 11 & 	\underline{94} & 	\underline{99} & 	\underline{93} & 	\underline{97} & 	\underline{100} & 	\underline{89} & 	\underline{93} & 	\underline{69} & 	\underline{99} & 	\underline{78} & 	\underline{94} & 	97 & 	97 & 	92.2 \\ 
\# CPs = 12 & 	\underline{91} & 	\underline{96} & 	\underline{89} & 	\underline{93} & 	\underline{99} & 	\underline{92} & 	\underline{97} & 	\underline{72} & 	\underline{99} & 	\underline{84} & 	\underline{94} & 	\underline{97} & 	100 & 	92.5 \\ 
\# CPs = 13 & 	\underline{85} & 	\underline{96} & 	\underline{91} & 	\underline{92} & 	\underline{100} & 	\underline{90} & 	\underline{96} & 	\underline{77} & 	\underline{99} & 	\underline{82} & 	\underline{86} & 	\underline{95} & 	\underline{77} & 	89.7 \\[0.2cm]
& \multicolumn{14}{c}{ \textbf{Lasso - GARCH variance}} \\														
\# CPs = 0 & 	100 & 	100 & 	100 & 	100 & 	100 & 	99 & 	99 & 	100 & 	100 & 	99 & 	98 & 	100 & 	100 & 	99.6 \\ 
\# CPs = 1 & 	\underline{20} & 	93 & 	98 & 	100 & 	99 & 	87 & 	84 & 	81 & 	88 & 	85 & 	79 & 	99 & 	94 & 	85.2 \\ 
\# CPs = 2 & 	\underline{26} & 	\underline{99} & 	97 & 	100 & 	100 & 	80 & 	82 & 	80 & 	89 & 	66 & 	63 & 	100 & 	90 & 	82.5 \\ 
\# CPs = 3 & 	\underline{49} & 	\underline{99} & 	\underline{64} & 	100 & 	100 & 	56 & 	52 & 	46 & 	76 & 	51 & 	48 & 	100 & 	77 & 	70.6 \\ 
\# CPs = 4 & 	\underline{68} & 	\underline{96} & 	\underline{75} & 	\underline{50} & 	99 & 	37 & 	38 & 	41 & 	51 & 	26 & 	28 & 	100 & 	51 & 	58.5 \\ 
\# CPs = 5 & 	\underline{85} & 	\underline{90} & 	\underline{82} & 	\underline{75} & 	\underline{84} & 	19 & 	19 & 	24 & 	25 & 	17 & 	15 & 	93 & 	29 & 	50.5 \\ 
\# CPs = 6 & 	\underline{86} & 	\underline{93} & 	\underline{87} & 	\underline{77} & 	\underline{86} & 	\underline{87} & 	19 & 	21 & 	21 & 	16 & 	17 & 	94 & 	23 & 	55.9 \\ 
\# CPs = 7 & 	\underline{89} & 	\underline{92} & 	\underline{85} & 	\underline{81} & 	\underline{88} & 	\underline{91} & 	\underline{89} & 	19 & 	30 & 	10 & 	21 & 	95 & 	25 & 	62.7 \\ 
\# CPs = 8 & 	\underline{87} & 	\underline{87} & 	\underline{83} & 	\underline{70} & 	\underline{87} & 	\underline{83} & 	\underline{86} & 	\underline{77} & 	23 & 	16 & 	21 & 	97 & 	26 & 	64.8 \\ 
\# CPs = 9 & 	\underline{88} & 	\underline{85} & 	\underline{86} & 	\underline{81} & 	\underline{90} & 	\underline{83} & 	\underline{85} & 	\underline{80} & 	\underline{87} & 	16 & 	15 & 	95 & 	29 & 	70.8 \\ 
\# CPs = 10 & 	\underline{97} & 	\underline{88} & 	\underline{88} & 	\underline{91} & 	\underline{96} & 	\underline{85} & 	\underline{85} & 	\underline{86} & 	\underline{91} & 	\underline{84} & 	9 & 	90 & 	28 & 	78.3 \\ 
\# CPs = 11 & 	\underline{92} & 	\underline{89} & 	\underline{91} & 	\underline{88} & 	\underline{93} & 	\underline{79} & 	\underline{81} & 	\underline{75} & 	\underline{88} & 	\underline{77} & 	\underline{78} & 	97 & 	18 & 	80.5 \\ 
\# CPs = 12 & 	\underline{93} & 	\underline{88} & 	\underline{84} & 	\underline{85} & 	\underline{91} & 	\underline{80} & 	\underline{83} & 	\underline{82} & 	\underline{86} & 	\underline{78} & 	\underline{80} & 	\underline{90} & 	26 & 	80.5 \\ 
\# CPs = 13 & 	\underline{92} & 	\underline{82} & 	\underline{85} & 	\underline{86} & 	\underline{92} & 	\underline{80} & 	\underline{82} & 	\underline{75} & 	\underline{85} & 	\underline{77} & 	\underline{80} & 	\underline{88} & 	\underline{81} & 	83.5 \\[0.2cm]
		\bottomrule
	\end{tabular}}
\end{table}%

To further illustrate the issue with the Lasso method, Table \ref{DGPEmpirical5::MC} shows the detailed results based on 100 simulated series when the first five parameters exhibit a CP (equivalent to the variant called '\# CPs = 5' in Table \ref{DGPEmpirical::MC}). While the selective segmentation method accurately detects the number of regimes for each parameter, the Lasso approach finds two regimes for most of the parameters that are constant over the sample.

\begin{table}[h!]
\centering
\singlespacing
\caption{\justifying \textbf{Empirical DGP with 5 CPs - Break detection rates of the Selective segmentation and the Lasso approaches.}\\ 
Based on 100 replications, this Table assesses the break detection performance of the selective segmentation and the Lasso methods on the fifth variant of the empirical DGP detailed in Equation \eqref{DGP:Emp}. \textbf{Number of regimes} is the rate of detecting a specific number of regimes per model parameter. Bold values correspond to the true number of regimes.\label{DGPEmpirical5::MC}}
\scalebox{0.6}{
	\begin{tabular}{l cccccc @{\hskip 1cm} cccccc @{\hskip 2cm} cccccc @{\hskip 1cm} cccccc }
		\toprule
		& \multicolumn{12}{c}{ \textbf{Constant variance}} &  \multicolumn{12}{c}{ \textbf{GARCH variance}}\\	
		& \multicolumn{6}{c}{ \textbf{Sel. segmentation}} 	& \multicolumn{6}{c}{ \textbf{Lasso}} & &\multicolumn{4}{c}{ \textbf{Sel. segmentation}}& 	& &\multicolumn{4}{c}{ \textbf{Lasso}}&\\	
		\# of regimes & 1 & 2 & 3 & 4 & 5 & 6  & 1 & 2 & 3 & 4 & 5 & 6 & 1 & 2 & 3 & 4 & 5 & 6  & 1 & 2 & 3 & 4 & 5 & 6 \\
\midrule
\hline																									
Inter. & 	2 & 	\textbf{98} & 	0 & 	0 & 	0 & 	0 & 	12 & 	\textbf{85} & 	3 & 	0 & 	0 & 	0 & 	3 & 	\textbf{96} & 	1 & 	0 & 	0 & 	0 & 	15 & 	\textbf{85} & 	0 & 	0 & 	0 & 	0\\
PMKT & 	0 & 	\textbf{100} & 	0 & 	0 & 	0 & 	0 & 	0 & 	\textbf{92} & 	8 & 	0 & 	0 & 	0 & 	0 & 	\textbf{100} & 	0 & 	0 & 	0 & 	0 & 	0 & 	\textbf{90} & 	10 & 	0 & 	0 & 	0\\
SMB & 	5 & 	\textbf{93} & 	2 & 	0 & 	0 & 	0 & 	10 & 	\textbf{84} & 	6 & 	0 & 	0 & 	0 & 	6 & 	\textbf{92} & 	2 & 	0 & 	0 & 	0 & 	9 & 	\textbf{82} & 	9 & 	0 & 	0 & 	0\\
TERM & 	6 & 	\textbf{94} & 	0 & 	0 & 	0 & 	0 & 	22 & 	\textbf{78} & 	0 & 	0 & 	0 & 	0 & 	2 & 	\textbf{97} & 	1 & 	0 & 	0 & 	0 & 	25 & 	\textbf{75} & 	0 & 	0 & 	0 & 	0\\
DEF & 	0 & 	\textbf{98} & 	1 & 	1 & 	0 & 	0 & 	12 & 	\textbf{87} & 	1 & 	0 & 	0 & 	0 & 	0 & 	\textbf{100} & 	0 & 	0 & 	0 & 	0 & 	16 & 	\textbf{84} & 	0 & 	0 & 	0 & 	0\\
SBD & 	\textbf{97} & 	3 & 	0 & 	0 & 	0 & 	0 & 	\textbf{20} & 	70 & 	10 & 	0 & 	0 & 	0 & 	\textbf{97} & 	2 & 	1 & 	0 & 	0 & 	0 & 	\textbf{19} & 	76 & 	4 & 	1 & 	0 & 	0\\
SFX & 	\textbf{100} & 	0 & 	0 & 	0 & 	0 & 	0 & 	\textbf{17} & 	73 & 	10 & 	0 & 	0 & 	0 & 	\textbf{96} & 	4 & 	0 & 	0 & 	0 & 	0 & 	\textbf{19} & 	70 & 	10 & 	1 & 	0 & 	0\\
SCOM & 	\textbf{98} & 	2 & 	0 & 	0 & 	0 & 	0 & 	\textbf{20} & 	71 & 	9 & 	0 & 	0 & 	0 & 	\textbf{97} & 	3 & 	0 & 	0 & 	0 & 	0 & 	\textbf{24} & 	69 & 	7 & 	0 & 	0 & 	0\\
UMD & 	\textbf{98} & 	2 & 	0 & 	0 & 	0 & 	0 & 	\textbf{28} & 	66 & 	6 & 	0 & 	0 & 	0 & 	\textbf{97} & 	2 & 	1 & 	0 & 	0 & 	0 & 	\textbf{25} & 	68 & 	6 & 	1 & 	0 & 	0\\
SIR & 	\textbf{96} & 	3 & 	1 & 	0 & 	0 & 	0 & 	\textbf{20} & 	69 & 	11 & 	0 & 	0 & 	0 & 	\textbf{97} & 	2 & 	1 & 	0 & 	0 & 	0 & 	\textbf{17} & 	71 & 	11 & 	1 & 	0 & 	0\\
STK & 	\textbf{98} & 	2 & 	0 & 	0 & 	0 & 	0 & 	\textbf{19} & 	71 & 	10 & 	0 & 	0 & 	0 & 	\textbf{94} & 	6 & 	0 & 	0 & 	0 & 	0 & 	\textbf{15} & 	77 & 	7 & 	1 & 	0 & 	0\\
CPI & 	\textbf{94} & 	6 & 	0 & 	0 & 	0 & 	0 & 	\textbf{89} & 	11 & 	0 & 	0 & 	0 & 	0 & 	\textbf{98} & 	2 & 	0 & 	0 & 	0 & 	0 & 	\textbf{93} & 	7 & 	0 & 	0 & 	0 & 	0\\
NAREIT & 	\textbf{95} & 	4 & 	1 & 	0 & 	0 & 	0 & 	\textbf{27} & 	67 & 	6 & 	0 & 	0 & 	0 & 	\textbf{98} & 	1 & 	1 & 	0 & 	0 & 	0 & 	\textbf{29} & 	68 & 	3 & 	0 & 	0 & 	0\\
		\bottomrule
	\end{tabular}}
\end{table}%

We end this simulation section with a "big data" example motivated by the fact that when the number of explanatory variables is large, the current Bayesian alternatives do not work \citep[see][]{GiordaniKohn2008,Eo2012,huber2019should,dufays2019relevant} (see SA \ref{App:Bayesian} for more details). To do so, we propose the DGP J that is specified by 100 explanatory variables and one CP as follows:

\vspace{0.2cm}
\noindent \textbf{DGP J: } piecewise linear model with big data\\
	$$ Y_{t} = \begin{cases} \mathbf{x}'_t\boldsymbol{\beta_1} + \varepsilon_t & \quad\text{if } 1\leq t \leq 499, \\ 
	\mathbf{x}'_t\boldsymbol{\beta_2} + \varepsilon_t & \quad\text{if } 500\leq t \leq T,\\
	\end{cases} 
	$$
	
\noindent where $T=1024$, $\forall$ $t\in [1,T]$ and for $i=1,...,100$, $\mathbf{x}_{t,i} \sim \NORM{\left( 0,1\right)}$ and $\varepsilon_t \sim \NORM(0,1)$. The parameter values of $\boldsymbol{\beta_1}$ are uniformly and randomly set to $-1$ or $1$. In the second regime, the parameter values of $\boldsymbol{\beta_2}$ are equal to $\boldsymbol{\beta_1}$ except for 10 of them randomly chosen that are set to the opposite value (i.e. $-\bbeta_1$). Thus, 10 parameters of DGP J does experience a break at observation 500. \\
We simulate 100 series from DGP J to assess the SELO performance in detecting which parameters experience a breakpoint. For every simulation, the selective segmentation approach identifies 10 parameters that experience one breakpoint in the sample while the others remain constant. In addition, the exact model specification was always among the specification exhibiting a posterior probability of at least 10\%.


\section{Empirical application} \label{s:empirics}
We illustrate the selective segmentation method with 14 monthly Credit Suisse HF indices spanning from March 1994 to March 2016. These indices are the weighted average of HF returns following specific trading strategies. \cite{fung2004hedge} suggested a risk-based approach to model HF returns and identified seven factors on which HF strategies are generally exposed \cite[see also][]{fung2001risk}. Since this seminal work, many other risk factors have been uncovered. So, we include five other risk factors that are also popular in the literature. We add two Fung and Hsieh trend following risk factors, PTFSIR, returns on PTFS  short term interest rate lookback straddle, and PTFSSTK, returns on PTFS stock index lookback straddle. Following \cite{ Agarwal2004Risks}, among many others, we also use the Up-minus-Down (UMD) factor \citep[see][]{Carhart-JF-1997}. As suggested by \cite{chen1986economic}, we include the expected inflation, the log relative of US Consumer Price Index (CPI). Finally, we also take into account a factor for real estate risk relevant to explain HF and stocks returns \citep[see, e.g.,][]{Ambrose-DLima-JREFE-2016,Carmichael-Coen-REE-2018}. Table \ref{tab:CS_factors} documents the fourteen strategies on which we focus as well as the twelve factors. \\

\begin{table}[h!]
\footnotesize
	\centering
	\caption{\justifying Description of the HF returns and the risk factors. The column 'Paper' highlights a paper in which the factor has already been used. FH, C, CRR and AD refer to \cite{fung2004hedge}, \cite{Carhart-JF-1997}, \cite{chen1986economic} and \cite{Ambrose-DLima-JREFE-2016}, respectively.  \label{tab:CS_factors}}
\begin{tabular}{ll lll}
 \toprule	
\multicolumn{2}{c}{Credit Suisse Hedge fund indices} & \multicolumn{3}{c}{Risk factors }\\
\cmidrule(lr){1-2}\cmidrule(lr){3-5}	
Name & Description & Name & Description & Paper \\[0.2cm]
    HFI & Hedge Fund Index  & PMKT & Market factor (S\&P 500) & FH\\
    CNV & Convertible Arbitrage & SMB & Small firm minus big firm & FH\\
    DSB & Dedicated Short Bias  & TERM & Change in 10-year treasury yields & FH\\
    EME & Emerging Markets  & DEF & Change in the yield spread of   & FH\\
    EMN & Equity Market Neutral & & 10-year treasury and Moody's Baa bonds \\
    EDR & Event Driven  & PTFSBD & Lookback options on Bonds & FH\\
    EDD & Event Driven Distressed &PTFSFX & Lookback options on currencies & FH\\
    EDM & Event Driven Multi-Strategy &  PTFSCOM & Lookback options on commodities & FH\\
    EDRA& Event Driven Risk Arbitrage & UMD & Momentum (Up-minus-Down) & C  \\
    FIA & Fixed Income Arbitrage & PTFSIR & Lookback options on short term interest rate & FH \\
    GMA & Global Macro & PTFSSTK& Lookback options on Stock Index & FH \\
    LES & Long/Short Equity & CPI& Consumer price index & CRR \\
    MFU & Managed Futures & NAREIT & Real estate investment trust index & AD \\
    MUS & Multi-Strategy & & & \\
\bottomrule
	\end{tabular}%
\end{table}

It is well acknowledged in the financial literature that HF strategies (or trading techniques) are time-varying.  Their changing risk exposures are directly related to market events and economic fluctuations (see, e.g., \cite{ Agarwal2004Risks}, \cite{Fung2008hedgefunds} or \cite{Patton2015change} among others). Hedge fund time-varying risk dynamics has important implications for performance appraisal. As pointed out by \cite{Mitchell2001characteristics}, the changes can be in response to arbitrage opportunities. The cycles of mergers and acquisitions in the 1990s and the 2000s and the corresponding level of risk arbitrage led by HF are illustrations of these changing dynamics. In standard linear asset pricing models, the intercept and risk factor loadings are not constant but time-varying. Moreover, HF returns  exhibit significant non-linearities. Therefore, there is a need of dynamic models able to capture  non-linearities and changes in risk exposures.\\
Following \cite{meligkotsidou2008detecting}, we suggest the use of CP risk factor models. This class of models is suited for studying the changes in risk exposures and their time-varying parameters. However, instead of directly focusing on the twelve factors, we take a slightly different approach since we additionally take into account autocorrelations of the returns.\footnote{As reported by \cite{getmansky2004econometric}, the analysis of serial dependence of returns is a reasonable way of assessing the liquidity of hedge fund investments.} To do so, we first look at the best autoregressive model that fits the returns. In particular, for each HF returns, we estimate ARX(q) models with $q$ ranging from 0 to 4 and in which the explanatory variables are the twelve factors (and an intercept) and we select the best AR order using the Bayesian information criterion (BIC). Table \ref{JBF::order} documents the best order for each strategy.
\begin{table}[h!]
\small
\centering
	\caption{\justifying \textbf{Order of the optimal ARX-model for each HF strategy.} \\
	The optimal AR order is chosen by maximizing the Bayesian information criterion over the whole sample. When looking for the best autoregressive lag order, the explanatory variables include the seven factors and an intercept.}
\begin{tabular}{lccccccc}
\toprule
Strat. & 	HFI & 	CNV & 	DSB & 	EME & 	EMN & 	EDR & 	EDD \\
Lag order & 	0 & 	1 & 	0 & 	1 & 	0 & 	1 & 	2 \\[0.1cm]
Strat. & EDM & 	EDRA & 	FIA & 	GMA & 	LES & 	MFU & 	MUS \\
Lag order & 	1 & 	1 & 	1 & 	0 & 	1 & 	0 & 	0 \\
\bottomrule
\end{tabular}
\label{JBF::order}
\end{table}

As reported by \cite{fung2004hedge}, composites obtained from the individual funds may be contaminated with severe survivorship, selection and instant history biases. Therefore, to avoid these problems,  we use the Credit Suisse indices that provide full transparency about their constituents.

Section \ref{sec:in} discusses in-sample results of our selective segmentation method and we compare them to those of standard CP models and time-varying parameter models. We then illustrate the difference of our approach with the CP method of \cite{meligkotsidou2008detecting} in Section \ref{sec:CPJBF}. Section \ref{sec:oos} documents a forecasting exercise in which we assess the predictive performance of the selective segmentation approach with respect to flexible alternatives. Importantly, all the subsequent results include the optimal AR order documented in Table \ref{JBF::order} as additional explanatory variables.

\subsection{Hedge funds strategies evolve over time \label{sec:in}}
\cite{fung2004hedge} focus on linear models. However, as the period covers critical events such as the Long Term Capital Management (LTCM) collapse, the dot-com crisis and the global financial crisis (GFC), one could argue that CP models are more appropriate. In this Section, we focus on two specific indices, namely the Hedge Fund Index (HFI) and the HF returns that are applying a Fixed-Income Arbitrage (FIA) strategy. Results for all the other returns are available upon request. \\
Tables \ref{tab:OLS_SELO1} and \ref{tab:OLS_SELO10} show how the selective segmentation method can improve the interpretation of CP models. The Tables document how the results evolve from a standard linear risk model to a selective segmentation model passing by a standard CP process. As expected, for the two HF returns, ignoring breakpoints can be misleading as the CP results emphasize that they modify the risk exposition of the returns. Also, although one can study in details the results of the standard CP model, the selective segmentation model offers a straightforward picture of the relevant risk factors and how the risk exposition evolves. It also estimates more accurately the parameters that do not change when a break occurs. As the CP model detects three breakpoints for the HFI and six abrupt changes for the FIA strategy, the number of models to consider amounts to $2^{36}$ and $2^{84}$ respectively. Our selective segmentation strategy explores these large model spaces and find the most promising configurations in several minutes on a standard laptop. Let us now discuss in more details the results of the two returns.\\

\noindent \textbf{Hedge Fund Index}\\
As documented in Table \ref{tab:OLS_SELO1}, the CP model with breakpoints determined by the approach in Section \ref{sec:breakYau} finds four regimes (hereafter CP-YZ). The relevant breakpoints occur in April 2000, in December 2001 and in August 2014. Interestingly, these dates coincide with the dot-com crash that spanned from March 2000 to October 2002 and when  stocks suffer steepest drop in 2014 (introducing the first significant stock market scares after the GFC). It is well acknowledged in the financial literature that the end of the dot-com bubble had important consequences for financial markets in the early 2000s. While all the parameters change for the CP model, the selective segmentation mainly identifies that the factors related to the breaks are the market factor (PMKT),  the default risk factor (DEF), the momentum risk factor (UMD) and the real estate risk factor (NAREIT). Moreover, it discards two spurious breaks occurring in December 2001  and in August 2014 making the model even more parsimonious. We can notice that the market factor decreases from 0.34 during the first period to 0.21 during the second period. HFI is indeed more conservative during the 2000s and less correlated with the financial markets. We observe the same trend for the default risk factor increasing from -9.15 to -2.42. The momentum factor, UMD, sharply declines after the nineties known as very volatile as reported by \cite{Shiller-PUP-2015} and documented by \cite{Campbell-JF-2000} who both highlight the bullish market of  this decade. The momentum is still significant since 2000 but its impact has significantly weakened. It decreases from 0.22 to 0.05. The real estate risk factor also exhibits a breakdown in the early 2000s. It is indeed well acknowledged in the real estate economics literature that the 1990s are considered as the new era of real estate investment trusts (REITs) and the period beginning in the early 2000s as the maturity REITs era (\cite{Pagliari-Scherer-Monopoli-REE-2005}, \cite{Ambrose-Lee-Peek-REE-2007} and \cite{Carmichael-Coen-REE-2018} among others). From the new REITs era to the maturity REITs era, the real estate risk factor declines from 0.15 to -0.01. These results are consistent with the important variations of interest rates during these two sub-periods and the important increase of credit risk in the 2000s. As a final note, the credible interval of the breakpoint is narrow which indicates a sharp change in the risk exposition in March 2000 (see also Figure \ref{fig:strat1} below). \\

\begin{table}[h!]
\centering
\singlespacing
\caption{\justifying \textbf{Hedge Fund Index: linear, CP and selective segmentation regression models.}\\ 
The Table details the parameter estimates of the linear model, of the CP model and of the selective segmentation process with HFI returns as the dependent variable. Parentheses and brackets indicate standard deviations and 90\% credible intervals, respectively. A cell filled with '---' indicates that the parameter does not vary over the related period. The posterior probability of the selective segmentation model amounts to 77\%.}
\scalebox{0.6}{
\begin{tabular}{lccccccccccccc}
\toprule
Period & 	Int. & 	PMKT & 	SMB & 	TERM & 	DEF & 	PTFSBD & 	PTFSFX & 	PTFSCOM & UMD & PTFSIR & PTFSSTK & CPI & NAREIT\\ 
\midrule
 & \multicolumn{13}{c}{\textbf{Standard linear risk model}} \\													
  \cmidrule(lr){2-14}													
03-1994 to 03-2016 & 	0.33 & 	0.27 & 	0.07 & 	-0.86 & 	-3.03 & 	-0.01 & 	0.01 & 	0.00 & 	0.11 & 	-0.00 & 	0.02 & 	1.07 & 	-0.01 \\ 
 & 	(0.10) & 	(0.02) & 	(0.03) & 	(0.45) & 	(0.59) & 	(0.01) & 	(0.00) & 	(0.01) & 	(0.02) & 	(0.00) & 	(0.01) & 	(0.32) & 	(0.02) \\[0.2cm] 
 & \multicolumn{13}{c}{\textbf{CP-YZ risk model}} \\													
  \cmidrule(lr){2-14}													
03-1994 to 04-2000 & 	0.59 & 	0.34 & 	0.03 & 	-1.46 & 	-10.24 & 	-0.02 & 	0.02 & 	0.02 & 	0.22 & 	-0.02 & 	0.04 & 	0.26 & 	0.16 \\ 
 & 	(0.37) & 	(0.05) & 	(0.05) & 	(1.35) & 	(2.52) & 	(0.01) & 	(0.01) & 	(0.02) & 	(0.06) & 	(0.01) & 	(0.02) & 	(1.56) & 	(0.07) \\ 
05-2000 to 12-2001 & 	-0.12 & 	0.17 & 	0.10 & 	-1.19 & 	-0.54 & 	0.00 & 	0.03 & 	-0.04 & 	0.11 & 	0.03 & 	-0.03 & 	0.32 & 	0.00 \\ 
 & 	(0.84) & 	(0.10) & 	(0.15) & 	(3.17) & 	(6.08) & 	(0.03) & 	(0.04) & 	(0.10) & 	(0.11) & 	(0.04) & 	(0.06) & 	(2.48) & 	(0.13) \\ 
01-2002 to 08-2014 & 	0.22 & 	0.22 & 	-0.00 & 	-0.71 & 	-2.17 & 	-0.01 & 	0.01 & 	0.00 & 	0.05 & 	-0.01 & 	0.02 & 	1.21 & 	-0.01 \\ 
 & 	(0.17) & 	(0.05) & 	(0.06) & 	(0.76) & 	(0.83) & 	(0.01) & 	(0.01) & 	(0.01) & 	(0.03) & 	(0.01) & 	(0.01) & 	(0.45) & 	(0.03) \\ 
09-2014 to 03-2016 & 	-0.40 & 	0.39 & 	-0.06 & 	-1.95 & 	-3.35 & 	0.01 & 	0.03 & 	-0.02 & 	0.19 & 	-0.01 & 	-0.00 & 	-0.93 & 	-0.16 \\ 
 & 	(0.68) & 	(0.21) & 	(0.19) & 	(5.53) & 	(5.64) & 	(0.04) & 	(0.03) & 	(0.03) & 	(0.18) & 	(0.03) & 	(0.03) & 	(1.87) & 	(0.20) \\ 
 & \multicolumn{13}{c}{\textbf{Selective segmentation risk model (77\%) }} \\													
  \cmidrule(lr){2-14}													
03-1994 to 04-2000 & 	0.24 & 	0.34 & 	0.04 & 	-1.05 & 	-9.15 & 	-0.01 & 	0.01 & 	0.01 & 	0.22 & 	-0.01 & 	0.02 & 	1.20 & 	0.15 \\ 
~[02-2000~~05-2000] & 	(0.09) & 	(0.04) & 	(0.02) & 	(0.42) & 	(1.35) & 	(0.01) & 	(0.00) & 	(0.01) & 	(0.04) & 	(0.00) & 	(0.01) & 	(0.29) & 	(0.05) \\ 
05-2000 to 03-2016 & 	 ---  & 	0.21 & 	 ---  & 	 ---  & 	-2.42 & 	 ---  & 	 ---  & 	 ---  & 	0.05 & 	 ---  & 	 ---  & 	 ---  & 	-0.01 \\ 
 & 	 ---  & 	(0.03) & 	 ---  & 	 ---  & 	(0.54) & 	 ---  & 	 ---  & 	 ---  & 	(0.02) & 	 ---  & 	 ---  & 	 ---  & 	(0.02) \\ 
\bottomrule
\end{tabular}}
\label{tab:OLS_SELO1}
\end{table}

Figure \ref{fig:strat1} shows the posterior medians over time and their corresponding credible intervals of the parameters related to the MKT, DEF and UMD factors given by our method (see Section \ref{sec:breakuncertainty} for the related Bayesian model and how the breakpoints are integrated out) and the time-varying parameter (TVP) model (see SA \ref{App:TVP} for the model specification). As with the CP model, one can easier interpret the time-varying dynamics of the parameters given by the selective segmentation method than those of the TVP model. For instance, while the exposition to the default factor seems fixed over the sample due to the smooth transition of the parameter, it is clear that the exposition is changing when we look at the selective segmentation results. Regarding the market factor, we also observe with the TVP model that the exposition seems different before and after the dot-com crash but the credible intervals are too wide to confirm the statement.

\renewcommand{\baselinestretch}{1}
\begin{figure}[h!]
	\begin{center}
		\subfloat[SELO - PMKT]{\includegraphics[width=6.5cm,height=4.5cm]{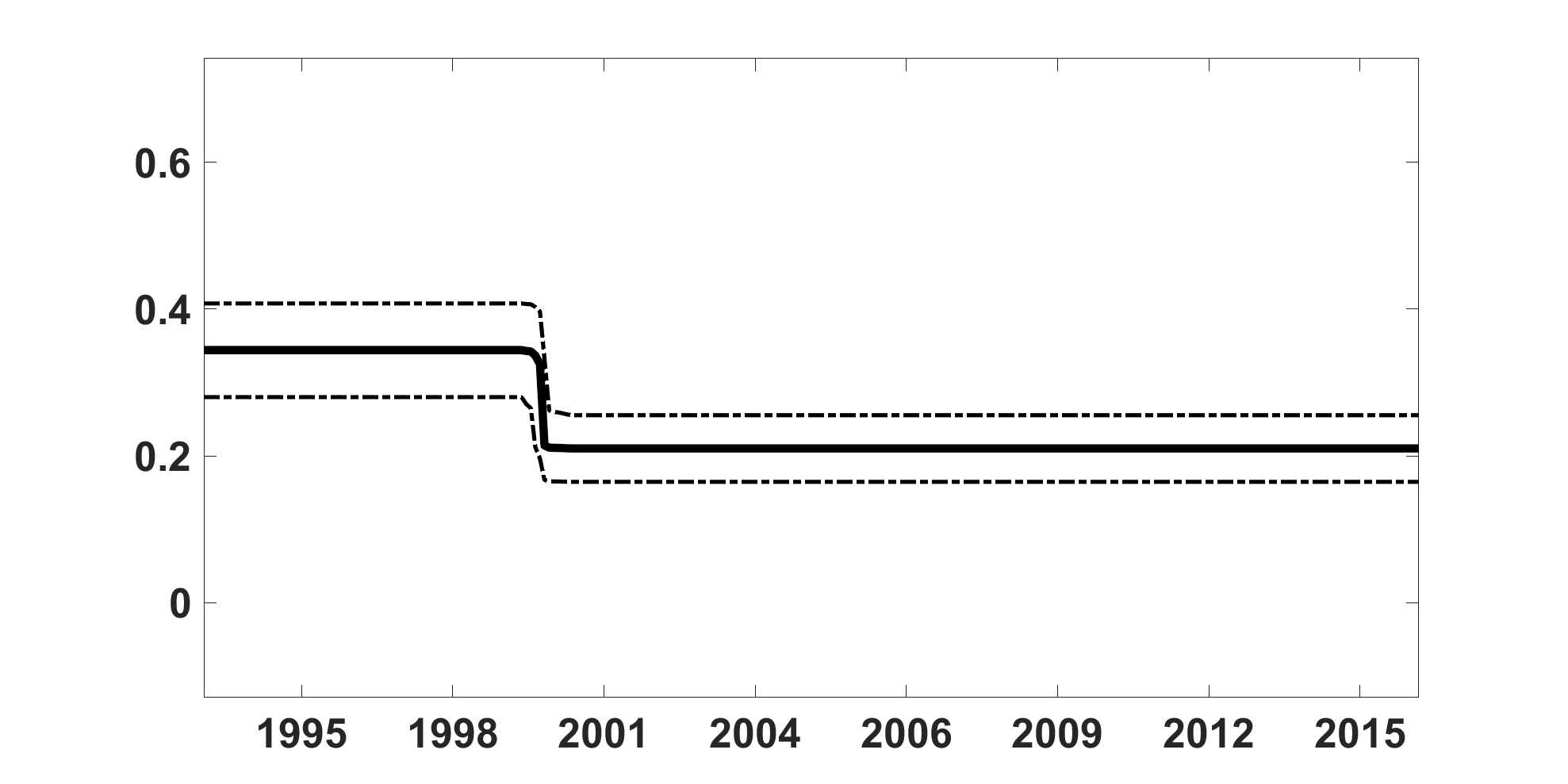}}
		\subfloat[TVP - PMKT]{\includegraphics[width=6.5cm,height=4.5cm]{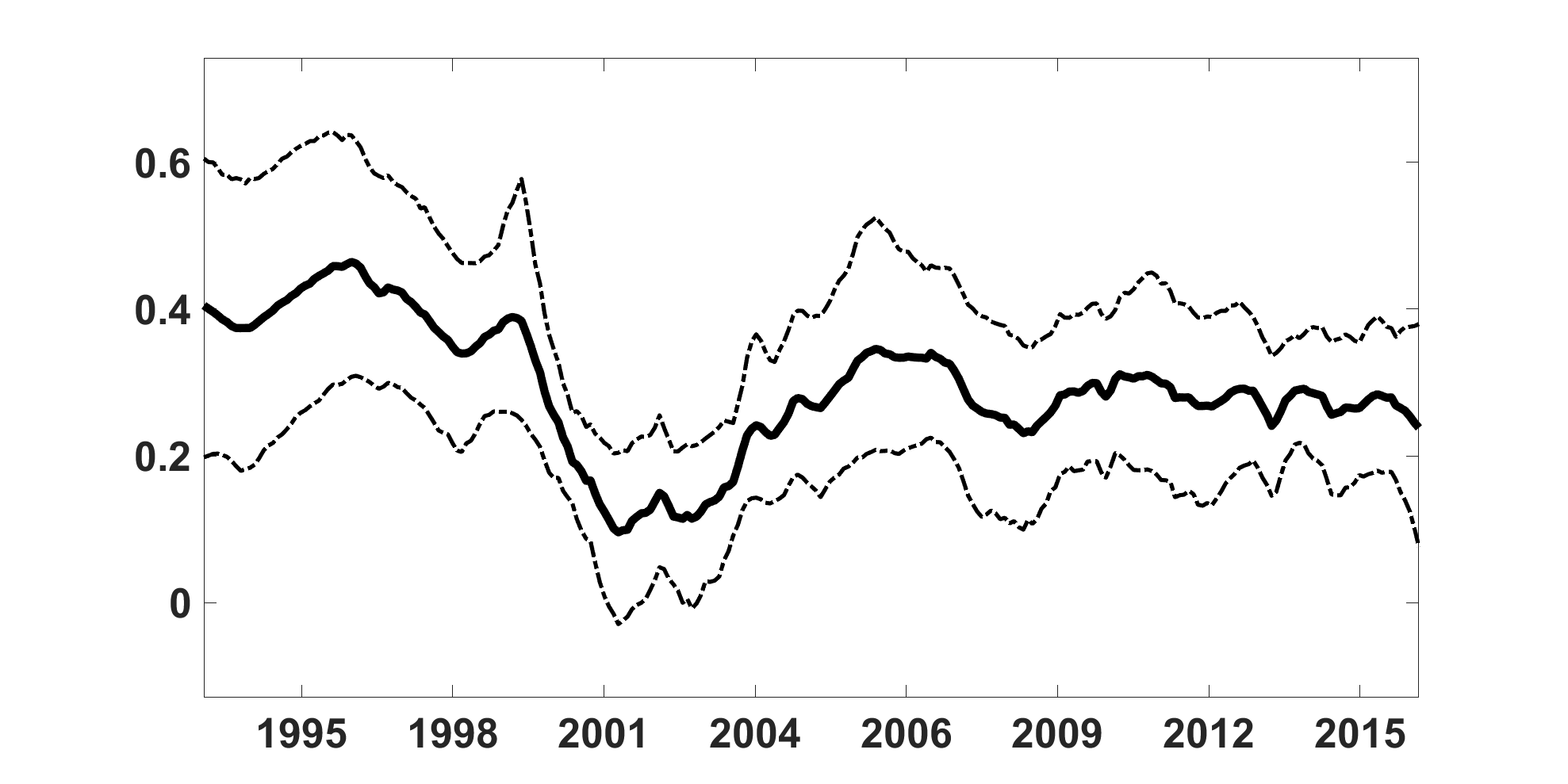}}\\
		\subfloat[SELO - DEF]{\includegraphics[width=6.5cm,height=4.5cm]{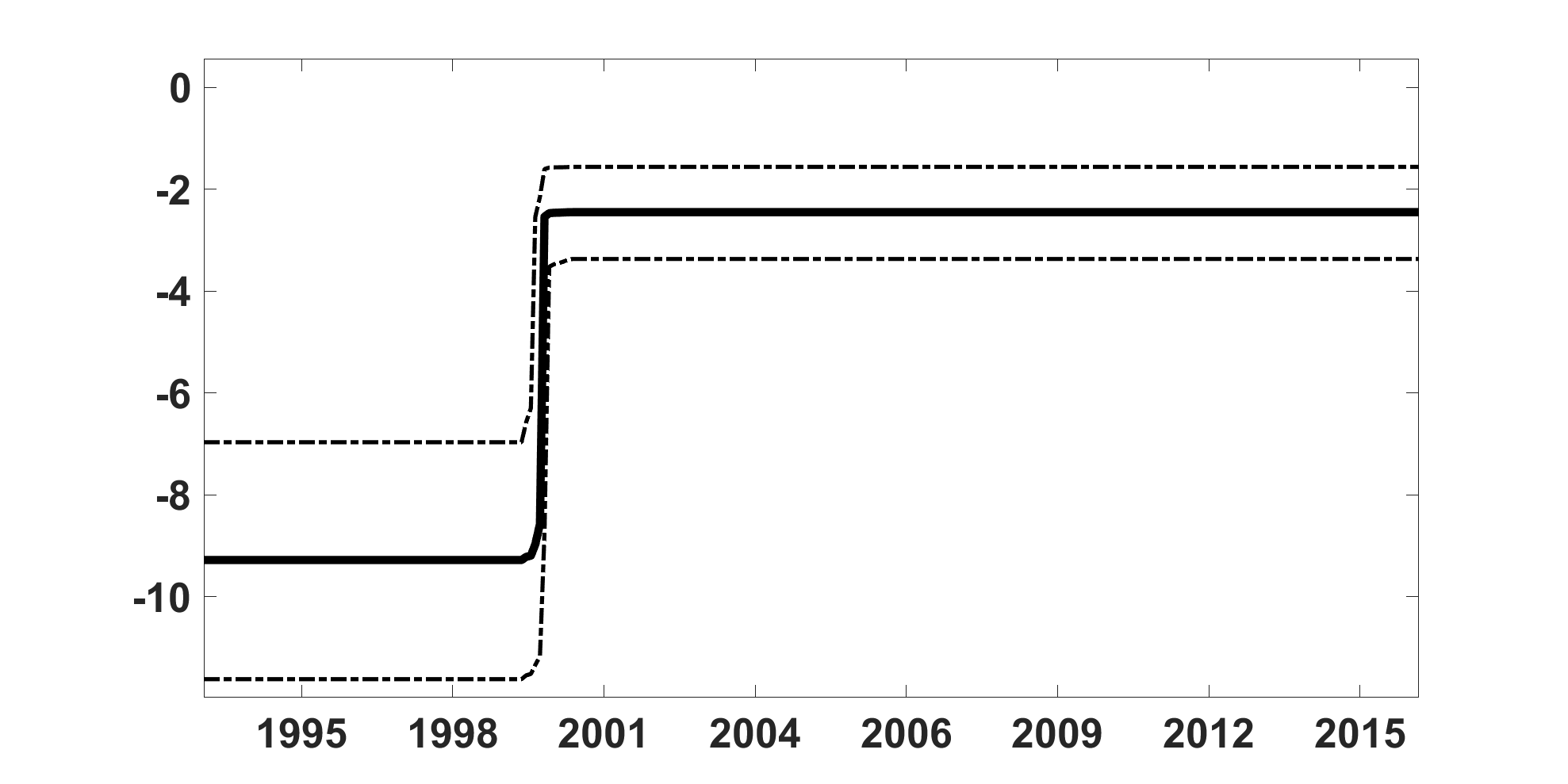}}
		\subfloat[TVP - DEF]{\includegraphics[width=6.5cm,height=4.5cm]{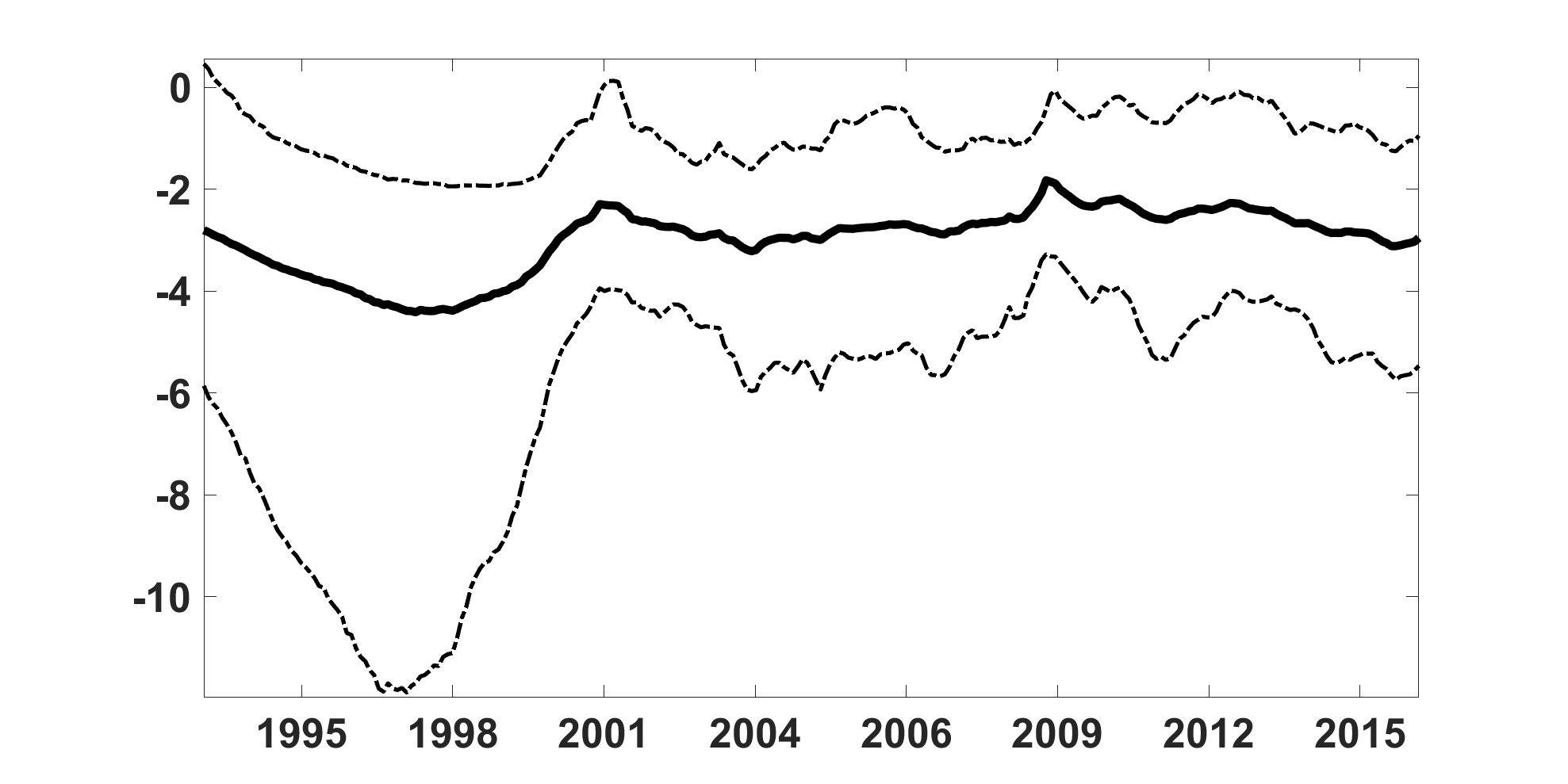}}\\
		\subfloat[SELO - UMD]{\includegraphics[width=6.5cm,height=4.5cm]{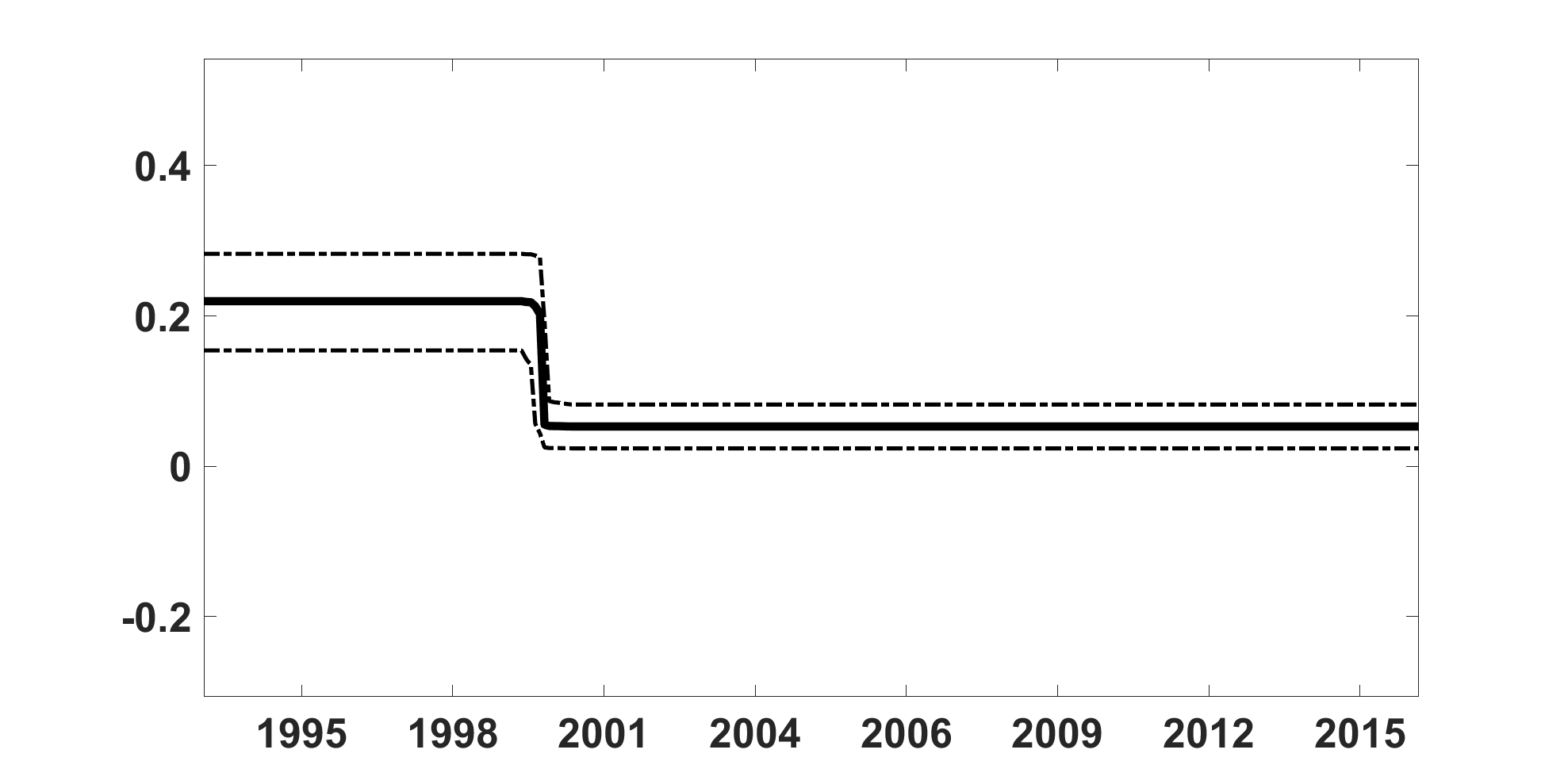}}
		\subfloat[TVP - UMD]{\includegraphics[width=6.5cm,height=4.5cm]{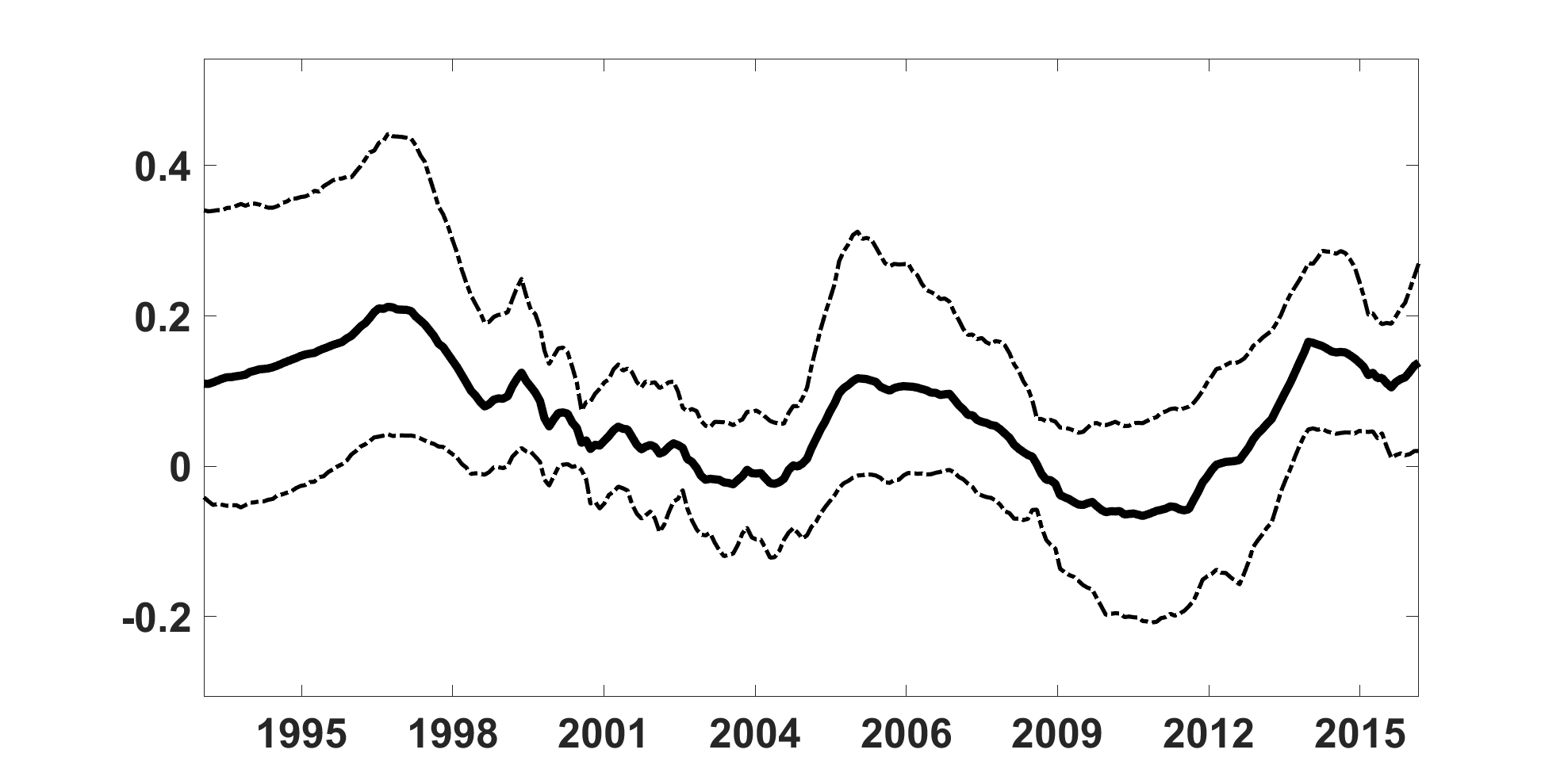}}\\
		\caption{\justifying \textbf{HFI returns - Selective segmentation (SELO) model and Time-varying parameter (TVP) model.} Posterior medians (black) and the 90\% credible intervals (dotted black lines) of the model parameters over time. For the SELO method, we take the break uncertainty into account using the MCMC algorithm presented in Section \ref{sec:breakuncertainty}. \label{fig:strat1}}
	\end{center}
\end{figure}
\renewcommand{\baselinestretch}{1.5}

\clearpage
\noindent \textbf{Fixed Income Arbitrage (FIA) strategy}\\
The FIA strategy is based on the exploitation of inefficiencies in the pricing of bonds and interest rate derivatives (including futures, options, swaps and also mortgage back securities). It was very appreciated among hedge fund managers until the collapse of the LTCM fund in September 1998. After this incident, a change of behavior among managers has been observed for this strategy on financial markets. \\
The results of the FIA returns from a standard linear regression to the selective segmentation model are documented in Table \ref{tab:OLS_SELO10}. Focusing on the latter model, breakpoints are detected in September 1997, in March 1999, August 2000 which are related to the Russia financial crisis and the dot-com crash. In addition, three other breakpoints capture the financial crisis: the beginning of the global financial crisis in late June 2007, the turmoil of September and October 2008 after the collapse of Lehman Brothers and the sovereign debt crisis in the euro zone. The selective segmentation specification highlights the role played by the market factor (i.e., PMKT, before and after the global financial crisis with estimates of 0.02 and 0.08 respectively), the variation of the size effect, the default risk factor, the momentum factor, the inflation factor and four trend following risk factors, especially during the sub-prime crisis from late June 2007 to October 2008.  The size effect changes from 0.02 before the crisis to -0.05 afterwards. The bond trend following factor, PTFSBD, the currency trend following factor, PTFSFX, the commodity trend following factor, PTFSCOM and the stock index trend following factor, PTFSSTK are highly significant during the GFC. The credit spread factor, DEF, is highly significant and constant after the dot com crisis (estimate amounting to -1.18). Significant during the first and the second periods related to the bullish market of the 1990s, the momentum effect, UMD, dramatically changes during the GFC, with a negative estimate of -0.34 (as expected). Our results also report the impact of inflation risk, CPI, during the crisis (with an estimate of 2.69) until the sovereign debt crisis (fifth and sixth periods). The real estate risk factor, NAREIT, is significant and positive during the 1990s (first, second and third periods) and during the GFC, with a negative estimate (as expected) of -0.16. As shown by Figures \ref{fig:strat10} to \ref{fig:strat10_3} in SA \ref{App:TVP1}, the selective segmentation method allows easier detection of the relevant factors as compared to the TVP model. 

\begin{table}[h!]
\centering
\singlespacing
\caption{\justifying \textbf{Fixed Income Arbitrage: linear, CP and selective segmentation risk models.}\\ 
The Table details the parameter estimates of the linear model, of the CP model and of the selective segmentation process with FIA returns as the dependent variable. Parentheses indicate standard deviations and brackets [-] document the 95\% credible intervals of the breakpoints that are computed from the Bayesian model given in Section \ref{sec:breakuncertainty}.  A cell filled with '---' indicates that the parameter does not vary over the related period. The posterior probability of the selective segmentation model amounts to 99.6\%.}
\scalebox{0.6}{
\begin{tabular}{lcccccccccccccc}
\toprule
Period & 	Int. & AR1 & 	PMKT & 	SMB & 	TERM & 	DEF & 	PTFSBD & 	PTFSFX & 	PTFSCOM & UMD & PTFSIR & PTFSSTK & CPI & NAREIT\\ 
\midrule
 & \multicolumn{14}{c}{\textbf{Standard linear risk model}} \\														
  \cmidrule(lr){2-15}														
03-1994 to 03-2016 & 	0.12 & 	0.24 & 	0.02 & 	-0.03 & 	-0.99 & 	-3.60 & 	-0.00 & 	-0.01 & 	0.01 & 	-0.00 & 	-0.01 & 	0.01 & 	1.04 & 	0.03 \\ 
 & 	(0.08) & 	(0.04) & 	(0.02) & 	(0.02) & 	(0.33) & 	(0.46) & 	(0.00) & 	(0.00) & 	(0.00) & 	(0.01) & 	(0.00) & 	(0.00) & 	(0.24) & 	(0.02) \\[0.2cm] 
 & \multicolumn{14}{c}{\textbf{CP-YZ risk model}} \\														
  \cmidrule(lr){2-15}														
03-1994 to 09-1997 & 	0.57 & 	0.43 & 	0.03 & 	-0.05 & 	0.23 & 	0.63 & 	-0.01 & 	-0.00 & 	0.00 & 	0.00 & 	-0.00 & 	0.00 & 	-1.12 & 	0.09 \\ 
 & 	(0.22) & 	(0.13) & 	(0.03) & 	(0.03) & 	(0.68) & 	(2.13) & 	(0.01) & 	(0.00) & 	(0.01) & 	(0.04) & 	(0.01) & 	(0.01) & 	(0.79) & 	(0.03) \\ 
10-1997 to 03-1999 & 	-0.21 & 	0.24 & 	0.13 & 	-0.18 & 	-2.09 & 	-14.47 & 	0.04 & 	-0.03 & 	-0.04 & 	0.18 & 	-0.03 & 	0.02 & 	-0.24 & 	0.08 \\ 
 & 	(0.30) & 	(0.07) & 	(0.06) & 	(0.10) & 	(2.79) & 	(4.18) & 	(0.01) & 	(0.01) & 	(0.02) & 	(0.06) & 	(0.01) & 	(0.01) & 	(1.13) & 	(0.07) \\ 
04-1999 to 08-2000 & 	1.99 & 	-1.68 & 	0.15 & 	0.10 & 	1.79 & 	-1.31 & 	-0.00 & 	0.03 & 	-0.01 & 	-0.08 & 	0.00 & 	-0.04 & 	-3.13 & 	0.04 \\ 
 & 	(0.90) & 	(0.90) & 	(0.06) & 	(0.05) & 	(1.66) & 	(1.84) & 	(0.02) & 	(0.02) & 	(0.02) & 	(0.06) & 	(0.02) & 	(0.03) & 	(1.95) & 	(0.07) \\ 
09-2000 to 05-2007 & 	0.39 & 	0.29 & 	-0.01 & 	0.02 & 	-0.55 & 	-2.41 & 	-0.00 & 	0.01 & 	0.00 & 	-0.02 & 	0.00 & 	0.01 & 	0.25 & 	-0.03 \\ 
 & 	(0.09) & 	(0.08) & 	(0.02) & 	(0.02) & 	(0.37) & 	(0.72) & 	(0.00) & 	(0.00) & 	(0.00) & 	(0.01) & 	(0.00) & 	(0.01) & 	(0.20) & 	(0.02) \\ 
06-2007 to 10-2008 & 	-1.89 & 	-0.79 & 	0.49 & 	0.86 & 	-17.65 & 	-9.84 & 	-0.00 & 	-0.10 & 	0.16 & 	-0.10 & 	-0.01 & 	0.00 & 	3.14 & 	-0.66 \\ 
 & 	(0.33) & 	(0.20) & 	(0.10) & 	(0.20) & 	(3.48) & 	(1.94) & 	(0.02) & 	(0.01) & 	(0.02) & 	(0.06) & 	(0.00) & 	(0.02) & 	(0.56) & 	(0.10) \\ 
11-2008 to 08-2010 & 	0.99 & 	0.28 & 	0.13 & 	0.16 & 	-0.10 & 	-0.14 & 	0.02 & 	-0.04 & 	0.05 & 	-0.01 & 	-0.03 & 	0.05 & 	3.97 & 	-0.10 \\ 
 & 	(0.21) & 	(0.06) & 	(0.05) & 	(0.07) & 	(0.89) & 	(0.70) & 	(0.01) & 	(0.01) & 	(0.02) & 	(0.02) & 	(0.01) & 	(0.01) & 	(0.73) & 	(0.04) \\ 
09-2010 to 03-2016 & 	0.23 & 	0.36 & 	0.05 & 	-0.05 & 	-0.24 & 	-1.52 & 	-0.01 & 	-0.00 & 	0.00 & 	0.01 & 	0.01 & 	-0.00 & 	0.22 & 	-0.01 \\ 
 & 	(0.11) & 	(0.11) & 	(0.03) & 	(0.03) & 	(0.48) & 	(0.71) & 	(0.01) & 	(0.00) & 	(0.00) & 	(0.02) & 	(0.00) & 	(0.00) & 	(0.31) & 	(0.02) \\ 
 & \multicolumn{14}{c}{\textbf{Selective segmentation risk model (100\%) }} \\														
  \cmidrule(lr){2-15}														
03-1994 to 09-1997 & 	0.23 & 	0.45 & 	0.02 & 	0.02 & 	-0.24 & 	-1.35 & 	-0.00 & 	-0.01 & 	0.00 & 	0.12 & 	-0.00 & 	0.01 & 	0.17 & 	0.07 \\ 
~[10-1996~~09-1997] & 	(0.03) & 	(0.07) & 	(0.01) & 	(0.01) & 	(0.25) & 	(1.62) & 	(0.00) & 	(0.00) & 	(0.00) & 	(0.03) & 	(0.01) & 	(0.00) & 	(0.22) & 	(0.02) \\ 
10-1997 to 03-1999 & 	 ---  & 	 ---  & 	 ---  & 	 ---  & 	 ---  & 	-10.56 & 	 ---  & 	 ---  & 	 ---  & 	 ---  & 	-0.03 & 	 ---  & 	 ---  & 	 ---  \\ 
~[01-1999~~06-1999] & 	 ---  & 	 ---  & 	 ---  & 	 ---  & 	 ---  & 	(0.97) & 	 ---  & 	 ---  & 	 ---  & 	 ---  & 	(0.01) & 	 ---  & 	 ---  & 	 ---  \\ 
04-1999 to 08-2000 & 	 ---  & 	 ---  & 	 ---  & 	 ---  & 	 ---  & 	-1.18 & 	 ---  & 	0.01 & 	 ---  & 	-0.01 & 	0.00 & 	 ---  & 	 ---  & 	 ---  \\ 
~[07-2000~~05-2003] & 	 ---  & 	 ---  & 	 ---  & 	 ---  & 	 ---  & 	(0.34) & 	 ---  & 	(0.00) & 	 ---  & 	(0.01) & 	(0.00) & 	 ---  & 	 ---  & 	 ---  \\ 
09-2000 to 05-2007 & 	 ---  & 	 ---  & 	 ---  & 	 ---  & 	 ---  & 	 ---  & 	 ---  & 	 ---  & 	 ---  & 	 ---  & 	 ---  & 	 ---  & 	 ---  & 	-0.03 \\ 
~[04-2007~~06-2007] & 	 ---  & 	 ---  & 	 ---  & 	 ---  & 	 ---  & 	 ---  & 	 ---  & 	 ---  & 	 ---  & 	 ---  & 	 ---  & 	 ---  & 	 ---  & 	(0.02) \\ 
06-2007 to 10-2008 & 	 ---  & 	-1.35 & 	0.08 & 	-0.05 & 	 ---  & 	 ---  & 	-0.09 & 	-0.06 & 	0.10 & 	-0.34 & 	 ---  & 	-0.08 & 	2.69 & 	-0.16 \\ 
~[10-2008~~10-2008] & 	 ---  & 	(0.24) & 	(0.02) & 	(0.03) & 	 ---  & 	 ---  & 	(0.02) & 	(0.01) & 	(0.01) & 	(0.04) & 	 ---  & 	(0.01) & 	(0.40) & 	(0.03) \\ 
11-2008 to 08-2010 & 	 ---  & 	0.63 & 	 ---  & 	 ---  & 	 ---  & 	 ---  & 	-0.00 & 	-0.01 & 	0.00 & 	-0.00 & 	 ---  & 	0.00 & 	 ---  & 	-0.01 \\ 
~[06-2010~~11-2010] & 	 ---  & 	(0.14) & 	 ---  & 	 ---  & 	 ---  & 	 ---  & 	(0.01) & 	(0.00) & 	(0.00) & 	(0.01) & 	 ---  & 	(0.00) & 	 ---  & 	(0.02) \\ 
09-2010 to 03-2016 & 	 ---  & 	0.19 & 	 ---  & 	 ---  & 	 ---  & 	 ---  & 	 ---  & 	 ---  & 	 ---  & 	 ---  & 	 ---  & 	 ---  & 	0.21 & 	 ---  \\ 
 & 	 ---  & 	(0.09) & 	 ---  & 	 ---  & 	 ---  & 	 ---  & 	 ---  & 	 ---  & 	 ---  & 	 ---  & 	 ---  & 	 ---  & 	(0.35) & 	 ---  \\ 
\bottomrule
\end{tabular}}
\label{tab:OLS_SELO10}
\end{table}


\clearpage

\subsection{Comparison with advanced CP models \label{sec:CPJBF}}
We now compare our results with those of \cite{meligkotsidou2008detecting}. \cite{meligkotsidou2008detecting} rely on CP models to capture the risk exposition of HF returns over time. In particular, they consider the $4096$ distinct combinations of the twelve risk factors and for each of them, they estimate a CP model exhibiting several numbers of segments $m$ (from one to ten). Eventually, they use the marginal likelihood to select the best model among the set of $m \times 2^K$ estimated processes (i.e., 40960 models since $m=10$ and $K=12$). Their approach consists therefore in first selecting the relevant factors and then, in investigating if the exposition to them is time-varying. \\
There is a striking difference with our approach since, for each breakpoint, our procedure can detect what are the time-varying factors. In fact, our approach discriminates between $2^{m \times K}$ models; a number of models that exponentially increases with the amount of breaks. Note that we could also search for the best regressors to include by considering all the 4096 distinct combinations of the twelve factors. In such a case, the number of models to consider would reach $2^{(m+1) \times K}$.\\ 
We reproduce the results of \cite{meligkotsidou2008detecting} on our data by additionally taking the autocorrelation structure into account. Fixing the AR order $q$ to the value given in Table \ref{JBF::order}, for each possible combination of the factors, we estimate CP-ARX(q) models with different numbers of breaks (ranging from 1 to 10) by (globally) minimizing the MDL criterion. Then, we report the combination of factors exhibiting the best MDL value. Hereafter, we denote this model by CP-MV.\footnote{Our approach is slightly different as the one used in \cite{meligkotsidou2008detecting} since we minimize the MDL criterion instead of maximizing the marginal likelihood of a Bayesian CP model for finding the best combination of the factors and the breakpoints. This is motivated by the fact that the MDL criterion consistently selects the true number of regimes while there is no equivalent proof for the marginal likelihood used in \cite{meligkotsidou2008detecting}. In addition, \cite{Dufays2019} show that the MDL criterion is equal to minus the marginal log-likelihood of a CP Bayesian model with particular g-prior distributions. So, our approach can be understood as the method of \cite{meligkotsidou2008detecting} with different hyper-parameters. We globally minimize the MDL criterion using the dynamic programming of \cite{bai2003computation}.} \\
Table \ref{JBF::fact} documents the factors of the CP-MV model for the two HF strategies. It also reports the factors which exhibit significant parameter estimates at least at one period for the TVP model and the selective segmentation approach. For the HFI, the selected factors by the selective segmentation methods seem more complete than those of the TVP model in light of the current literature. While the TVP selects eight risk factors, the selective segmentation reports the same factors, except the SMB factor, and adds two trend following risk factors, PTFSBD and PTFSSTK, and the real estate risk factor, NAREIT. For the FIA, the results are much more contrasted. All risk factors are selected by the selective segmentation methods whereas five factors are omitted by the TVP (SMB, TERM, PTFSFX, PTFSCOM  and UMD). The CP-MV approach does not select SMB, NAREIT and four trend following risk factors (PTFSBD, PTFSCOM, PTFSSIR and PTFSSTK) for the HFI. The analysis of the single strategy, Fixed Income Arbitrage (FIA), also highlights important and significant differences as far as only two factors are selected, DEF and PTFSFX, by the CP-MV process. It may be very surprising since the market premium is (almost) always used in linear asset pricing models as pointed out by \cite{fung2001risk}. We may also note that the look-back straddles on commodities, PTFSCOM, on bond, PTFSBD and on stock index and PTFSSTK designed to capture non-linearities especially during changes in international economic policies, are not selected whereas the phenomenon is observed just after the GFC.

\begin{table}[h!]
\centering
\singlespacing
\caption{\justifying \textbf{HFI and FIA strategies: Selected factors given several time-varying parameter models.}\\ 
Selected factors by the TVP, the selective segmentation process and the CP-MV model of \cite{meligkotsidou2008detecting}. The factors of the latter process are chosen by minimizing the MDL criterion while for the TVP and the selective segmentation model, a factor is selected if its related parameter estimate is significant at least at one period over the sample.}
\scalebox{0.85}{
\begin{tabular}{l cccccc}
\toprule
        & \multicolumn{3}{c}{ HFI }  & \multicolumn{3}{c}{ FIA } \\
        & TVP & Sel. Seg. & CP-MV & TVP & Sel. Seg. & CP-MV\\
				\cmidrule(lr){2-4}	\cmidrule(lr){5-7}
PMKT & 	 $\surd$  &  	 $\surd$  &  	 $\surd$  &  	 $\surd$  &  	 $\surd$  &  	          \\ 
SMB & 	 $\surd$  &  	          &  	          &  	          &  	 $\surd$  &  	          \\ 
TERM & 	 $\surd$  &  	 $\surd$  &  	 $\surd$  &  	          &  	 $\surd$  &  	          \\ 
DEF & 	 $\surd$  &  	 $\surd$  &  	 $\surd$  &  	 $\surd$  &  	 $\surd$  &  	 $\surd$  \\ 
PTFSBD & 	          &  	 $\surd$  &  	          &  	 $\surd$  &  	 $\surd$  &  	          \\ 
PTFSFX & 	 $\surd$  &  	 $\surd$  &  	 $\surd$  &  	          &  	 $\surd$  &  	 $\surd$  \\ 
PTFSCOM & 	          &  	          &  	          &  	          &  	 $\surd$  &  	          \\ 
UMD & 	 $\surd$  &  	 $\surd$  &  	 $\surd$  &  	          &  	 $\surd$  &  	          \\ 
PTFSIR & 	 $\surd$  &  	 $\surd$  &  	          &  	 $\surd$  &  	 $\surd$  &  	          \\ 
PTFSSTK & 	          &  	 $\surd$  &  	          &  	 $\surd$  &  	 $\surd$  &  	          \\ 
CPI & 	 $\surd$  &  	 $\surd$  &  	 $\surd$  &  	 $\surd$  &  	 $\surd$  &  	          \\ 
NAREIT & 	          &  	 $\surd$  &  	          &  	 $\surd$  &  	 $\surd$  &  	          \\ 
\bottomrule
\end{tabular}}
\label{JBF::fact} 
\end{table}

Table \ref{JBF::fact} does not inform on the dynamic of the selected factors by the CP-MV process. Although the preferred specification of the CP-MV model does not include all the factors, the risk exposure of the HF strategies is still abruptly changing over time. Regarding the HFI, two breakpoints are detected and occur in April 2000 and in March 2005, respectively. Table \ref{tab:OLS_SELO1:CPMV} shows how the selective segmentation method improves the interpretation of the CP-MV results. It also illustrates the improvement in economic modelling as far as it highlights the relative role played by static and dynamic parameters. First, we observe that the alpha, the currency lookback straddle, PTFSFX, and the consumer prince index, CPI, are quite static during the full period. Interestingly, the selective segmentation also reports the CP in the bullish market in the early 2000s. As mentioned earlier, the financial markets were indeed very volatile during the 1990s. This break is clearly reported by the dynamic risk factors, PMKT, UMD and DEF. PMKT declines from 0.41 to 0.21 during the 2000s and afterwards. The trend is more striking for the momentum factor, UMD, with a decline from 0.18 to 0.05, and for the credit default risk factor, DEF, rising from -11.71 to -2.21. The term structure risk factor, TERM, is not statistically significant after the rise of the housing price index in 2005-2006, announcing the collapse of the financial markets in 2008 (as anticipated by \cite{Shiller-PUP-2015} among others) and the following quantitative easing policies with very low inters rates.

The FIA strategy exhibits seven regimes which makes the CP-MV model heavily parametrized (i.e., $K \times m = 35 $ parameters). This large number of regimes is probably related to the fact that more breakpoints are needed to adequately fit the FIA returns since the CP-MV specification includes only two risk factors, DEF and PTFSFX and/or because in this specific economic modelling the breakpoints also capture the variance dynamic. Using the selected factors and the breakpoints of the best CP-MV specification, we estimate the selective segmentation model to uncover what are the static and the dynamic parameters. First, we must acknowledge that the best specification selected by the CP-MV model is doubtful with regards to the financial literature and the practice. Nevertheless, we compare the CP-MV approach with the selected segmentation to highlight the contribution of the latter. In particular, we observe that the 'alpha' is varying and statistically positive during the 1990s until LTCM collapse. As expected, the default risk factor is negative, time-varying and high during crises (-9.49 during the LTCM collapse and -6.17 during the GFC). After the GFC, the default factor is constant, negative and not statistically significant. This result is consistent with the trend observed on financial markets (especially on fixed incomes markets after the GFC). The currency trend following factor, PTFSFX, is very low, time-varying and statistically significant during the global financial crisis (as expected) and before the impact of the quantitative easing policies starting in the late 2010. After this date (11/2010) the factor is not statistically significant. This is an illustration of the impact of quantitative easing on fixed income arbitrage. 

As illustrated by these empirical results, our method uncovers which parameters truly vary when a CP is detected. This technical improvement induces financial consequences and especially cost reductions.\footnote{We thank an anonymous referee for this relevant comment.} Moreover, our method should imply more accurate and thus less frequent portfolio rebalancing strategies. Investor could indeed change his timing by using our approach and decide to rebalance parsimoniously (and thus efficiently) his investment when a break is detected, with a special focus on the relevant benchmark. 

\begin{table}[h!]
\centering
\singlespacing
\caption{\justifying \textbf{Hedge Fund Index: Best CP-MV model and best selective segmentation model.}\\ 
The Table details the parameter estimates of the preferred CP-MV model and of the selective segmentation process given the selected factors and the breakpoints found by the CP-MV model. Parentheses indicate standard deviations.  A cell filled with '---' indicates that the parameter does not vary over the related period. The posterior probability of the selective segmentation model amounts to 77\%.}
\scalebox{0.65}{
\begin{tabular}{lccccccc||ccccccc}
\toprule					
  & \multicolumn{7}{c}{\textbf{Preferred CP-MV model}} &								  \multicolumn{7}{c}{\textbf{Selective segmentation (77\%) }} \\						
  \cmidrule(lr){2-8}								  \cmidrule(lr){9-15}						
Period & 	Int. & 		PMKT & 	TERM & 	DEF & 	PTFSFX & 	UMD & 	CPI & 		 	Int. & 		PMKT & 	TERM & 	DEF & 	PTFSFX & 	UMD & 	CPI \\ 
03-1994 to 04-2000 & 	0.42 & 	0.41 & 	-2.47 & 	-12.13 & 	0.02 & 	0.18 & 	0.52 & 	0.22 & 	0.41 & 	-2.44 & 	-11.71 & 	0.01 & 	0.18 & 	1.17 \\ 
 & 	(0.38) & 	(0.05) & 	(1.28) & 	(2.62) & 	(0.01) & 	(0.06) & 	(1.63) & 	(0.09) & 	(0.03) & 	(0.52) & 	(1.44) & 	(0.00) & 	(0.04) & 	(0.29) \\ 
05-2000 to 03-2005 & 	0.44 & 	0.19 & 	-2.20 & 	-2.57 & 	0.02 & 	0.05 & 	0.64 & 	 ---  & 	0.21 & 	 ---  & 	-2.21 & 	 ---  & 	0.05 & 	 ---  \\ 
 & 	(0.34) & 	(0.06) & 	(1.21) & 	(2.21) & 	(0.01) & 	(0.04) & 	(1.11) & 	 ---  & 	(0.02) & 	 ---  & 	(0.54) & 	 ---  & 	(0.02) & 	 ---  \\ 
04-2005 to 03-2016 & 	0.11 & 	0.23 & 	0.28 & 	-1.82 & 	0.01 & 	0.05 & 	1.25 & 	 ---  & 	 ---  & 	0.29 & 	 ---  & 	 ---  & 	 ---  & 	 ---  \\ 
 & 	(0.18) & 	(0.04) & 	(0.85) & 	(0.91) & 	(0.01) & 	(0.03) & 	(0.51) & 	 ---  & 	 ---  & 	(0.55) & 	 ---  & 	 ---  & 	 ---  & 	 ---  \\ 
\bottomrule
\end{tabular}}
\label{tab:OLS_SELO1:CPMV}
\end{table}

\begin{table}[h!]
\centering
\singlespacing
\caption{\justifying \textbf{Fixed Income Arbitrage: Best CP-MV model and best selective segmentation model.}\\ 
The Table details the parameter estimates of the preferred CP-MV model and of the selective segmentation process given the selected factors and the breakpoints found by the CP-MV model. Parentheses indicate standard deviations.  A cell filled with '---' indicates that the parameter does not vary over the related period. The posterior probability of the selective segmentation model amounts to 52\%.}
\scalebox{0.7}{
\begin{tabular}{lcccc||cccc}
\toprule
 & \multicolumn{4}{c}{\textbf{Preferred CP-MV model}} &					  \multicolumn{4}{c}{\textbf{Selective segmentation (52\%) }} \\			
  \cmidrule(lr){2-5}					  \cmidrule(lr){6-9}			
Period & 	Int. & 	AR1 & 	DEF & 	PTFSFX & 			 	Int. & 	AR1 &	DEF & 	PTFSFX \\ 	
03-1994 to 05-1995 & 	0.23 & 	0.40 & 	7.57 & 	-0.03 & 	0.55 & 	0.36 & 	-0.14 & 	-0.01 \\ 
 & 	(0.23) & 	(0.19) & 	(3.14) & 	(0.01) & 	(0.13) & 	(0.06) & 	(2.04) & 	(0.01) \\ 
06-1995 to 08-1997 & 	-0.13 & 	1.28 & 	-0.68 & 	0.00 & 	 ---  & 	 ---  & 	 ---  & 	 ---  \\ 
 & 	(0.44) & 	(0.53) & 	(2.57) & 	(0.01) & 	 ---  & 	 ---  & 	 ---  & 	 ---  \\ 
09-1997 to 11-1998 & 	0.10 & 	-0.23 & 	-7.52 & 	-0.07 & 	0.20 & 	 ---  & 	-9.49 & 	-0.05 \\ 
 & 	(0.08) & 	(0.22) & 	(1.47) & 	(0.02) & 	(0.04) & 	 ---  & 	(1.52) & 	(0.02) \\ 
12-1998 to 02-2008 & 	0.29 & 	0.38 & 	-1.46 & 	0.00 & 	 ---  & 	 ---  & 	-1.52 & 	0.00 \\ 
 & 	(0.09) & 	(0.08) & 	(0.48) & 	(0.00) & 	 ---  & 	 ---  & 	(0.59) & 	(0.00) \\ 
03-2008 to 05-2009 & 	0.04 & 	-0.41 & 	-6.97 & 	-0.03 & 	 ---  & 	 ---  & 	-6.17 & 	-0.04 \\ 
 & 	(0.05) & 	(0.20) & 	(0.51) & 	(0.01) & 	 ---  & 	 ---  & 	(0.54) & 	(0.01) \\ 
06-2009 to 10-2010 & 	0.06 & 	1.23 & 	0.09 & 	-0.05 & 	 ---  & 	 ---  & 	-0.99 & 	 ---  \\ 
 & 	(0.21) & 	(0.35) & 	(0.94) & 	(0.01) & 	 ---  & 	 ---  & 	(0.56) & 	 ---  \\ 
11-2010 to 03-2016 & 	0.28 & 	0.24 & 	-1.91 & 	-0.01 & 	 ---  & 	 ---  & 	 ---  & 	-0.01 \\ 
 & 	(0.16) & 	(0.10) & 	(0.73) & 	(0.00) & 	 ---  & 	 ---  & 	 ---  & 	(0.01) \\ 
\bottomrule
\end{tabular}}
\label{tab:OLS_SELO10:CPMV}
\end{table}

\subsection{Out-of-sample}\label{sec:oos}
Sections \ref{sec:in} and \ref{sec:CPJBF} highlight the in-sample advantages of detecting which parameter truly varies when a break is detected. In addition to that, since the selective segmentation method can more accurately estimate parameters that do not change when a break occurs, we could also expect some prediction gains with respect to the standard CP model. In this Section, we investigate this aspect using the root mean squared forecast errors (RMSFE) and the cumulative log predictive density (CLPD), two standard loss functions specified as,
\begin{eqnarray*}
\text{RMSFE} & = & \sqrt{\frac{1}{T-\underline{t}} \sum_{t=\underline{t}+1}^{T}(y_t-\hat{y}_t)^2}, \text{ and } \text{CLPD}  = \sum_{t=\underline{t}+1}^{T} \log f(y_{t}|y_{1:t-1},\bx_t), 
\end{eqnarray*}
in which $\hat{y}_t$ is the conditional mean of $y_t$ given the information up to period $t$ (i.e., $\mathbb{E}(y_t|y_{1:t-1},\bx_t)$), $f(y_{t}|y_{1:t-1},\bx_t)$ denotes the predictive density of the model and $\underline{t}+1$ denotes the beginning of the out-of-sample forecasting period. In our prediction exercise, the training set is fixed to 20\% of the sample size and the 80\% remaining observations are used to assess the forecast performance (i.e., $\underline{t} = 0.2T$). Since our data comprise 265 monthly returns, the out-of-sample set of observations amounts to 212 months. Each time we move forward by one month, all the considered models are re-estimated and a forecast for the next period is produced.\\ 
As competitors to our model, we consider three other processes: i) a linear regression, ii) a standard CP model with breakpoints determined by the modified method of \cite{yau2016inference} documented in Section \ref{sec:breakYau} (hereafter CP-YZ), iii) a CP model with the number and the locations of the breakpoints selected by minimizing the MDL criterion (hereafter CP-MDL).\footnote{We do not compare with the CP model of \cite{meligkotsidou2008detecting} since the model is computationally too involved due to the number of explanatory variables. When an AR(2) model is selected, the number of models to consider at each iteration of the prediction exercise amounts to $10 \times 2^{15} = 327680$.} The minimization of the MDL criterion is carried out using the dynamic programming of \cite{bai2003computation}.\footnote{See \cite{eckley2011analysis} for a discussion on the implementation of the algorithm for the MDL criterion. Minimum regime duration is set to $\frac{3}{2}(K+1)$ to avoid capturing outliers. This choice is in favor of the standard CP model as the parameter estimates of the new regimes are based on at least $\frac{3}{2}(K+1)$ observations.} In addition to the factors and an intercept, we also account for the autocorrelation of the HF returns by fixing the AR order to the value given in Table \ref{JBF::order}. Regarding the CLPD metric, we assume a normal distribution for the error term and we also use the prior distributions given in Equation \eqref{eq:prior} for the linear and the full CP models.

Table \ref{tab::RMSFE_CLPD} documents the RMSFE and the CLPD criteria for all the Credit Suisse HF returns. For both metrics, we observe that the linear model dominates at least half of the times. Overall, the selective segmentation method improves the RMSFE and the CLPD for 6 and 5 out of 14 HF returns, respectively. Importantly, Table \ref{tab::RMSFE_CLPD} highlights that the selective segmentation process provides the most robust predictions. In particular, our approach delivers at least the second best predictive performance for all the HF returns. This is evidence that model averaging stabilizes the forecast by reducing its variance as argued in \cite{RapachEtAl2009}. Interestingly, our method compares extremely well with respect to the two CP models since it outperforms them 13 out of 14 HF returns for both metrics. Since the CP models are based on the same breakpoints as the selective segmentation processes, it is remarkable that the latter models almost systematically dominate CP models where all the parameters are time-varying. From this small sample of series, we could argue that the selective segmentation approach should replace the CP process as it would likely improve the forecast performance.

\begin{table}[htbp]
\small
\centering
	\caption{\justifying \textbf{RMSFE and CLPD for the fourteen HF strategies ($\underline{t} = 0.2T$).} \\
	The Table details the RMSFE and the CLPD for five processes. Bold values indicate the model that delivers the best prediction performance. A star points out the second best model.}
\begin{tabular}{lccccccc}
\toprule
       & \multicolumn{7}{c}{\textbf{RMSFE}}  \\[0.2cm]
 Series & HFI & 	CNV & 	DSB & 	EME & 	EMN & 	EDR & 	EDD \\ 
\cmidrule(lr){2-8}
Linear & 	 1.41  & 	\textbf{1.63} & 	\textbf{2.80} & 	 2.99  & 	\textbf{2.95} & 	\textbf{1.40} & 	 1.54* \\ 
CP-MDL & 	 1.41  & 	 1.92  & 	 2.95  & 	\textbf{2.88} & 	 4.12  & 	 1.62  & 	 1.84  \\ 
SELO-MDL & 	 1.30* & 	 1.84* & 	 2.88* & 	 2.91* & 	 3.95* & 	 1.52  & 	\textbf{1.53} \\ 
CP-Yau & 	 1.70  & 	 2.63  & 	 4.03  & 	 3.90  & 	 6.06  & 	 3.24  & 	 2.89  \\ 
SELO-Yau & 	\textbf{1.28} & 	 2.07  & 	 3.08  & 	 2.93  & 	 5.26  & 	 1.50* & 	 1.55  \\ [0.3cm] 
 Series & EDM & 	EDRA & 	FIA & 	GMA & 	LES & 	MFU & 	MUS \\ 
		\cmidrule(lr){2-8}		
Linear & 	\textbf{1.56} & 	 1.06* & 	\textbf{1.22} & 	 2.57  & 	 1.52  & 	 3.31* & 	\textbf{1.21} \\ 
CP-MDL & 	 1.72  & 	 1.12  & 	 1.36  & 	 2.57  & 	 1.64  & 	 3.50  & 	 1.40  \\ 
SELO-MDL & 	 1.61  & 	 1.07  & 	 1.33* & 	\textbf{2.39} & 	 1.49* & 	\textbf{3.31} & 	 1.30* \\ 
CP-Yau & 	 2.29  & 	 2.39  & 	 5.66  & 	 3.23  & 	 2.69  & 	 4.28  & 	 1.83  \\ 
SELO-Yau & 	 1.60* & 	\textbf{1.06} & 	 1.57  & 	 2.45* & 	\textbf{1.47} & 	 3.31  & 	 1.40  \\   [0.3cm] 
\midrule
       & \multicolumn{7}{c}{\textbf{CLPD}}  \\[0.2cm]
 Series & HFI & 	CNV & 	DSB & 	EME & 	EMN & 	EDR & 	EDD \\ 
\cmidrule(lr){2-8}  
Linear & 	 -365.63  & 	\textbf{-401.02} & 	\textbf{-513.39} & 	 -527.76  & 	\textbf{-547.17} & 	\textbf{-356.96} & 	 -360.72  \\ 
CP-MDL & 	 -359.08  & 	 -444.16  & 	 -526.48  & 	\textbf{-515.71} & 	 -713.12  & 	 -433.56  & 	 -389.30  \\ 
SELO-MDL & 	\textbf{-347.80} & 	 -434.44  & 	 -519.29* & 	 -516.35* & 	 -688.23* & 	 -402.20  & 	 -353.99* \\ 
CP-Yau & 	 -366.11  & 	 -483.20  & 	 -570.98  & 	 -565.77  & 	 -916.15  & 	 -482.06  & 	 -439.75  \\ 
SELO-Yau & 	 -347.83* & 	 -434.19* & 	 -529.15  & 	 -519.49  & 	 -932.26  & 	 -395.13* & 	\textbf{-346.94} \\ [0.3cm] 
 Series & EDM & 	EDRA & 	FIA & 	GMA & 	LES & 	MFU & 	MUS \\  
		\cmidrule(lr){2-8}		
Linear & 	\textbf{-375.00} & 	 -301.93* & 	\textbf{-329.73} & 	 -490.10  & 	 -382.79  & 	\textbf{-553.71} & 	\textbf{-330.48} \\ 
CP-MDL & 	 -445.12  & 	 -307.96  & 	 -377.46  & 	 -484.06  & 	 -387.78  & 	 -564.47  & 	 -343.62  \\ 
SELO-MDL & 	 -389.83  & 	 -303.77  & 	 -352.53* & 	\textbf{-467.38} & 	 -369.13* & 	 -555.02  & 	 -339.35  \\ 
CP-Yau & 	 -460.95  & 	 -426.12  & 	 -472.99  & 	 -492.14  & 	 -464.27  & 	 -606.26  & 	 -368.72  \\ 
SELO-Yau & 	 -384.86* & 	\textbf{-301.57} & 	 -380.37  & 	 -472.48* & 	\textbf{-364.02} & 	 -554.74* & 	 -337.58* \\ 
\bottomrule
\end{tabular}
\label{tab::RMSFE_CLPD}
\end{table}



\section{Conclusion} \label{s:conclusion}
Since the seminal work of \cite{chernoff1964estimating}, many CP detection methods for linear models have been proposed. Most of these CP models have in common to assume, at least in practice, that all the model parameters have to change when a break is detected. In this paper, we propose to go beyond this standard framework by capturing which parameters vary when a structural break occurs. Even when conditioning to the break dates, detecting the parameters that vary from one segment to the next is not straightforward since the number of possibilities grows exponentially with the number of breaks and the number of explanatory variables. To solve this dimensional problem, we propose a penalized regression method to explore the model space and we select the best specification by maximizing a criterion that can be interpreted as a marginal likelihood in the Bayesian paradigm. \\
To carry out the model space exploration, we use an almost unbiased penalty function, a desirable property in CP frameworks that is not exhibited by standard penalty functions (e.g., LASSO and Ridge estimators). Also, we prove the consistency of our estimator and we show how to estimate it using the DAEM algorithm. To apply the DAEM algorithm in our context, we transform the penalty function into a mixture of Normal distributions. This simple transformation greatly speeds the estimation as the DAEM algorithm iterates over closed-form expressions.\\
Once the promising models have been uncovered by the penalized regression approach, selecting the parameters of the penalty function is carried out by maximizing a marginal likelihood. Thanks to the Bayesian interpretation of this consistent criterion, we can take model uncertainty into account and can do Bayesian model averaging, a feature that generally improves forecast performance. A simulation study highlights that our selective segmentation method works well in practice for a range of diversified data generating processes.\\ 
We illustrate our approach with HF returns. The selective segmentation model has two main advantages. First, as the standard CP models, it detects the breakpoints and the corresponding regimes. Second, it highlights the time-varying dynamics of the changing risk factors. When we compare our model with previous advanced CP models, we observe that it is particularly appealing to capture the time-varying dynamics of risk exposures. Then, we test the predictive performance of the selective segmentation approach with respect to the linear regression and standard CP processes. We note that our method produces the most robust forecasts and almost systematically dominates the CP processes based on the same breakpoints. These encouraging results suggest promising developments and applications in financial economics. \\
Importantly, an R package for estimating the model is available on the corresponding author's web page. This package stands for a building block of our future research that will include a dynamic variance and multivariate models.


%

{\setstretch{0.9}
\bibliography{bibliography}
\bibliographystyle{ecta}}

\newpage
\appendix
\setcounter{page}{1}
\setcounter{equation}{0}
\setcounter{footnote}{0} 
\setcounter{section}{0} 
\setcounter{table}{0} 

\onehalfspacing
\begin{mytitlepage}
\title{\bf Supplementary Appendix for \\ 'Selective linear segmentation for detecting relevant parameter changes'}
\maketitle

\blfootnote{
\emph{Email addresses}:
\texttt{arnaud.dufays@unamur.be} (Arnaud Dufays), \texttt{elysee-aristide.houndetoungan.1@ulaval.ca} (Aristide Houndetoungan), \texttt{coen.alain@uqam.ca} (Alain Co\"{e}n).
}

\begin{abstract}
\begin{spacing}{1.15} 
\noindent 
Section A is devoted to the proof of the consistency of our estimator (see Proposition 1 in the paper). We also discuss the theoretical consequences of imposing bounded eigenvalues and how to calibrate the spike and slab prior to mimic the SELO penalty function. Section B shows that the criterion used to select the best model stands as the marginal likelihood of a linear model with specific prior distributions. It also develops the posterior and the predictive distributions of the Bayesian model. Section C contains the proof of the consistency of the marginal likelihood criterion as well as its link with the Bayesian information criterion. The time-varying parameter (TVP) model used as an alternative process is specified in Section D. The section also provides how the TVP model is estimated and additional graphics of the parameter dynamics. Finally, Section E briefly discusses the Bayesian alternatives to the selective segmentation approach.  
\end{spacing}
\vspace{1cm}

\noindent
{\bf Keywords:} change-point, structural change, time-varying parameter, model selection, Hedge funds.\\

\noindent {\bf JEL Classification:} C11, C12, C22, C32, C52, C53.
\thispagestyle{empty}	
	
\end{abstract}

\end{mytitlepage}

\newpage
\doublespacing

\section{Proofs of the consistency of the Penalty function \label{App:theo1}}
In this appendix, we proof Proposition \ref{theoconsislimitdist}. To do so, we first state and prove two Lemmas.
\begin{lem}
\label{lemforcons1}
Under the conditions \ref{tauistruetau}-\ref{tauisO} and let,
\begin{align}
f_{T}(\bbeta) = \dfrac{1}{T}||\by- \bX_{\btau} \bbeta||_2^2 + \sum_{j=2}^m\sum_{k=1}^{K}\pen(\Delta\beta_{jk}|a_k,\lambda) 
\end{align}
Then of every $\nu \in (0,1)$, there exists a constant $C_0 > 0$ such that 
$$ \liminf_{T \rightarrow \infty} ~ \mathbf{P} \left[ \argmin_{||\bbeta- \bbeta^*||_2 \leq C\sqrt{\frac{Km\sigma^2}{T}}} f_{T}(\bbeta) \subseteq \left\{\bbeta  \in \Re^{Km\times 1}; ||\bbeta- \bbeta^*||_2 < C\sqrt{\frac{Km\sigma^2}{T}} \right\} \right] > 1-\nu$$
for all $C \geq C_0$.
\end{lem}

\begin{proof}
The proof is given in Appendix \ref{app:proof1}.
\end{proof}

\begin{lem}
\label{lemforcons2}
Let $C>0$ and $f_{T}$ as defined by Equation \eqref{ojecfun}. Under the conditions \ref{tauistruetau}-\ref{tauisO},  
$$\liminf_{T \rightarrow \infty} ~ \mathbf{P}  \left[ \argmin_{||\bbeta- \bbeta^*||_2 \leq C\sqrt{\frac{Km\sigma^2}{T}}} f_{T}(\bbeta) \subseteq \left\{\bbeta  \in \Re^{Km\times 1}; \bbeta_{A^c} = 0\right\} \right] = 1$$
where $A^c = \{(j,k), j=1,\dots,m \text{ and } k=0,\dots,K-1\} \backslash A$ is the complement of $A$ in $\{(j,k), j=1,\dots,m \text{ and } k=0,\dots,K-1\}$, $\bbeta_{A^c} \in \Re^{|A^c|\times 1}$ is the $|A^c|$-dimensional sub-vector of $\bbeta$ containing components subscripted by $A^c$.
\end{lem}

\begin{proof}
See Appendix \ref{App:proof2} for the proof.
\end{proof}

\subsection{Proof of Lemma 1 \label{app:proof1}}
\begin{Proof}
	We consider the objective function 
	$$
	f_{T}(\bbeta) = \dfrac{1}{T}||\by- \bX_{\btau} \bbeta||_2^2 + \sum_{j=2}^m\sum_{k=1}^{K}\pen(\Delta\beta_{jk}|a_k,\lambda) 
	$$
	
	\noindent Let $\alpha_T = \sqrt{\frac{Km\sigma^2}{T}}$ and $\nu \in (0,1)$. To prove the lemma \ref{lemforcons1}, It suffices to show that
	$$ \mathbf{P} \left( f_{T}(\bbeta^*) < \inf_{||\bu||_2 = 1} f_{T}(\bbeta^* + C\alpha_T\bu) \right) = 1 -\nu $$
	for C>0 sufficiently large and for any $T$ sufficiently large. In other words, we shall show that  $H_{T}(\bu) = f_{T}(\bbeta^* + C\alpha_T \bu)  - f_{T}(\bbeta^*)$ is positive for any $T$ when $C$ is large enough and for all $||\bu||_2 = 1$, where $\bu = (u_{11}, u_{12} ~\dots, u_{1K}, ~\dots, ~ u_{mK}) \in \REAL^{mK\times 1}$.\\
	We can easily show that
	\begin{align*}
	H_{T}(\bu) &= \dfrac{1}{T}\left(C^2\alpha_T^2||\bX_{\btau}\bu||_2^2 - 2C\alpha_T\bepsilon'\bX_{\btau}\bu  \right) + \\
	& \sum_{j=2}^m\sum_{k=1}^{K}\left(\pen(\Delta\beta_{jk}^* + C\alpha_T u_{jk}|a_k,\lambda) - \pen(\Delta\beta_{jk}^*|a_k,\lambda)\right)\\
	H_{T}(\bu) &\geq \dfrac{1}{T}\left(C^2\alpha_T^2||\bX_{\btau}\bu||_2^2 - 2C\alpha_T\bepsilon'\bX_{\btau}\bu\right) +\\
	& \sum_{(j,k) \in \mathbf{D}(\bu)}\left(\pen(\Delta\beta_{jk}^{*+}|a_k,\lambda) - \pen(\Delta\beta_{jk}^*|a_k,\lambda)\right)
	\end{align*} 
	\noindent where $\Delta\beta_{jk}^{*+}= \Delta\beta_{jk}^* + C\alpha_T u_{jk}$ and\\
	$ \mathbf{D}(\bu) = \left\{(j,k); j\geq 2 ~ \text{and} ~ \pen(\Delta\beta_{jk}^{*+}|a_k,\lambda) - \pen(\Delta\beta_{jk}^*|a_k,\lambda) <0 \right\}$.\\
	For any $(j,k) \in \mathbf{D}(\bu)$, clearly $\Delta\beta_{jk}^* \ne 0$, otherwise $\pen(\Delta\beta_{jk}^{*+}|a_k,\lambda) - \pen(\Delta\beta_{jk}^*|a_k,\lambda) \geq 0$. Thus, if $C>0$ is sufficiently large and fixed, as $\lim_{T \to \infty} C\alpha_T = 0$, we can consider that $\Delta\beta_{jk}^{*+}$ and $\Delta\beta_{jk}^{*}$  have the same sign for $T$ sufficiently large; that is $0 \not\in (c_T^- , c_T^+)$, where 
	$c_T^- = \min(\Delta\beta_{jk}^{*+},\Delta\beta_{jk}^*)$ and $c_T^+ = \max(\Delta\beta_{jk}^{*+},\Delta\beta_{jk}^*)$. By the fact that $\pen(x|a_k,\lambda)$ is a concave function on $x\in (-\infty, 0]$ and on $x \in [0, +\infty)$, thus also on $(c_T^- , c_T^+)$, we can establish the following conditions using the mean value theorem.
	
	$$ \dfrac{\pen(\Delta\beta_{jk}^{*+}|a_k,\lambda) - \pen(\Delta\beta_{jk}^*|a_k,\lambda)}{C\alpha_T u_{jk}} \leq \max\left(\pen'(\Delta\beta_{jk}^{*+}|a_k,\lambda),\pen'(\Delta\beta_{jk}^*|a_k,\lambda)\right) $$
	$$ \dfrac{\pen(\Delta\beta_{jk}^{*+}|a_k,\lambda) - \pen(\Delta\beta_{jk}^*|a_k,\lambda)}{C\alpha_T u_{jk}} \geq \min\left(\pen'(\Delta\beta_{jk}^{*+}|a_k,\lambda),\pen'(\Delta\beta_{jk}^*|a_k,\lambda)\right) $$
	\noindent where $\pen'$ stands for the $\pen$ first derivative.\\
	Let us note that $\forall ~(j,k) \in \mathbf{D}(\bu)$, $\Delta\beta_{jk}^* > 0 \implies u_{jk}<0$ and $\Delta\beta_{jk}^* < 0 \implies u_{jk}>0$ so that $\pen(\Delta\beta_{jk}^{*+}|a_k,\lambda) - \pen(\Delta\beta_{jk}^*|a_k,\lambda) <0$ holds. Applying the mean value theorem in both cases, we end up with a common condition given by 
	\begin{align*}
	\pen(\Delta\beta_{jk}^{*+}|a_k,\lambda) - \pen(\Delta\beta_{jk}^*|a_k,\lambda)&\geq -C\alpha_T|u_{jk}||\pen'(\Delta\beta_{jk}^* + C\alpha_T u_{jk}|a_k,\lambda)|\\
	&\geq -\frac{C\lambda \alpha_T a_k\zeta}{\ln(2)(\rho^2 - 2\rho C \alpha_T)}\\ 
	&\geq -\frac{C\lambda \alpha_T a_{max}\zeta}{\ln(2)(\rho^2 - 2\rho C \alpha_T)}
	\end{align*}
	\noindent where the last two inequalities come from the $|\pen'(\Delta\beta_{jk}^* + C\alpha_T u_{jk}|a_k,\lambda)|$ minimization with respect to $\Delta\beta_{jk}^*$. Then
	\begin{align*}
	H_{T}(\bu) &\geq \underbrace{\dfrac{C^2\alpha_T^2||\bX_{\btau}\bu||_2^2}{T}}_{Q_1} - \underbrace{\dfrac{2C\alpha_T\bepsilon'\bX_{\btau}\bu }{T}}_{Q_2}  -\underbrace{\frac{CKm\lambda\alpha_T a_{max}\zeta}{\ln(2)(\rho^2 - 2\rho C \alpha_T)}}_{Q_3}
	\end{align*}
	\noindent Focusing on each term, we can show that
	$$
	Q_1 \equiv  \dfrac{C^2\alpha_T^2||\bX_{\btau}\bu||_2^2}{T} = C^2\alpha_T^2\bu'\dfrac{\bX_{\btau}'\bX_{\btau}}{T}\bu \geq C^2\alpha_T^2\lambda_{T,\min}$$
	\noindent where $\lambda_{T,\min}$ is the smallest eigenvalue of $\dfrac{\bX_{\btau}'\bX_{\btau}}{T}$.\\
	To show this condition, we can decompose  $\dfrac{\bX_{\btau}'\bX_{\btau}}{T}$ into $\boldsymbol{U}\boldsymbol{\Lambda}\boldsymbol{U}'$ (by \ref{Xfullrank}). Moreover, any vector of $Km$ dimension can be decomposed into a linear combination of the eigenvectors (i.e., $\displaystyle \boldsymbol{u} = \boldsymbol{U}\boldsymbol{\omega}$). Note that $\displaystyle \boldsymbol{u}'\boldsymbol{u} = \boldsymbol{\omega}'\boldsymbol{U}'\boldsymbol{U}\boldsymbol{\omega} = \boldsymbol{\omega}'\boldsymbol{\omega} = \sum_{i=1}^{Km}\omega_i^2 = 1$.\\
	Thus $\displaystyle \bu'\dfrac{\bX_{\btau}'\bX_{\btau}}{T}\bu = \boldsymbol{\omega}'\boldsymbol{U}'\boldsymbol{U}\boldsymbol{\Lambda}\boldsymbol{U}'\boldsymbol{U}\boldsymbol{\omega} = \boldsymbol{\omega}'\boldsymbol{\Lambda}\boldsymbol{\omega} = \sum_{i=1}^{Km}\omega_i^2\lambda_i$ $\geq \lambda_{T,\min}$.\\
	The second term is given by
	\begin{align*}
	Q_2 \equiv  \dfrac{2C\alpha_T\bepsilon'\bX_{\btau}\bu }{T} &\leq \dfrac{2C\alpha_T|\bepsilon'\bX_{\btau}\bu |}{T}\\
	&\leq \dfrac{2C\alpha_T||\bepsilon'\bX_{\btau}||_2||\bu||_2}{T} ~ \text{by Cauchy-Schwartz}\\
	&\leq \frac{2C\alpha_T^2}{\sqrt{Km\sigma^2}}\sqrt{\frac{(\bepsilon' (\bX_{\btau}\bX_{\btau}') \bepsilon)}{T}} \\
	&\leq \mathcal{O}_p(C\alpha_T^2) \quad (By~ \ref{Xfullrank} ~ and ~ \ref{exogeneity}).
	\end{align*}
	To show that $\sqrt{\frac{(\bepsilon' (\bX_{\btau}\bX_{\btau}') \bepsilon)}{T}}=\mathcal{O}_p(1)$, we rely on the spectral theorem to decompose $\bX_{\btau}\bX_{\btau}'$ into two orthogonal matrices and a diagonal matrix of eigenvalues. With this decomposition, we can show that   $\epsilon' \frac{\bX_{\btau}\bX_{\btau}'}{T} \epsilon \leq \text{max}_i \lambda_i \frac{\epsilon'\epsilon}{T}$ which is $\mathcal{O}_p(1)$ under Assumption \ref{Xfullrank} and the fact that the variance is bounded.

	\noindent The last term $Q_3$ is defined by
	\begin{align*}
	Q_3 &\equiv \frac{CKm\lambda\alpha_T a_{max}\zeta}{\ln(2)(\rho^2 - 2\rho C \alpha_T)} \\
	&= C\alpha_T^2 \dfrac{\lambda \zeta\dfrac{a_{\max}}{(Km)^{-1}\alpha_T^3}}{\ln(2)\left(\left(\dfrac{\rho}{\alpha_T}\right)^2 - 2C \left(\dfrac{\rho}{\alpha_T}\right)\right)} 
	\end{align*}
	
	\noindent By \ref{rhopositive}, $\displaystyle \left(\frac{\rho}{\alpha_T}\right)^2 - 2C \left(\frac{\rho}{\alpha_T}\right) \to \infty$ and by  \ref{tauisO}, $\displaystyle \lim_{T \to \infty} \lambda \zeta\frac{a_{\max}}{(Km)^{-1}\alpha_T^3} < \infty$.
	Hence
	$Q_3 = o(C\alpha_T^2)$.
	
	\noindent Combining the conditions on $Q_1$, $Q_2$ and $Q_3$ we establish that $$H_{T}(\bu) \geq C^2\alpha_T^2\lambda_{T,\min} + \mathcal{O}_p(C\alpha_T^2) + o(C\alpha_T^2)$$. 
	\noindent It follows that there exists $C_0>0$ is large such that for all $C > C_0$, $\displaystyle \mathbf{P}\left(\inf_{||\bu||_2 =1}  H_{T}(\bu) > 0\right) = 1-\nu$, for $T$ sufficiently large.
\end{Proof}

\subsection{Lemma 2 \label{App:proof2}}
\begin{Proof}
	Let $\bbeta \in \Re^{Km\times 1}$ such that $||\bbeta - \bbeta^*||_2 < C\sqrt{\frac{Km\sigma^2}{T}}$. We consider $\tilde{\bbeta} \in \Re^{Km\times 1}$, where $\tilde{\bbeta}_{A^c} = 0$ and $\tilde{\bbeta}_{A} = \bbeta_{A}$.
	We can notice that 
	\begin{align*}
	||\bbeta - \tilde{\bbeta}||_2 &=  ||\bbeta_{A^c} - \tilde{\bbeta}_{A^c}||_2 = ||\bbeta_{A^c} -  \bbeta^*_{A^c}||_2\\
	||\bbeta - \tilde{\bbeta}||_2 &< C\alpha_T
	\end{align*}
	
	\noindent On the other hand
	\begin{align*}
	||\bbeta^* - \tilde{\bbeta}||_2 &=  ||\bbeta^*_{A} - \tilde{\bbeta}_{A}||_2 =  ||\bbeta^*_{A} -  \bbeta_{A}||_2\\
	||\bbeta^* - \tilde{\bbeta}||_2 &< C\alpha_T
	\end{align*}
	\noindent Let us define $\displaystyle G_{T}(\bbeta) = f_{T}(\bbeta) - f_{T}(\tilde{\bbeta})$. Similarly to the proof of the lemma \ref{lemforcons1}, it suffices to show that $G_{T}(\bbeta,\tilde{\bbeta}) > 0$.
	\begin{align*}
	G_{T}(\bbeta) &=  \dfrac{1}{T}\left(||\by - \bX_{\btau}\bbeta||_2^2 - ||\by - \bX_{\btau}\tilde{\bbeta}||_2^2\right) + \sum_{\substack{
			(j,k) \in A^c \\
			j \geq  2}}\pen(\Delta\beta_{jk}|a_k,\lambda) \\
	&= \dfrac{1}{T}\left(||\by - \bX_{\btau}\tilde{\bbeta} - \bX_{\btau}(\bbeta-\tilde{\bbeta})||_2^2 -||\by - \bX_{\btau}\tilde{\beta}||_2^2\right) + \sum_{\substack{
			(j,k) \in A^c \\
			j \geq  2}}\pen(\Delta\beta_{jk}|a_k,\lambda)  \\
	&= (\bbeta - \tilde{\bbeta})'\dfrac{\bX_{\btau}'\bX_{\btau}}{T}(\bbeta - \tilde{\bbeta}) - 2(\bbeta - \tilde{\bbeta})'\dfrac{\bX_{\btau}'(\by - \bX_{\btau}\tilde{\bbeta})}{T}  + \sum_{\substack{
			(j,k) \in A^c \\
			j \geq  2}}\pen(\Delta\beta_{jk}|a_k,\lambda)  \\
	&= (\bbeta - \tilde{\bbeta})'\dfrac{\bX_{\btau}'\bX_{\btau}}{T}(\bbeta - \tilde{\bbeta}) - 2(\bbeta - \tilde{\bbeta})'\dfrac{\bX_{\btau}'\boldsymbol{\epsilon}}{T} - 2(\bbeta - \tilde{\bbeta})'\dfrac{\bX_{\btau}'\bX_{\btau}}{T}  (\bbeta^* - \tilde{\bbeta} ) \\
	& \quad + \sum_{\substack{
			(j,k) \in A^c \\
			j \geq  2}}\pen(\Delta\beta_{jk}|a_k,\lambda)  \\
	&= (\bbeta - \tilde{\bbeta})'\dfrac{\bX_{\btau}'\bX_{\btau}}{T}(\bbeta - \tilde{\bbeta}) - 2\alpha_T\dfrac{(\bbeta - \tilde{\bbeta})'}{\sqrt{Km\sigma^2}}\dfrac{\bX_{\btau}'\boldsymbol{\epsilon}}{\sqrt{T}} - 2(\bbeta - \tilde{\bbeta})'\dfrac{\bX_{\btau}'\bX_{\btau}}{T}  (\bbeta^* - \tilde{\bbeta} ) \\
	& \quad + \sum_{\substack{
			(j,k) \in A^c \\
			j \geq  2}}\pen(\Delta\beta_{jk}|a_k,\lambda)
	\end{align*}
	
	\noindent By \ref{tauistruetau}, \ref{Xfullrank} and \ref{exogeneity}, $\dfrac{\bX_{\btau}'\bX_{\btau}}{T} = \mathcal{O}_p(1)$ and $\dfrac{\bX_{\btau}'\boldsymbol{\epsilon}}{\sqrt{T}} = \mathcal{O}_p(1)$ (as it is a martingale difference sequence following assumption \ref{exogeneity}). Moreover $||\bbeta - \tilde{\bbeta}||_2 < C\alpha_T$ and
	$||\bbeta^* - \tilde{\bbeta}||_2 \leq C\alpha_T$. Then, for any $T$ sufficiently large 
	$$ G_{T}(\bbeta) =\mathcal{O}_p\left(||\bbeta - \tilde{\bbeta}||_2\alpha_T\right) + \sum_{\substack{
			(j,k) \in A^c \\
			j \geq  2}}\pen(\Delta\beta_{jk}|a_k,\lambda) $$
	\noindent As $\pen(x|a_k,\lambda)$ is a concave function on $x \in ]-\infty, 0]$ and on $x \in [0, +\infty[$, for any $\nu_1 < \nu_2 \leq \nu_3 \leq 0$ (resp. $0 \leq \nu_1 \leq \nu_2 < \nu_3$), $\displaystyle \dfrac{\pen(\nu_1) - \pen(\nu_3) }{\nu_1 - \nu_3} \geq \dfrac{\pen(\nu_2) - \pen(\nu_3) }{\nu_2 - \nu_3}$  (resp. $\displaystyle \dfrac{\pen(\nu_3) - \pen(\nu_1) }{\nu_3 - \nu_1} \leq \dfrac{\pen(\nu_2) - \pen(\nu_1) }{\nu_2 - \nu_1}$).\\
	$\forall~ (j,k) \in A^c$, $\Delta\beta^*_{jk} = 0$ and $\Delta\beta_{jk}$ is strictly  positive or negative. \\Thus, $-C\alpha_T \leq \beta_{jk} < 0$ or $0 < \Delta\beta_{jk} < C\alpha_T$, since $|\Delta\beta_{jk}| \leq ||\bbeta - \bbeta^*||_2 < C\alpha_T$.
	In both cases, we end up with
	\begin{align*}
	\dfrac{\pen\left(C\alpha_T|a_k,\lambda \right)}{C\alpha_T} &\leq \dfrac{\pen\left(\Delta\beta_{jk}|a_k,\lambda\right)}{|\Delta\beta_{jk}|}\\
	\displaystyle \pen(\Delta\beta_{jk}|a_k,\lambda) &\geq \frac{\lambda}{\ln(2)C\alpha_T}\ln\left(\frac{C\alpha_T}{C\alpha_T + a_k\zeta } + 1\right)|\Delta\beta_{jk}|\\
	\displaystyle \pen(\Delta\beta_{jk}|a_k,\lambda) &\geq \frac{\lambda}{\ln(2)C}\ln\left(\frac{C\alpha_T}{C\alpha_T + a_{\max}\zeta } + 1\right)|\Delta\beta_{jk}|
	\end{align*}
	\noindent for any $T$ sufficiently large. Thus 
	\begin{align*}
	\sum_{\substack{
			(j,k) \in A^c \\
			j \geq  2}}\pen(\Delta\beta_{jk}|a_k,\lambda) \geq \frac{\lambda}{\ln(2)C}\ln\left(\frac{C}{C + a_{\max}\zeta \sqrt{\frac{T}{K m \sigma^2}}} + 1\right)||\bbeta - \tilde{\bbeta}||_2.
	\end{align*}
	\noindent Furthermore, by \ref{tauisO} $\displaystyle a_{\max} = \mathcal{O}_p\left(\sqrt{\frac{mK\sigma^2}{T}} \frac{\sigma_2}{T}\right)$. Thereby\\ $$a_{\max}\zeta \sqrt{\frac{T}{K m \sigma^2}} \overset{p}{\to} 0 \text{ and }\displaystyle \liminf_{T \to \infty}  \left(\frac{C}{C + a_{\max}\zeta \sqrt{\frac{T}{K m \sigma^2}}} + 1\right) > 0$$
	\noindent It follows  that, there exists $\tilde{C} > 0$ such that
	$$ \dfrac{G_{T}(\bbeta,\tilde{\bbeta})}{||\bbeta - \tilde{\bbeta}||_2} \geq \tilde{C}\lambda + \mathcal{O}_p\left(\sqrt{\dfrac{Km\sigma^2}{T}}\right) $$
	\noindent Thereby the result follows.
\end{Proof}

\subsection{Proof of the Proposition \ref{theoconsislimitdist}}
\begin{Proof}
The theorem is immediately given by the lemmas (\ref{lemforcons1}) and (\ref{lemforcons2}), in the sense that there exists a sequence of local minima $\hat{\bbeta}$ of $f_{T}(\bbeta)$ such that $||\hat{\bbeta} - \bbeta^*|| = \mathcal{O}_p\left(T^{-\frac{1}{2}}\right)$ (since $m$, $K$ and $\sigma^2$ are constants) and $\hat{\bbeta}_{A^c}  = \boldsymbol{0} \in \REAL^{|A^c|\times 1}$. Thus, as $T^{-\frac{1}{2}} \to 0$, it follows that $||\hat{\bbeta}_{A} - \bbeta_{A}^*||_2 =  o_P(1)$.
\end{Proof}

\subsection{Consequence of bounded eigenvalues \label{App:eig}}
We show that bounded eigenvalues of the matrix $\frac{\bX_{\btau}'\bX_{\btau}}{T}$ implies a fixed number of regimes. Note first that 
\begin{eqnarray}
\bX_{\btau}'\bX_{\btau} & = & \sum_{t=1}^T (\1{t} \otimes \bx_t)(\1{t} \otimes \bx_t)', \\
 & = & \sum_{t=1}^T (\bx_t  \bx_t') \otimes (\1{t} \1{t}'),
\end{eqnarray}
where we define $\1{t} = (\1{t > \tau_0},\1{t > \tau_1},\ldots,\1{t > \tau_{m-1}})'$. Let us define $n_i = \sum_{t=1}^T \1{t > \tau_{i-1}}$, i.e., the number of observations from the beginning of regime $i$ to the end of the sample. Working with $\bx_t \equiv 1$, we have that 
\begin{eqnarray}
\frac{\bX_{\btau}'\bX_{\btau}}{T} & = & \frac{1}{T} \sum_{t=1}^T (\1{t} \1{t}'),\\
& = & \frac{1}{T}  \begin{pmatrix} n_1 & n_2 & n_3 & \ldots & n_m \\ n_2 & n_2 & n_3 & \ldots & n_m \\ n_3 & n_3 & n_3 & \ldots & n_m \\ & & &\ldots & \\n_m & n_m & n_m & \ldots & n_m \end{pmatrix}
\end{eqnarray}
in which $n_1 = T$. It leads to the following determinant, when $m>1$,
\begin{eqnarray}
|\frac{\bX_{\btau}'\bX_{\btau}}{T}| & = &T^{-m} n_m  \prod_{i=1}^{m-1}(n_{i}-n_{i+1}), \\
& = & \frac{n_m}{T}  \prod_{i=1}^{m-1} \frac{(n_{i}-n_{i+1})}{T}.
\end{eqnarray}
Note that $n_i = T - \tau_{i - 1} = \tau_{m} - \tau_{i - 1}$, for $i = 1, \dots, m$. Thus,   
\begin{eqnarray}
|\frac{\bX_{\btau}'\bX_{\btau}}{T}| &= & \prod_{i=1}^{m} \frac{\tau_{i} - \tau_{i - 1}}{T}\\
 &= & \prod_{i=1}^{m} \delta_{\tau_i} > 0
\end{eqnarray}
\noindent where $\delta_{\tau_i} = \dfrac{\tau_i - \tau_{i - 1}}{T}$.\\
It shows that the number of segments cannot increase with $T$ otherwise the determinant tends to zero. Let us assume that $m = \mathcal{O}(T^q)$ with $q>0$. It is clear that when T tends to $\infty$, there exists $r \in \mathbb{N}$ such that $\delta_{\tau_r}= \mathcal{O}(T^x)$ where $x < 0$. If such a $r$ does not exist, this would imply that $\forall i, ~ \delta_{\tau_i}$ does not drift to 0 as $T \rightarrow \infty$ and then $m = \mathcal{O}(1)$ because $\displaystyle\inf_{i}\{\delta_{\tau_i}\}m \leq 1$. But, because $m = \mathcal{O}(T^q)$ with $q>0$, there exists $r \in \mathbb{N}$ such that $\delta_{\tau_r} = \mathcal{O}(T^x)$ where $x < 0$. Therefore, we have that
\begin{eqnarray}
|\frac{\bX_{\btau}'\bX_{\btau}}{T}| &= & \delta_{\tau_r} \prod_{\substack{i = 1 \\ i \ne r}}^{m} \delta_{\tau_i} \to 0,
\end{eqnarray}
which contradicts Assumption \ref{Xfullrank}. Hence $m = \mathcal{O}(1)$. As a result  \ref{Xfullrank} implies that $m < \infty$. In addition, \ref{Xfullrank} also implies that $\displaystyle\min_{i}\{\delta_{\tau_i}\} > 0$; that is $\delta_{\tau_i}$ does not drift to 0 as $T \rightarrow \infty$.

\subsection{Approximation of the penalty function with mixture of normal densities \label{App:mixture}}
To derive the DAEM algorithm, a mixture of two normal densities has been assumed for the mean parameter. We now provide a simple mixture  approximation of the SELO penalty. Note that in practice, one can use the output of the DAEM algorithm as a starting point to optimize the function of Equation \eqref{eq:optimSELO}. Due to the mixture approximation and the continuity of the SELO penalty function, the starting point would be in general very close to the value that globally minimizes the function \eqref{eq:optimSELO}. We now use a mixture of two normal densities that can be understood as a spike and slab prior in the Bayesian paradigm. In particular, we calibrate a spike and slab prior \citep[see, e.g.,][]{George_1993} to the SELO penalty function. Given a mean parameter $\beta$, the spike and slab prior is specified as, 
\begin{align}
\begin{split}\label{eq:standardSS}
\beta & \sim  \mathcal{N}(0,r_{z}), \\
z & \sim  \text{Bernoulli}(1-\omega),
\end{split}
\end{align}
where $r_0<r_1$ such that the spike distribution arises with probability $P[z=0|\omega]=\omega$. By marginalizing out $z$, we get a mixture of two normal densities given by
\begin{align}
\begin{split}
f(\beta|\omega,r_{0},r_{1}) & = \omega  (\frac{1}{ r_0})^{\frac{1}{2}}f_Z(\frac{\beta }{\sqrt{ r_0}}) + (1-\omega)(\frac{1}{ r_1})^{\frac{1}{2}}f_Z(\frac{\beta}{\sqrt{ r_1}}).
\end{split}
\end{align} 
in which $f_Z(x)$ denotes the Normal density function evaluated at $x$ and with expectation and variance equal to 0 and 1, respectively. The calibration is done as follows:
\begin{align}
\begin{split}
 c \equiv \frac{r_1}{r_0}  & = 10000, \\
\omega& = \frac{(\exp(\lambda) - 1)}{(\sqrt{c}+(\exp(\lambda) - 1))}, \\
r_0 & = \frac{a^2}{8}\frac{1-c^{-1}}{|\ln(\exp(\lambda) -1)|}
\end{split}
\end{align}
We now detail how we come up with this simple calibration. Given $\beta$, note that the probability of being in the slab component is equal to
\begin{align}
\begin{split}
Pr(z = 1|\beta,\omega,r) & = \frac{1}{\frac{\omega}{(1-\omega)} \frac{(\frac{1}{r_0})^{\frac{1}{2}}f_Z(\frac{\beta}{\sqrt{ r_0}})}{(\frac{1}{ r_1})^{\frac{1}{2}}f_Z(\frac{\beta }{\sqrt{ r_1}})} + 1}. \\
\end{split}
\end{align}
To mimic the SELO penalty function, we impose the following constraints on the spike and slab hyper-parameters.
\begin{enumerate}
	\item As standards in the spike and slab literature, we fix $c \equiv \frac{r_1}{r_0} = 10000$ \cite[see, e.g.][]{Malsiner_SS2016}.
\item The SELO function imposes a penalty equal to $\pen(a)-\pen(0) = \lambda y$ with $y=0.99$. To fix the same penalty value (neglecting $y$ because $y\approx 1$), we set
\begin{align}
\begin{split}
(-\lambda) & = \ln f(\beta=a|\omega,r_{z}) - \ln f(\beta=0|\omega,r_{z}) , \\
 & = \ln \frac{f(\beta=a|\omega,r_{1},z=1)}{f(\beta=0|\omega,r_{1},z=1)} + \ln \frac{Pr(z = 1|\beta=0,\omega,r)}{Pr(z = 1|\beta=a,\omega,r)}, \\
 & = -\frac{a^2}{2r_1}  + \ln \frac{Pr(z = 1|\beta=0,\omega,r)}{Pr(z = 1|\beta=a,\omega,r)}, \\
 & \approx \ln \frac{Pr(z = 1|\beta=0,\omega,r)}{Pr(z = 1|\beta=a,\omega,r)},  \text{ because } r_1>>a^2, \\
 & \approx \ln Pr(z = 1|\beta=0,\omega,r), \text{ because } Pr(z = 1|\beta=a,\omega,r)\approx 1,\\
 & = -\ln (\frac{\omega}{(1-\omega)} \sqrt{c}+ 1).
\end{split}
\end{align}
\item Finally, we impose that $Pr(z = 1|\beta,\omega,r)=Pr(z = 0|\beta,\omega,r)$ when $\beta = \frac{a}{2}$ (this is called the intersection point in \cite{rovckova2018spike}). This means that the slab component starts to dominate when $|\beta|>\frac{a}{2}$. It leads to the constraints:
	\begin{align}
\begin{split}
0 & = (2\pi)^{-\frac{1}{2}} [\frac{\omega}{\sqrt{r_0}} \exp(-\frac{a^2}{8r_0}) - \frac{1-\omega}{\sqrt{r_1}} \exp(-\frac{a^2}{8r_1})], \\
r_0 & = \frac{a^2}{8}\frac{1-c^{-1}}{|\ln(\exp(\lambda) -1)|}.
	\end{split}
\end{align}
\end{enumerate}

Figure \ref{fig:SS} shows the penalty imposed by the SELO function and by the calibrated spike and slab prior for several values of $\lambda$ and $a$. We observe that the spike and slab prior provides a good approximation of the penalty function.
\renewcommand{\baselinestretch}{1}
\begin{figure}[h!]
	\begin{center}
		\subfloat[$\lambda=1$, $a=0.01$]{\includegraphics[width=7.5cm,height=5cm]{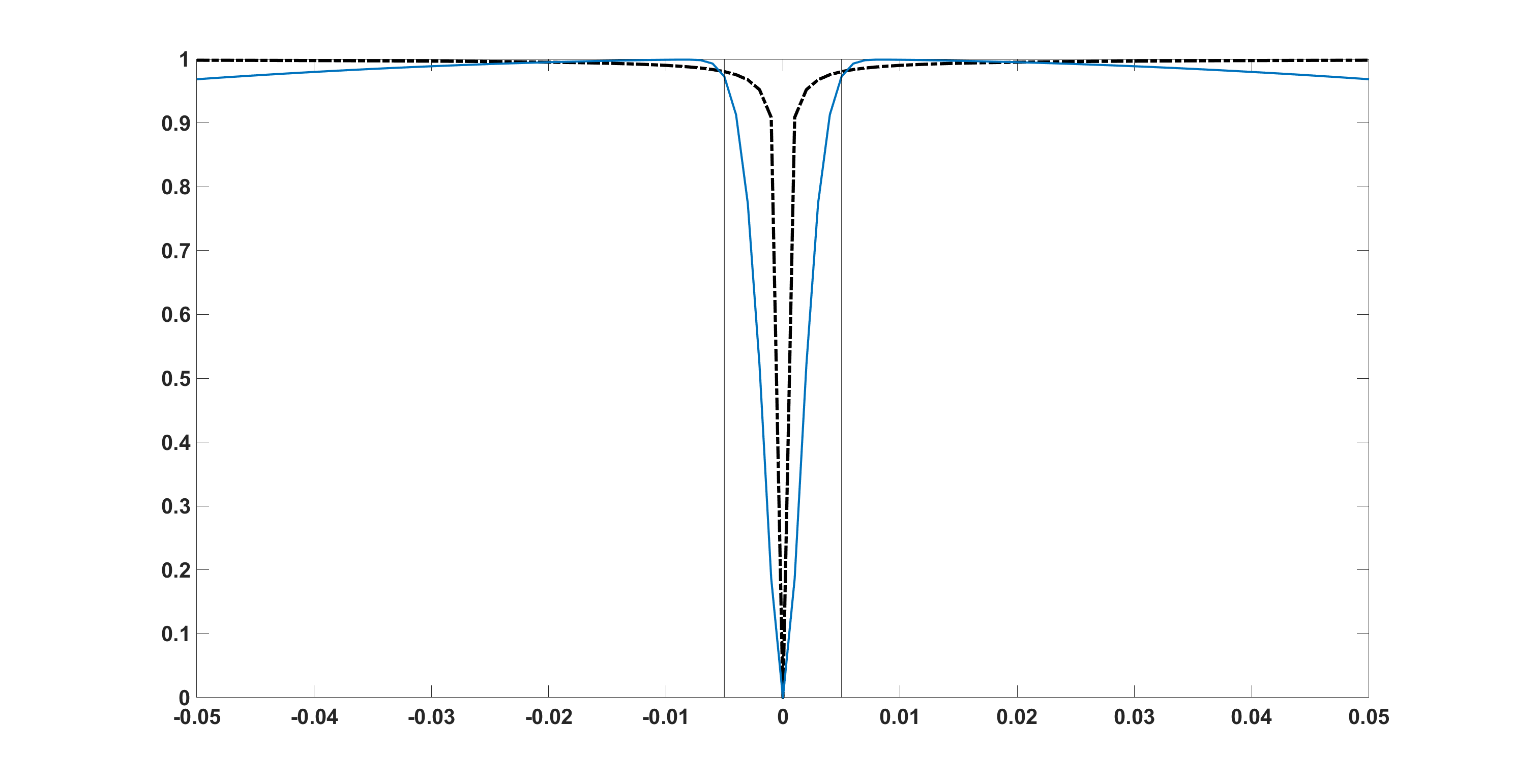}}
		\subfloat[$\lambda=1$, $a=0.1$]{\includegraphics[width=7.5cm,height=5cm]{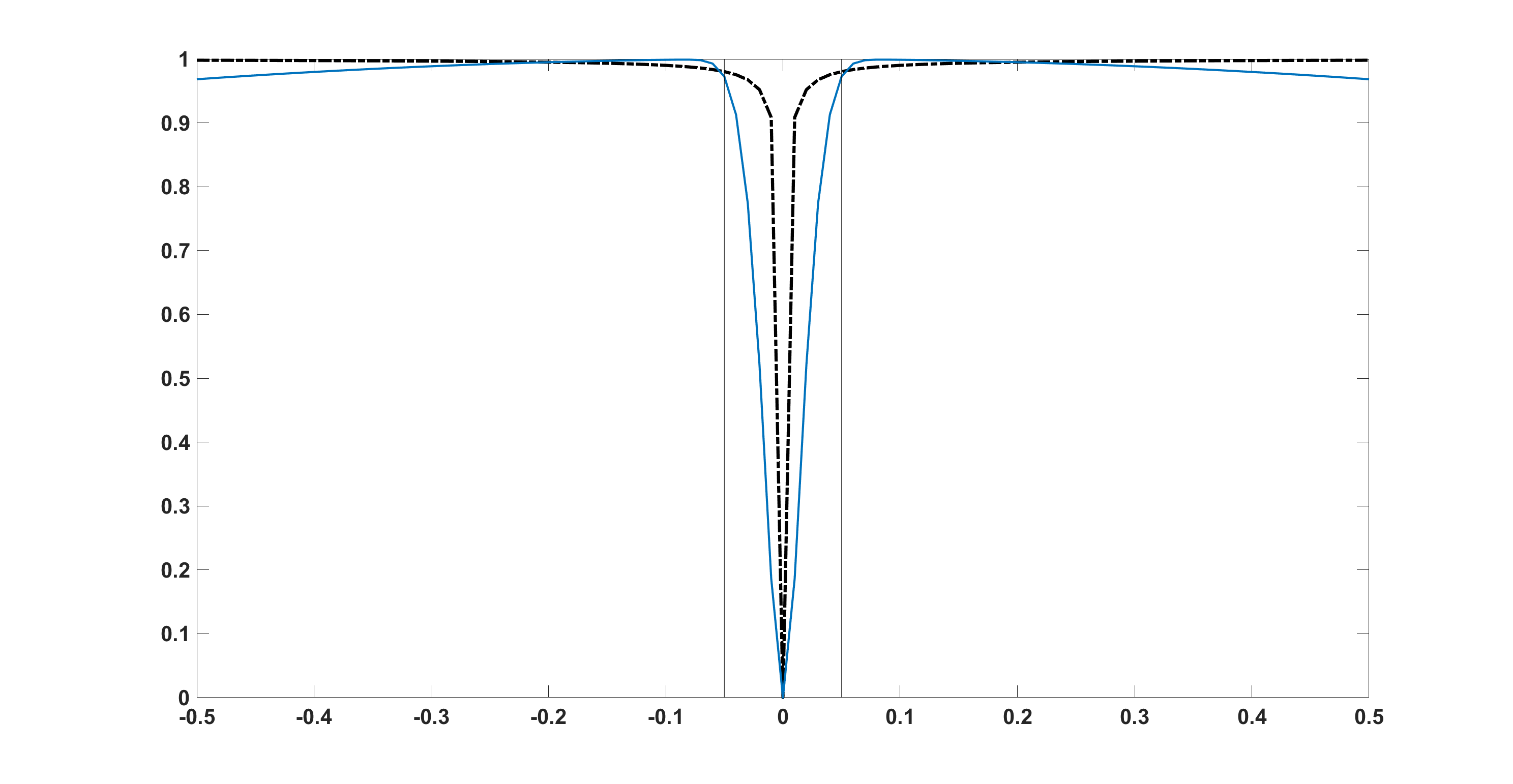}}\\
		\subfloat[$\lambda=5$, $a=0.01$]{\includegraphics[width=7.5cm,height=5cm]{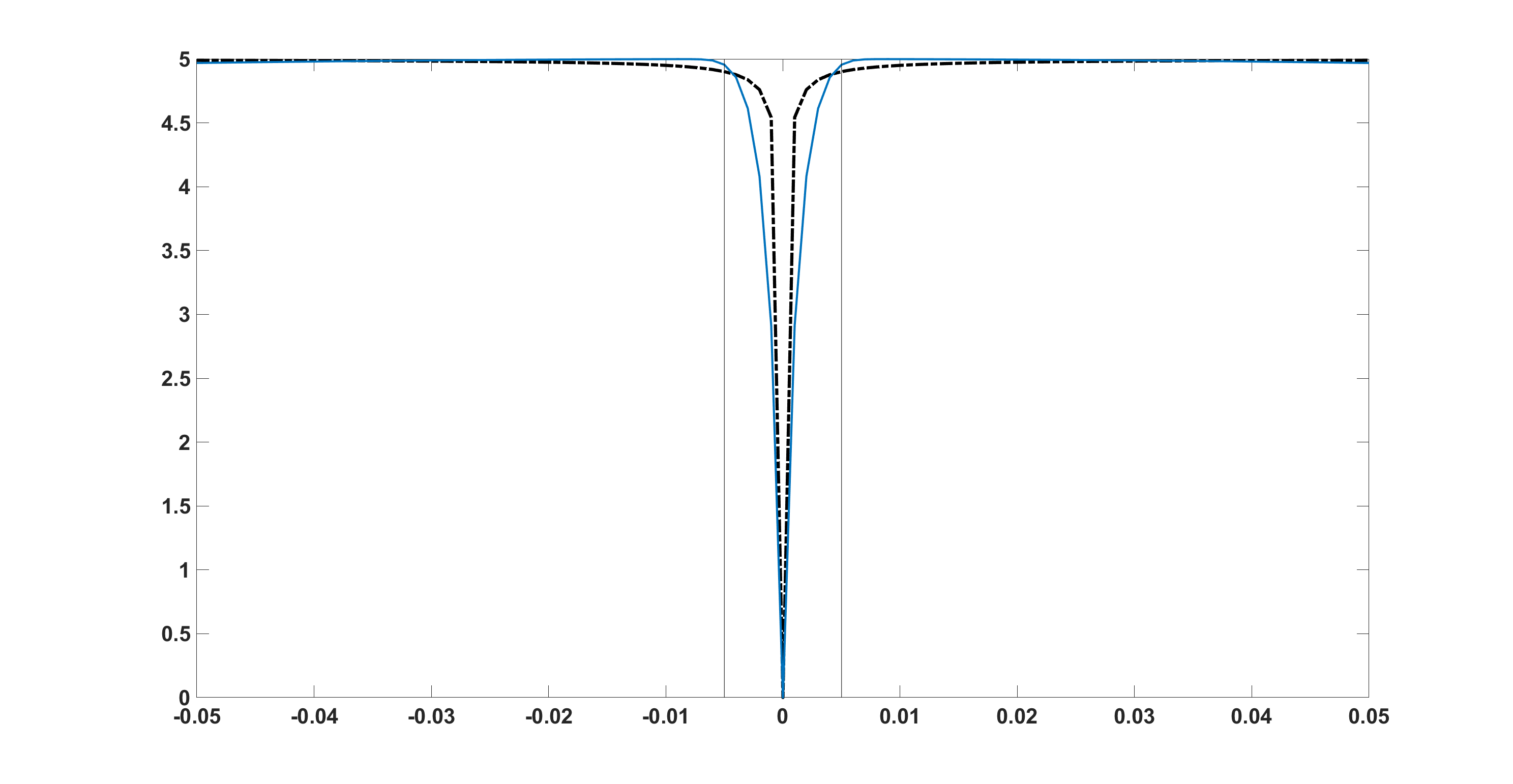}}
		\subfloat[$\lambda=5$, $a=0.1$]{\includegraphics[width=7.5cm,height=5cm]{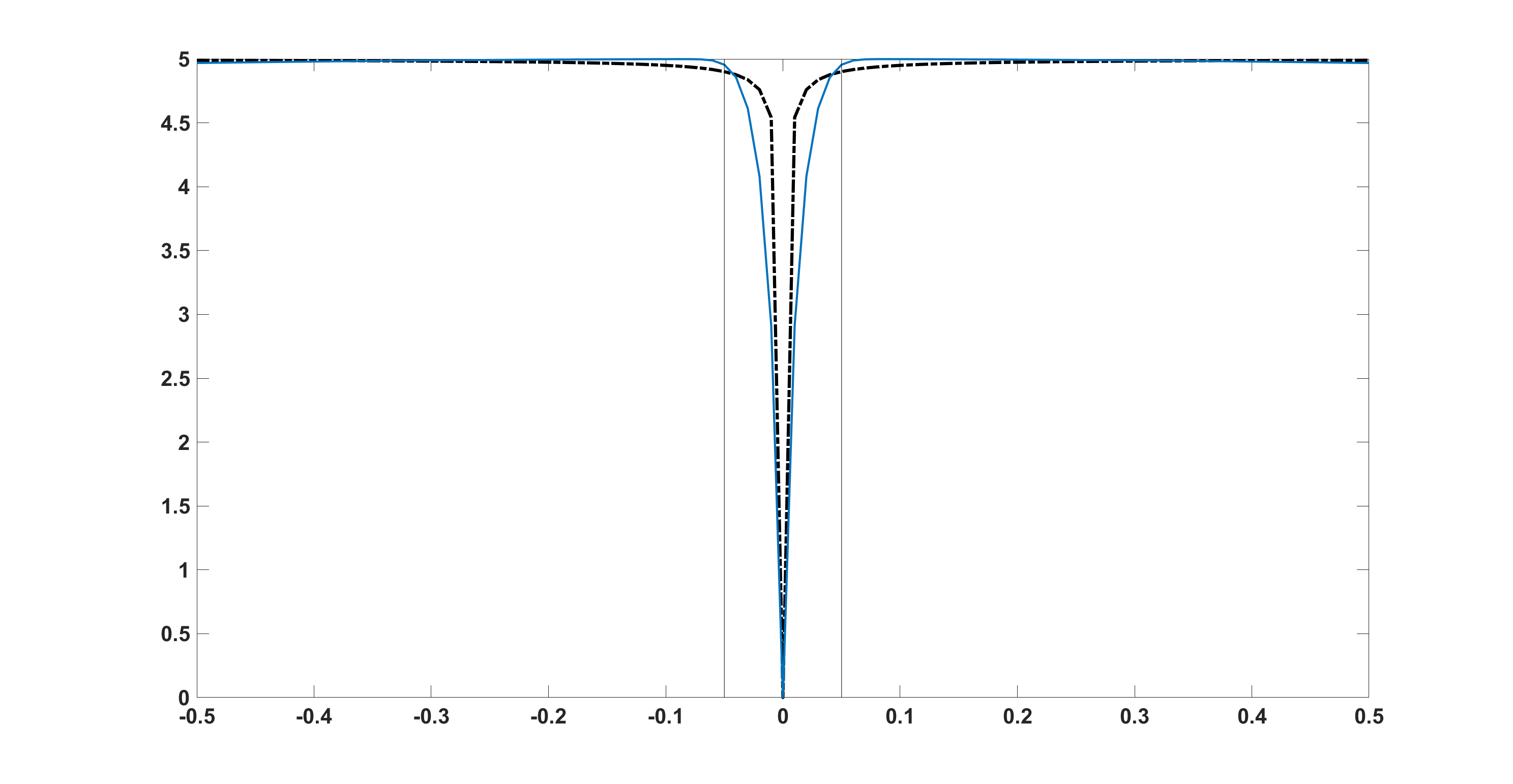}}\\
		\subfloat[$\lambda=10$, $a=0.01$]{\includegraphics[width=7.5cm,height=5cm]{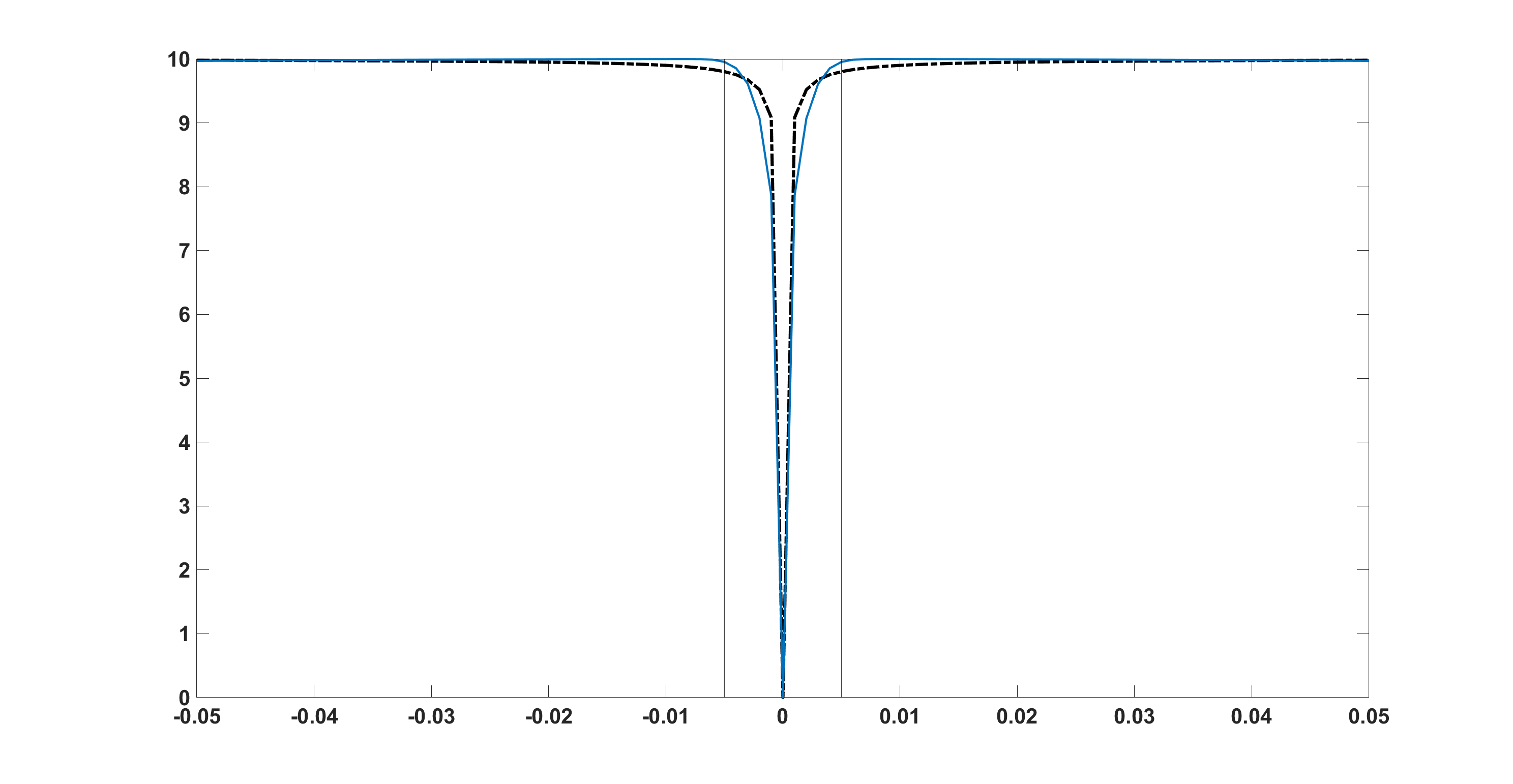}}
		\subfloat[$\lambda=10$, $a=0.1$]{\includegraphics[width=7.5cm,height=5cm]{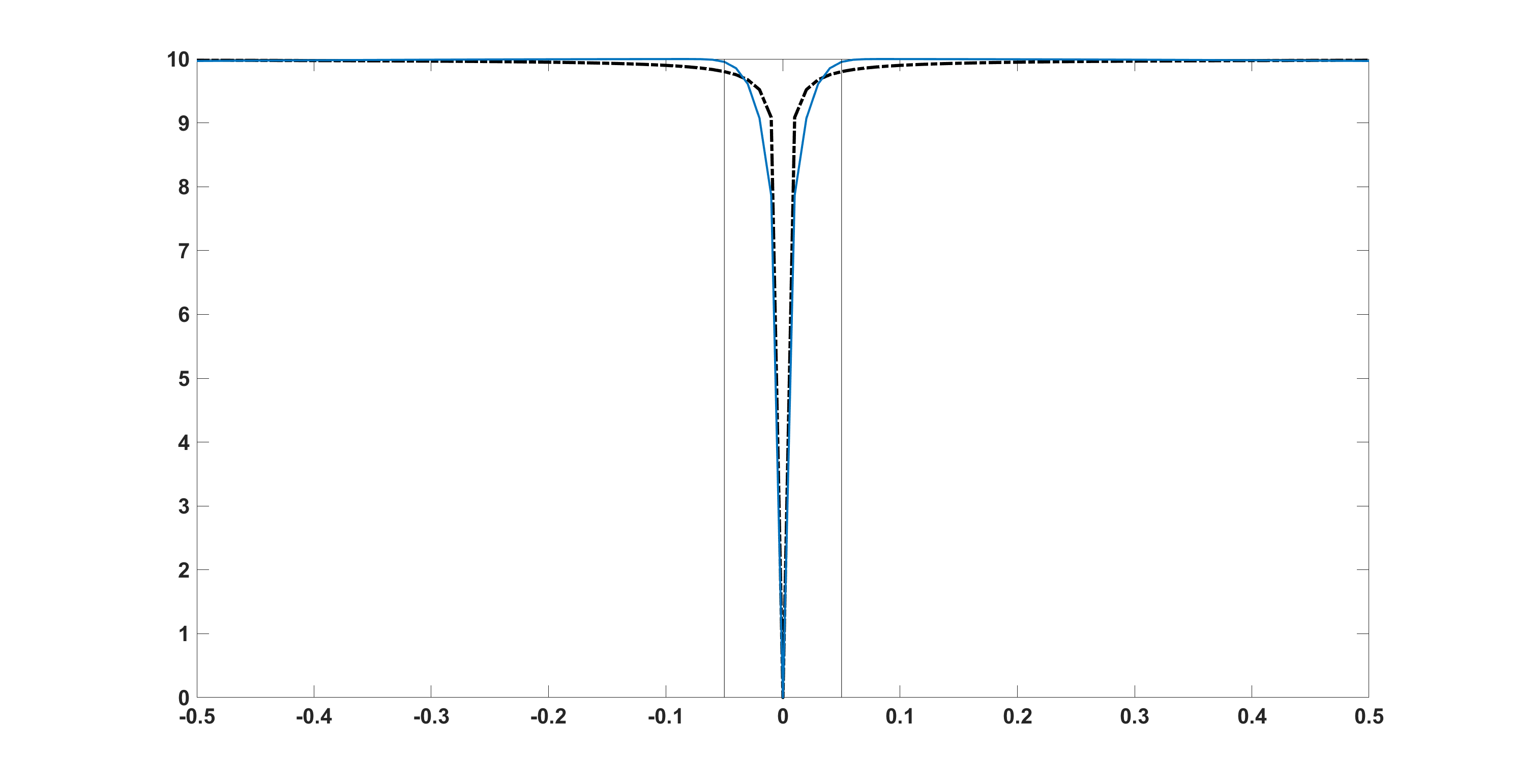}}\\
		\caption{\justifying Penalty imposed by the SELO function in dotted black line and the spike and slab prior in blue for several values of $\lambda$ and $a$. Vertical lines are the intersection points of the spike and slab densities (at $|x|=\frac{a}{2}$). \label{fig:SS}}
	\end{center}
\end{figure}
\renewcommand{\baselinestretch}{1.5}

\section{Marginal likelihood~\eqref{eq:ML} for the linear model} \label{app:linreg}
Let us derive the criterion \eqref{eq:ML}. We first define $X_1 = \tilde{\bX}_{\tau_0}$, $X_2=X_{\btau}^{\hat{A}}$ and $M_{X_1} = \bM_{\tilde{\bX}_{\tau_0}}$. Given the prior distributions in Equation \eqref{eq:prior}, the marginal likelihood is given by,
\begin{eqnarray*}
f(\by|\ba,\lambda,\btau) &= & \int \int (2\pi)^{\frac{-(T+k_{\hat{A}})}{2}}(\sigma^{2})^{-\frac{(T+2+k_{\hat{A}})}{2}}|g_{\hat{A}} (X_2)' M_{X_1} X_2)|^{1/2}\\
& & \hspace{-2cm} \exp \frac{-1}{2\sigma^2}\{\underbrace{(y-X_1 \bbeta_1 - X_2\Delta \bbeta)'(y-X_1 \bbeta_1 - X_2\Delta \bbeta) + \Delta \bbeta'g_{\hat{A}} (X_2)' M_{X_1} X_2) \Delta \bbeta}_{B}\}d(\bbeta_1,\Delta \bbeta) d\sigma^2.
\end{eqnarray*}
Focusing on the expression in the exponential, we have
\begingroup
\allowdisplaybreaks
\begin{align*}
B & =   (y-X_1 \bbeta_1 - X_2\Delta \bbeta)'(y-X_1 \bbeta_1 - X_2\Delta \bbeta) + \Delta \bbeta'g_{\hat{A}} \underbrace{ (X_2' M_{X_1} X_2)}_{\Sigma_X} \Delta \bbeta,\\
 B& =  (y- X_2\Delta \bbeta)'(y- X_2\Delta \bbeta) + \Delta \bbeta'g_{\hat{A}} \underbrace{ (X_2' M_{X_1} X_2)}_{\Sigma_X} \Delta \bbeta + \bbeta_1'X_1'X_1\bbeta_1 -2\bbeta_1'X_1'(y- X_2\Delta \bbeta),\\
B& =  (y- X_2\Delta \bbeta)'(y- X_2\Delta \bbeta) + \Delta \bbeta'g_{\hat{A}}\Sigma_X\Delta \bbeta + (\bbeta_1-\bar{\bbeta_1})'\Omega^{-1}(\bbeta_1-\bar{\bbeta_1}) - \bar{\bbeta_1}'\Omega^{-1}\bar{\bbeta_1},
\end{align*}
\endgroup
where $\Omega^{-1}=X_1'X_1$, $\bar{\bbeta_1}=(X_1'X_1)^{-1}X_1' (y- X_2\Delta \bbeta)$ and $\bar{\bbeta_1}'\Omega^{-1}\bar{\bbeta_1} = (y- X_2\Delta \bbeta)'X_1(X_1'X_1)^{-1}X_1' (y- X_2\Delta \bbeta) = (y- X_2\Delta \bbeta)'P_{X_1}(y- X_2\Delta \bbeta)$. The marginal likelihood can be simplifies as 
\begin{align*}
&f(\by|\ba,\lambda,\btau) =  |X_1'X_1|^{-\frac{1}{2}} \int \int (2\pi)^{\frac{-(T+k_{\hat{A}}-K)}{2}}(\sigma^{2})^{-\frac{(T+2+k_{\hat{A}}-K)}{2}}|g_{\hat{A}}\Sigma_X|^{1/2}\\
& \quad \exp \frac{-1}{2\sigma^2}\{\underbrace{(y- X_2\Delta \bbeta)'(y- X_2\Delta \bbeta) + \Delta \bbeta'g_{\hat{A}}\Sigma_X\Delta \bbeta - (y- X_2\Delta \bbeta)'P_{X_1}(y- X_2\Delta \bbeta)}_{C}\}d(\Delta \bbeta) d\sigma^2.
\end{align*}
Again, focusing on the expression of the exponential, we obtain
\begin{eqnarray*}
C & = & (y- X_2\Delta \bbeta)'(y- X_2\Delta \bbeta) + \Delta \bbeta'g_{\hat{A}}\Sigma_X\Delta \bbeta - (y- X_2\Delta \bbeta)'P_{X_1}(y- X_2\Delta \bbeta),\\
 & = & y'[I_T-P_{X_1}]y + \Delta \bbeta'[g_{\hat{A}} X_2'M_{X_1}X_2 + X_2'X_2 - X_2'P_{X_1}X_2]\Delta \bbeta -2\Delta \bbeta'X_2'[I_T - P_{X_1}]y,\\
 & = & y'M_{X_1}y + \Delta \bbeta'[(1+g_{\hat{A}})X_2'M_{X_1}X_2]\Delta \bbeta -2\Delta \bbeta'X_2'M_{X_1}y,\\
& = & y'M_{X_1}y + (\Delta \bbeta-\bar{\bmu})'\bar{\Sigma}^{-1}(\Delta \bbeta-\bar{\bmu}) - \bar{\bmu}'\bar{\Sigma}^{-1}\bar{\bmu},
\end{eqnarray*}
where $\bar{\Sigma}^{-1} = (1+g_{\hat{A}})X_2'M_{X_1}X_2 = (1+g_{\hat{A}})\Sigma_X$ and $\bar{\bmu} = \bar{\Sigma} X_2'M_{X_1}y$, $\bar{\bmu}'\bar{\Sigma}^{-1}\bar{\bmu} = (1+g_{\hat{A}})^{-1}y'M_{X_1}X_2 [X_2'M_{X_1}X_2]^{-1}X_2'M_{X_1}y$.

Eventually, we find the following marginal likelihood
\begin{eqnarray*}
f(\by|\ba,\lambda,\btau) &= & (2\pi)^{\frac{-(T-K)}{2}} |X_1'X_1|^{-\frac{1}{2}} |g_{\hat{A}} \Sigma_X|^{1/2} |(1+g_{\hat{A}})\Sigma_X|^{\frac{-1}{2}}  \int (\sigma^{2})^{-\frac{(T+2-K)}{2}}\\
& & \exp \frac{-1}{2\sigma^2}\{y'M_{X_1}y-(1+g_{\hat{A}})^{-1}y'M_{X_1}X_2 [X_2'M_{X_1}X_2]^{-1}X_2'M_{X_1}y\} d\sigma^2,\\
& = & (\pi)^{\frac{-(T-K)}{2}} \Gamma(\frac{T-K}{2})|X_1'X_1|^{-\frac{1}{2}} \\
& & (\frac{g_{\hat{A}}}{1+g_{\hat{A}}})^{k_{\hat{A}}/2} [y'M_{X_1}y-(1+g_{\hat{A}})^{-1}y'M_{X_1}X_2 [X_2'M_{X_1}X_2]^{-1}X_2'M_{X_1}y]^{-\frac{T-K}{2}}, \\
& = & (\pi)^{\frac{-(T-K)}{2}} \Gamma(\frac{T-K}{2})|X_1'X_1|^{-\frac{1}{2}} \\
& & (\frac{g_{\hat{A}}}{1+g_{\hat{A}}})^{k_{\hat{A}}/2} [\frac{g_{\hat{A}}}{1+g_{\hat{A}}}y'M_{X_1}y+\frac{1}{(1+g_{\hat{A}})}[\tilde{y}'\tilde{y}-\tilde{y}'X_2[X_2'M_{X_1}X_2]^{-1}X_2\tilde{y}]]^{-\frac{T-K}{2}}, \\
& = & (\pi)^{\frac{-(T-K)}{2}} \Gamma(\frac{T-K}{2})|X_1'X_1|^{-\frac{1}{2}} (\frac{g_{\hat{A}}}{1+g_{\hat{A}}})^{k_{\hat{A}}/2} [\frac{g_{\hat{A}}}{1+g_{\hat{A}}}s_{X_1}+\frac{1}{(1+g_{\hat{A}})}s_{X_1,X_2}]^{-\frac{T-K}{2}} ,
\end{eqnarray*}
where the penultimate equality comes from the Frisch-Waugh theorem.

\subsection{Posterior distribution \label{App:post}}
\begin{eqnarray*}
f(\bbeta_1,\Delta \bbeta, \sigma^2|\by,\btau) & \propto & (2\pi)^{\frac{-(T+k_{\hat{A}})}{2}}(\sigma^{2})^{-\frac{(T+2+k_{\hat{A}})}{2}}|g_{\hat{A}} (X_2)' M_{X_1} X_2)|^{1/2}\\
& & \exp \frac{-1}{2\sigma^2} \left(\underbrace{(y-X_1 \bbeta_1 - X_2\Delta \bbeta)'(y-X_1 \bbeta_1 - X_2\Delta \bbeta) + \Delta \bbeta'g_{\hat{A}} (X_2)' M_{X_1} X_2 \Delta \bbeta}_{\text{Exp}}\right)
\end{eqnarray*}
Focusing on the expression of the exponential, we have
\begin{eqnarray*}
\text{Exp} & = & (y- X_2\Delta \bbeta)'(y- X_2\Delta \bbeta) + \Delta \bbeta'g_{\hat{A}}\Sigma_X\Delta \bbeta + (\bbeta_1-\bar{\bbeta_1})'\Omega^{-1}(\bbeta_1-\bar{\bbeta_1}) - \bar{\bbeta_1}'\Omega^{-1}\bar{\bbeta_1},\\
 & = & y'M_{X_1}y - \bar{\bmu}'\bar{\Sigma}^{-1}\bar{\bmu} + (\Delta \bbeta-\bar{\bmu})'\bar{\Sigma}^{-1}(\Delta \bbeta-\bar{\bmu}) + (\bbeta_1-\bar{\bbeta_1})'\Omega^{-1}(\bbeta_1-\bar{\bbeta_1}),\\
& = & \frac{g_{\hat{A}}}{1+g_{\hat{A}}}s_{X_1}+\frac{1}{(1+g_{\hat{A}})}s_{X_1,X_2} + (\Delta \bbeta-\bar{\bmu})'\bar{\Sigma}^{-1}(\Delta \bbeta-\bar{\bmu}) + (\bbeta_1-\bar{\bbeta_1})'\Omega^{-1}(\bbeta_1-\bar{\bbeta_1}),\\
\end{eqnarray*}
where $\bar{\Sigma}^{-1} = (1+g_{\hat{A}})X_2'M_{X_1}X_2 = (1+g_{\hat{A}})\Sigma_X$ and $\bar{\bmu} = \bar{\Sigma} X_2'M_{X_1}y$, $\bar{\bmu}'\bar{\Sigma}^{-1}\bar{\bmu} = (1+g_{\hat{A}})^{-1}y'M_{X_1}X_2 [X_2'M_{X_1}X_2]^{-1}X_2'M_{X_1}y$ and $\Omega^{-1}=X_1'X_1$, $\bar{\bbeta_1}=(X_1'X_1)^{-1}X_1' (y- X_2\Delta \bbeta)$. The posterior distribution can be decomposed as 
\begin{eqnarray*}
f(\bbeta_1,\Delta \bbeta, \sigma^2|\by,\btau) & = & f(\sigma^2|\by,\btau) f(\Delta \bbeta|\by,\btau,\sigma^2) f(\bbeta_1|\by,\btau,\sigma^2,\Delta \bbeta)\\
 & \propto & (\sigma^{2})^{-\frac{(T+2-K)}{2}} \exp \frac{-1}{\sigma^2}\{\frac{\frac{g_{\hat{A}}}{1+g_{\hat{A}}}s_{X_1}+\frac{1}{(1+g_{\hat{A}})}s_{X_1,X_2}}{2}\}\\
& & (\sigma^{2})^{-\frac{(k_{\hat{A}})}{2}}|g_{\hat{A}} (X_2)' M_{X_1} X_2)|^{1/2} \exp \frac{-1}{2\sigma^2}\{(\Delta \bbeta-\bar{\bmu})'\bar{\Sigma}^{-1}(\Delta \bbeta-\bar{\bmu})\}\\
& & (\sigma^{2})^{-\frac{(K)}{2}}\exp \frac{-1}{2\sigma^2}\{(\bbeta_1-\bar{\bbeta_1})'\Omega^{-1}(\bbeta_1-\bar{\bbeta_1})\}.\\
\end{eqnarray*}
It gives the following posterior distribution
\begin{eqnarray*}
\sigma^2|\by,\btau & \sim & \mathcal{IG}(\frac{T-K}{2},\frac{\frac{g_{\hat{A}}}{1+g_{\hat{A}}}s_{X_1}+\frac{1}{(1+g_{\hat{A}})}s_{X_1,X_2}}{2}),\\
\Delta \bbeta|\by,\btau,\sigma^2 & \sim & \NORM((1+g_{\hat{A}})^{-1}[X_2'M_{X_1}X_2]^{-1}X_2'M_{X_1}y, \sigma^2 (1+g_{\hat{A}})^{-1}[X_2'M_{X_1}X_2]^{-1}),  \\
\bbeta_1|\by,\btau,\sigma^2,\Delta \bbeta & \sim & \mathcal N((X_1'X_1)^{-1}X_1' (y- X_2\Delta \bbeta), \sigma^2(X_1'X_1)^{-1}).
\end{eqnarray*}

\subsection{Predictive density \label{App:predictive}}
In Appendix \ref{App:post}, we derive the following posterior distributions:
\begin{eqnarray*}
\sigma^2|\by,\btau & \sim & \mathcal{IG}(\underbrace{\frac{T-K}{2}}_{a_{\sigma^2}},\underbrace{\frac{\frac{g_{\hat{A}}}{1+g_{\hat{A}}}s_{X_1}+\frac{1}{(1+g_{\hat{A}})}s_{X_1,X_2}}{2}}_{b_{\sigma^2}}),\\
\Delta \bbeta|\by,\btau,\sigma^2 & \sim & \NORM(\underbrace{(1+g_{\hat{A}})^{-1}[X_2'M_{X_1}X_2]^{-1}X_2'M_{X_1}y}_{\bmu_{\Delta \bbeta }}, \underbrace{\sigma^2 (1+g_{\hat{A}})^{-1}[X_2'M_{X_1}X_2]^{-1}}_{\bSigma_{\Delta \bbeta}}),  \\
\bbeta_1|\by,\btau,\sigma^2,\Delta \bbeta & \sim & \mathcal N(\underbrace{(X_1'X_1)^{-1}X_1' (y- X_2\Delta \bbeta)}_{\bmu_{\bbeta}}, \underbrace{\sigma^2(X_1'X_1)^{-1}}_{\bSigma_{\bbeta}}).
\end{eqnarray*}
Given these results, we can derive the joint posterior distribution of the variable $\bpsi  = \begin{pmatrix} \bbeta_1 \\ \Delta \bbeta \end{pmatrix}$. In particular, a standard algebraic calculus leads to
\begin{align}
\begin{split}
\bpsi|\by,\btau,\sigma^2 & \sim \NORM\left(\begin{pmatrix} \hat{\bbeta}_1 - \bB \mu_{\Delta \bbeta } \\ \bmu_{\Delta \bbeta }\end{pmatrix}, \begin{pmatrix}\bSigma_{\bbeta}^{-1} & \bSigma^{-1}_{\beta}\bB \\ \bB'\bSigma^{-1}_{\beta}& [\bSigma_{\Delta \bbeta }^{-1}+\bB'\bSigma_{\bbeta}^{-1}\bB]\end{pmatrix}^{-1}\right), \\
 & \sim \NORM\left( \bmu_{\bpsi},\bSigma_{\bpsi} \right)
\end{split}
\end{align}
with $\hat{\bbeta}_1 = (X_1'X_1)^{-1}X_1'y$ and $\bB= (X_1'X_1)^{-1}X_2$. Consequently, the predictive density is given by
\begin{align}
\begin{split}
y_{T+1}|\by,\btau,\sigma^2 & \sim \NORM\left(\bx_{T+1}'\bmu_{\bpsi},\sigma^2(\bx_{T+1}'\bSigma_{\psi}\bx_{T+1} + 1)\right).
\end{split}
\end{align}
Since $\sigma^2|\by,\btau$ follows an inverse gamma distribution, the predictive distribution of $y_{T+1}|\by$ is a student distribution. Its density is given by
\begin{align}
\begin{split} \label{eq:student}
f(y_{T+1}|\by,\btau) & = \frac{b_{\sigma^2}^{a_{\sigma^2}}}{\Gamma(a_{\sigma^2})}(2\pi (\bx_{T+1}'\bSigma_{\bpsi}\bx_{T+1} + 1))^{-\frac{1}{2}} \\
& \quad \int (\sigma^2)^{-(a_{\sigma^2}+1+0.5)} \exp \left(-\frac{1}{\sigma^2}[\frac{(y_{T+1}-\bx_{T+1}'\bmu_{\bpsi})^2 (\bx_{T+1}'\bSigma_{\bpsi}\bx_{T+1} + 1)^{-1} + 2b_{\sigma^2}}{2}] \right)d\sigma^2, \\
 & = \frac{b_{\sigma^2}^{a_{\sigma^2}}}{\Gamma(a_{\sigma^2})}(2\pi (\bx_{T+1}'\bSigma_{\bpsi}\bx_{T+1} + 1))^{-\frac{1}{2}} \Gamma(a_{\sigma^2}+0.5) \\
 & \quad \left(\frac{(y_{T+1}-\bx_{T+1}'\bmu_{\bpsi})^2(\bx_{T+1}'\bSigma_{\bpsi}\bx_{T+1} + 1)^{-1} + 2b_{\sigma^2}}{2}\right)^{-(a_{\sigma^2}+0.5)},\\
\end{split}
\end{align}
The final expression in Equation \eqref{eq:student} is equivalent to a student density with expectation $\bx_{T+1}'\bmu_{\bpsi}$, scale parameter $\frac{b_{\sigma^2}}{a_{\sigma^2}}(\bx_{T+1}'\bSigma_{\psi}\bx_{T+1} + 1)$ and degree of freedom equal to $2a_{\sigma^2}$.
\section{Consistency of the criterion} \label{app:consistency}
To prove the theorem, we focus on the ratio of the criterion for two different models $s=(a_s,\lambda_s)$ and $j=(a_j,\lambda_j)$ where $s$ is considered as the true model. To simplify the notation, we denote by $\bX_z$ the explanatory variable included by model z (i.e., $\bX_z = \bX_{\btau}^{\hat{A}_z}$) for $z=s,j$ and $g_{\hat{A}} = g = \frac{1}{w(T)}$ and write the marginal likelihood as $f(\by|a_z,\lambda_z)$ instead of $f(\by|a_z,\lambda_z,\btau)$. We need to show that $\frac{f(\by|a_j,\lambda_j)}{f(\by|a_s,\lambda_s)} \rightarrow_p 0 $ for any $j\neq s$. In particular, we have
\begin{equation} \label{eq:BF}
\begin{split}
\frac{f(\by|a_j,\lambda_j)}{f(\by|a_s,\lambda_s)} & = \underbrace{\frac{(\frac{g}{1+g})^{k_{\hat{A}_j}/2}}{(\frac{g}{1+g})^{k_{\hat{A}_s}/2}}}_{C_{js}} \underbrace{[\frac{\frac{g}{1+g}s_{\tilde{\bX}_{\tau_0}}+\frac{1}{(1+g)}s_{\tilde{\bX}_{\tau_0},\bX_j}}{\frac{g}{1+g}s_{\tilde{\bX}_{\tau_0}}+\frac{1}{(1+g)}s_{\tilde{\bX}_{\tau_0},\bX_s}}]^{-\frac{T-K}{2}}}_{D_{js}}.
\end{split}
\end{equation}
Focusing on the first term, it is easy to show that
\begin{equation*} 
\begin{split}
C_{js} & = \frac{(1+g)^{k_{\hat{A}_s}/2}}{(1+g)^{k_{\hat{A}_j}/2}} g^{\frac{k_{\hat{A}_j} - k_{\hat{A}_s}}{2}} \\
 & = \frac{(1+w(T)^{-1})^{k_{\hat{A}_s}/2}}{(1+w(T)^{-1})^{k_{\hat{A}_j}/2}} w(T)^{\frac{k_{\hat{A}_s}- k_{\hat{A}_j}}{2}} \\
 & = \mathcal{O}(w(T)^{\frac{k_{\hat{A}_s}- k_{\hat{A}_j}}{2}}).
\end{split}
\end{equation*}
When $T \ra \infty$, we have 
\begin{equation*} 
\begin{split}
C_{js} & = 0 \text{ when } k_{\hat{A}_s} < k_{\hat{A}_j}, \\
 & = 1 \text{ if } k_{\hat{A}_s} = k_{\hat{A}_j}, \\
 &  \ra + \infty \text{ when } k_{\hat{A}_s} > k_{\hat{A}_j}.
\end{split}
\end{equation*}
We now discuss three possible cases. 
\begin{enumerate}
	\item $k_{\hat{A}_s} < k_{\hat{A}_j}$ and the model $j$ does not nest the model $s$. In such case, the term $C_{js} \ra 0$. The second term also tends to zero since we have 
\begin{equation*} 
\begin{split}
D_{js} & = [\frac{\frac{g}{1+g}s_{\tilde{\bX}_{\tau_0}}+\frac{1}{(1+g)}s_{\tilde{\bX}_{\tau_0},\bX_s}}{\frac{g}{1+g}s_{\tilde{\bX}_{\tau_0}}+\frac{1}{(1+g)}s_{\tilde{\bX}_{\tau_0},\bX_j}}]^{\frac{T-K}{2}}\\
& = [\frac{g s_{\tilde{\bX}_{\tau_0}}+ s_{\tilde{\bX}_{\tau_0},\bX_s}}{g s_{\tilde{\bX}_{\tau_0}}+ s_{\tilde{\bX}_{\tau_0},\bX_j}}]^{\frac{T-K}{2}}
\end{split}
\end{equation*}	
Using he fact that $M_j$ does not nest $M_s$ and the Frisch-Waugh theorem (see also Lemma A.1 in \cite{fernandez2001benchmark}), we have that $\lim_{T \ra \infty} \frac{s_{\tilde{\bX}_{\tau_0},\bX_j}}{T} = \sigma^2 + b_j$ with $b_j>0$. Combining with the fact that $g\ra 0$, we end up with a limit of $D_{js}$ given by 
\begin{equation*} 
\begin{split}
\lim_{T \ra \infty} D_{js} & = [\frac{\sigma^2}{\sigma^2  + b_j}]^{\frac{T-K}{2}} \ra 0.
\end{split}
\end{equation*}
\item The model $j$ does not nest the true model but $K_j < K_s$. In such case, the term $C_{js} \ra + \infty$. However, we can show that $\lim_{T \ra \infty} C_{js} w(T)^{-\frac{(K_s-K_j)}{2} + \frac{K_s-K_j}{T-K}} \ra 1$. Indeed, we have that 
\begin{equation*} 
\begin{split}
\lim_{T \ra \infty} C_{js} w(T)^{-\frac{(K_s-K_j)}{2} + \frac{K_s-K_j}{T-K}} & = \frac{(1+w(T)^{-1})^{k_{\hat{A}_s}/2}}{(1+w(T)^{-1})^{k_{\hat{A}_j}/2}} w(T)^{\frac{K_s-K_j}{T-K}}.
\end{split}
\end{equation*}
Let us define $q_T = w(T)^{\frac{K_s-K_j}{T-K}}$. We can compute the limit as follows $\lim_{T \ra \infty} w(T)^{\frac{K_s-K_j}{T-K}} = \lim_{T \ra \infty} \exp \ln q_T$. The limit of $\ln q_T$ is given by 
\begin{equation*} 
\begin{split}
\lim_{T \ra \infty} q_T & =\lim_{T \ra \infty} \frac{K_s-K_j}{T-K} \ln w(T), \\
 & = \lim_{T \ra \infty} \frac{w'(T)}{w(T)} ~~( = 0~~\text{ by assumption}).
\end{split}
\end{equation*}
We conclude that $\lim_{T \ra \infty} w(T)^{\frac{K_s-K_j}{T-K}} = 1$. Now, we need to show that $D_{js}w(T)^{\frac{(K_s-K_j)}{2} - \frac{K_s-K_j}{T-K}} \ra 0$. In fact, we have 
		\begin{eqnarray*}
D_{js}w(T)^{\frac{\frac{(K_s-K_j)}{2}(T-K-2)}{T-K}}  & = & \lim_{T \ra \infty} (\underbrace{\frac{\sigma^2}{\sigma^2 + b_j}}_{a<1})^{\frac{T-K}{2}}  w(T)^{\frac{(K_s-K_j)}{2}}\\
 &  =& \lim_{T \ra \infty}  \frac{w(T)^{\frac{(K_s-K_j)}{2}}}{a^{-\frac{T-K}{2}}},
\end{eqnarray*}
By applying $\left\lceil \frac{(K_s-K_j)}{2}\right\rceil$ times the Hospital's rule, we find that $a^{\frac{T-K}{2}}$ dominates and so $D_{js}w(T)^{\frac{\frac{(K_s-K_j)}{2}(T-K-2)}{T-K}} \ra 0$. 
\item We now consider the last case in which the model $j$ nests the true model $s$. Consequently, we have $K_s < K_j$ and the term $C_{js} \ra 0$. Regarding the other term, we can express it as 
\begin{equation*} 
\begin{split}
D_{js} & = [\frac{g s_{\tilde{\bX}_{\tau_0}}+ s_{\tilde{\bX}_{\tau_0},\bX_s}}{g s_{\tilde{\bX}_{\tau_0}}+ s_{\tilde{\bX}_{\tau_0},\bX_j}}]^{\frac{T-K}{2}},\\
 & = \underbrace{[\frac{s_{\tilde{\bX}_{\tau_0},\bX_s}}{s_{\tilde{\bX}_{\tau_0},\bX_j}}]^{\frac{T-K}{2}}}_{Q_1} \underbrace{[\frac{A_s+ w(T)}{A_j+ w(T)}]^{\frac{T-K}{2}}}_{Q2},
\end{split}
\end{equation*}
where $A_i = \frac{s_{\tilde{\bX}_{\tau_0}}}{s_{\tilde{\bX}_{\tau_0},\bX_i}}$ for $i=j,s$. It is clear that the first term $Q_1$ has a limiting distribution related to the likelihood ratio test. In fact, we have that 
\begin{eqnarray*}
\frac{T-K}{2} \ln \frac{s_{\tilde{\bX}_{\tau_0},\bX_s}}{s_{\tilde{\bX}_{\tau_0},\bX_j}} & = & \underbrace{\frac{T-K}{2T}}_{\ra \frac{1}{2}} \underbrace{T\ln \frac{s_{\tilde{\bX}_{\tau_0},\bX_s}}{s_{\tilde{\bX}_{\tau_0},\bX_j}}}_{\ra_d \chi^2(\Delta_{js})},\\
& \ra_d & \text{Gamma}(\frac{\Delta_{js}}{2},1),
\end{eqnarray*}
in which $\Delta_{js} = |K_s-K_j|$. Since $Y \sim \text{Gamma}(\frac{\Delta_{js}}{2},1)$ is $\mathcal{O}_p(1)$, we have that $C_{js} \exp Y \ra_p 0$. \\
Focusing on the second term $Q_2$, using assumption (iii), we have that 
\begin{equation*} 
\begin{split}
\frac{T-K}{2} \ln [\frac{A_s+ w(T)}{A_j+ w(T)}] & = \frac{T-K}{2}\ln [1 + \frac{A_s - A_j}{A_j+ w(T)}], \\
 & = \mathcal{O}_p(\frac{T}{w(T)}).\\
 & \ra_p  [0,\infty).
\end{split}
\end{equation*}
Since $C_{js} \ra 0$, we conclude that $C_{js} Q_1 Q_2 \ra_p 0$.
\end{enumerate}

\subsection{Convergence to the BIC \label{App:BIC}}
The BIC of a linear regression model with $K$ parameters is given by 
\begin{align} 
\begin{split}
BIC(K) &  = -\frac{T}{2} \ln(\frac{s_{\bX}}{T})	-\frac{K}{2} \ln T, \\ 
       & = \underbrace{-\frac{T}{2} \ln(s_{\bX})	-\frac{K}{2} \ln T}_{BIC^*(K)} + \frac{T}{2}\ln T,
\end{split}
\end{align}
in which $s_{\bX}$ denotes the sum of squared residuals given the $T \times K$ dimensional exogenous variables $\bX$ evaluated at the OLS estimates. In this appendix, we show that the logarithm of the marginal likelihood and the BIC$^*(\alpha k_{\hat{A}})$ converges in probability to 0 when $g_{\hat{A}} = \frac{1}{T^{\alpha}} $ with $\alpha>1$. In particular, the marginal likelihood is given by 
\begin{equation}
\begin{split}
f(\by|\ba,\lambda,\btau) & = (\frac{g_{\hat{A}}}{1+g_{\hat{A}}})^{k_{\hat{A}}/2} [\frac{g_{\hat{A}}}{1+g_{\hat{A}}}s_{\tilde{\bX}_{\tau_0}}+\frac{1}{(1+g_{\hat{A}})}s_{\tilde{\bX}_{\tau_0},\tilde{\bX}_{\btau}^{\hat{A}}}]^{-\frac{T-K}{2}}.
\end{split}
\end{equation}
We have the following results:
\begin{align} 
\begin{split}
f(\by|\ba,\lambda,\btau) & = T^{-\frac{\alpha k_{\hat{A}}}{2}} [\frac{1}{T^{\alpha}}s_{\tilde{\bX}_{\tau_0}}+\frac{T^{\alpha}-1}{T^{\alpha}}s_{\tilde{\bX}_{\tau_0},\tilde{\bX}_{\btau}^{\hat{A}}}]^{-\frac{T-K}{2}}, \\
 & = T^{-\frac{\alpha k_{\hat{A}}}{2}} [s_{\tilde{\bX}_{\tau_0},\tilde{\bX}_{\btau}^{\hat{A}}}]^{-\frac{T-K}{2}} [\frac{T^{\alpha}-1}{T^{\alpha}}]^{-\frac{T-K}{2}} [\frac{1}{T^{\alpha}-1} \frac{s_{\tilde{\bX}_{\tau_0}}}{s_{\tilde{\bX}_{\tau_0},\tilde{\bX}_{\btau}^{\hat{A}}}}+ 1]^{-\frac{T-K}{2}}, \\
 & = T^{-\frac{\alpha k_{\hat{A}}}{2}} [s_{\tilde{\bX}_{\tau_0},\tilde{\bX}_{\btau}^{\hat{A}}}]^{-\frac{T-K}{2}}  \underbrace{[1-\frac{1}{T^{\alpha}}]^{-\frac{T-K}{2}}}_{C_1} \underbrace{[\frac{1}{T^{\alpha}-1} \frac{s_{\tilde{\bX}_{\tau_0}}}{s_{\tilde{\bX}_{\tau_0},\tilde{\bX}_{\btau}^{\hat{A}}}}+ 1]^{-\frac{T-K}{2}}}_{C_2}, \\
\ln f(\by|\ba,\lambda,\btau) & = -\frac{T - K}{T}\frac{T}{2}\ln s_{\tilde{\bX}_{\tau_0},\tilde{\bX}_{\btau}^{\hat{A}}} - \frac{\alpha k_{\hat{A}}}{2}\ln T + \ln C_1 + \ln C_2
\end{split}
\end{align}
We now show that the two quantities, i.e. $C_1$ and $C_2$, tends to 1 when $T\ra +\infty$:
\begin{align} 
\begin{split}
C_1 & = \exp \left(-\frac{T-K}{2} \ln [1-\frac{1}{T^{\alpha}}] \right), \\
    & \ra_p 1 \text{ since } \alpha>1, \\
C_2 & = \exp \left(-\frac{T-K}{2} \ln	[\frac{1}{T^{\alpha}-1} \frac{s_{\tilde{\bX}_{\tau_0}}}{s_{\tilde{\bX}_{\tau_0},\tilde{\bX}_{\btau}^{\hat{A}}}}+ 1]	\right), \\
    & \ra_p 1 \text{ since } \alpha>1 \text{ and } \frac{s_{\tilde{\bX}_{\tau_0}}}{s_{\tilde{\bX}_{\tau_0},\tilde{\bX}_{\btau}^{\hat{A}}}}=\mathcal{O}_p(1),
    \end{split}
\end{align}

It follows that 
\begin{align} 
\begin{split} \ln f(\by|\ba,\lambda,\btau) -  \left(-\frac{T - K}{T}\frac{T}{2} \ln s_{\tilde{\bX}_{\tau_0},\tilde{\bX}_{\btau}^{\hat{A}}}	-\frac{\alpha k_{\hat{A}}}{2} \ln T\right) &\overset{p}{\rightarrow} 0\\
\ln f(\by|\ba,\lambda,\btau) -  \left(-\frac{T}{2} \ln s_{\tilde{\bX}_{\tau_0},\tilde{\bX}_{\btau}^{\hat{A}}}	-\frac{\alpha k_{\hat{A}}}{2} \ln T\right) &\overset{p}{\rightarrow} 0 \\
\dfrac{f(\by|\ba,\lambda,\btau)}{\exp\left(-\frac{T}{2} \ln s_{\tilde{\bX}_{\tau_0},\tilde{\bX}_{\btau}^{\hat{A}}}	-\frac{\alpha k_{\hat{A}}}{2} \ln T\right)} &\overset{p}{\rightarrow} 1
\end{split}
\end{align}

\noindent As a result, the models selected using the marginal likelihood are equivalent (asymptotically) to those of the BIC with $\alpha k_{\hat{A}}$ parameters. In addition, the posterior probabilities computed using the marginal likelihood converge to the posterior probabilities that would have been computed using the BIC with $\alpha k_{\hat{A}}$ parameters (since the term $\frac{T}{2}\ln T$ cancels out).
\clearpage 

\section{Time-varying parameter model\label{App:TVP}}
We also consider a standard time-varying parameter process (see \cite{Primiceri02}). The model specification is given by 
\begin{eqnarray}
y_t & = & \bx_t' \bbeta_0 + \bx_t' \text{diag}(\bomega)\bbeta_t + \sigma_t \epsilon_t,\label{eq:Bitto1}\\
\bbeta_t|\bbeta_{t-1} & \sim & \NORM(\bbeta_{t-1},I_K), \\
\ln \sigma_t^2 & = &  \ln \sigma_{t-1}^2 + \eta_t, \text{for } t>0, \\
\ln \sigma_0^2 & \sim & N(0,1), \label{eq:Bitto2}
\end{eqnarray}
where $\bomega=(\omega_1,...,\omega_K)'$, $\eta_t \sim \NORM(0,q)$ with $q\sim \mathcal{IG}(\underline{q}_1=3,\underline{q}_2=20)$, $(\bbeta_0',\bomega')' \sim \NORM(0,I_{2K})$ and $K$ stands for the number of explanatory variables. We also define $\ln \sigma_{1:T} = (\ln \sigma_{1}^2, \ldots, \ln \sigma_{T}^2)'$ and $y_{1:T}$, $\bbeta_{0:T}$ analogously. In order to take into account the autocorrelation structure, we use the same lag orders as exposed in Table \ref{JBF::order}. The other explanatory variables consist in an intercept and the seven factors.\\
The model parameters can be estimated with a standard Markov-chain Monte Carlo \citep[e.g.,][]{bitto2018achieving}. In particular,  the model parameters are estimated using an MCMC algorithm which consists of the following steps:
\begin{itemize}
	\item Sampling $\ln \sigma_{1:T}$ using the offset mixture approach of \cite{kim1998stochastic}. In particular, given $\bbeta_{0:T}$ and $\bomega$, we compute $y_t^* = \ln(\nu_t^2 + c)$, for all $t=1,\ldots,T$ in which $\nu_t = y_t- \bx_t (\bbeta_0 + \text{diag}(\omega_1,...,\omega_K)\bbeta_t)$ and $c=0.0001$. The model boils down to a standard TVP model with time-varying intercept since we have
	\begin{eqnarray}
y_t^* & = & \ln \sigma_t^2  + \ln \epsilon_t^2,\\
\ln \sigma_t^2 & = &  \ln \sigma_{t-1}^2 + \eta_t.
\end{eqnarray}
Approximating the distribution $\ln \epsilon_t^2$ with a 8-component mixture of normal distributions, we sample the time-varying variance from a high-dimensional multivariate normal distribution using the sampler 'all without a loop' (AWOL) as suggested in \cite{mccausland2011simulation} \citep[see also][for a simple exposition of the algorithm]{kastner2014ancillarity}.
\item Sampling $\ln \sigma_0^2|\ln \sigma_1^2,q \sim N((q^{-1}+1)^{-1}\frac{\ln \sigma_1^2}{q},(q^{-1}+1)^{-1})$.
\item Sampling $q|\ln \sigma_{0:T} \sim IG(\underline{q}_1+\frac{T}{2},\underline{q}_2 + \sum_{t=1}^{T}\frac{(\ln \sigma_t^2-\ln \sigma_{t-1}^2)^2}{2}$).
\item Sampling $\bbeta_{1:T}|y_{1:T},\bbeta_0,\bomega,\ln \sigma_{1:T}$ using a Kalman filter. Note that this step could also be carried out with the AWOL approach.
\item Sampling $(\bbeta_0',\bomega')'|y_{1:T},\bbeta_{1:T},\ln \sigma_{1:T}$ from a multivariate normal distribution. In fact, conditioning to $\bbeta_{1:T}$ and $\ln \sigma_{1:T}$, the model boils down to a standard regression model since we have
\begin{eqnarray}
y_t & = & \bx_t' \bbeta_0 + (\bx_t'\text{diag}(\bbeta_t)) \bomega + \sigma_t \epsilon_t,\\
    & = & \bx_t' \bbeta_0 + \mathbf{z}_t'\bomega + \sigma_t \epsilon_t,
\end{eqnarray}
in which $\mathbf{z}_t=(\bx_t'\text{diag}(\bbeta_t))'$. 
\end{itemize}

\subsection{Time-varying parameters computed with the FIA returns \label{App:TVP1}}
\renewcommand{\baselinestretch}{1}
\begin{figure}[h!]
	\begin{center}
		\subfloat[SELO - PMKT]{\includegraphics[width=6.5cm,height=4.5cm]{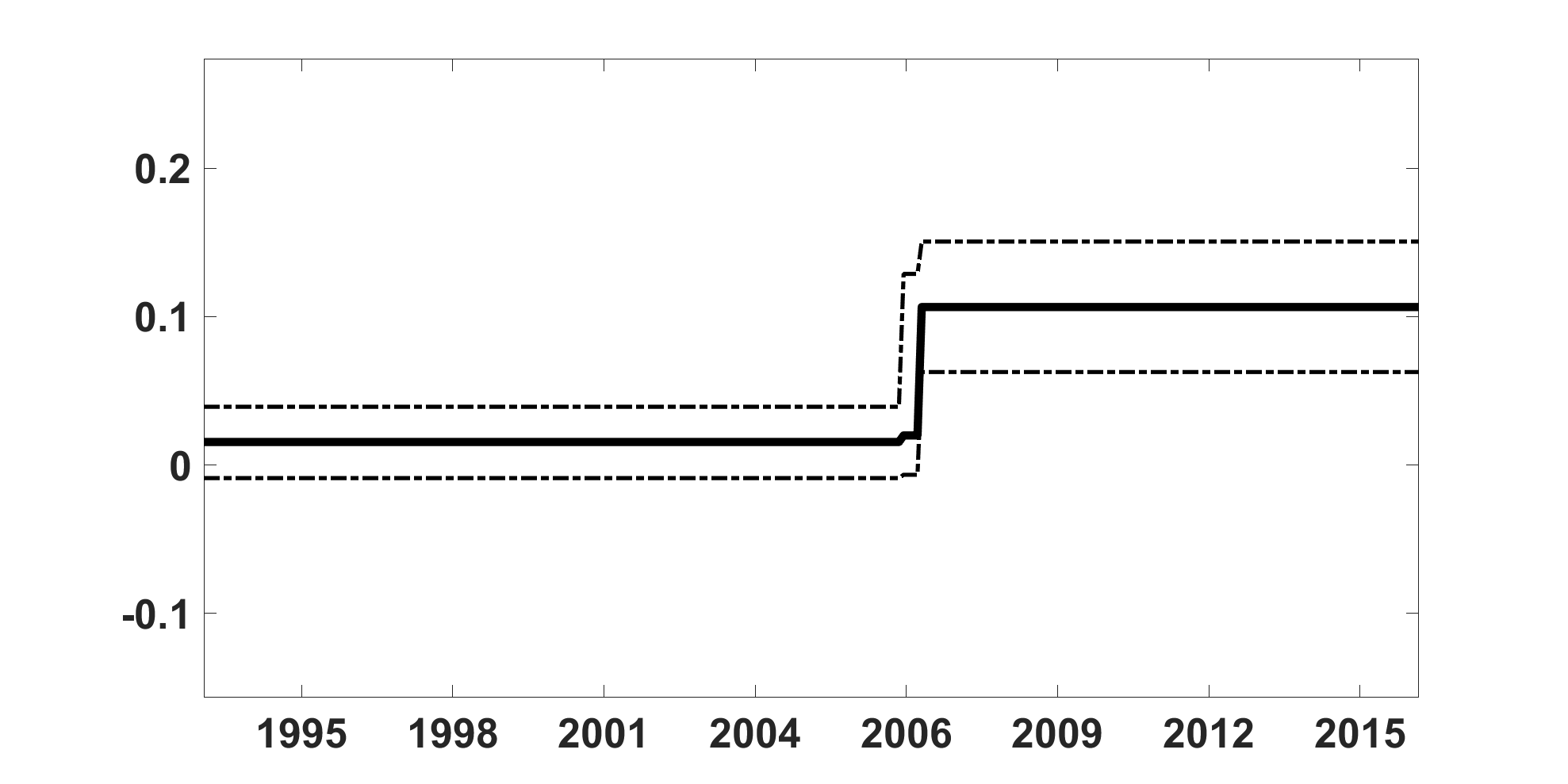}}
		\subfloat[TVP - PMKT]{\includegraphics[width=6.5cm,height=4.5cm]{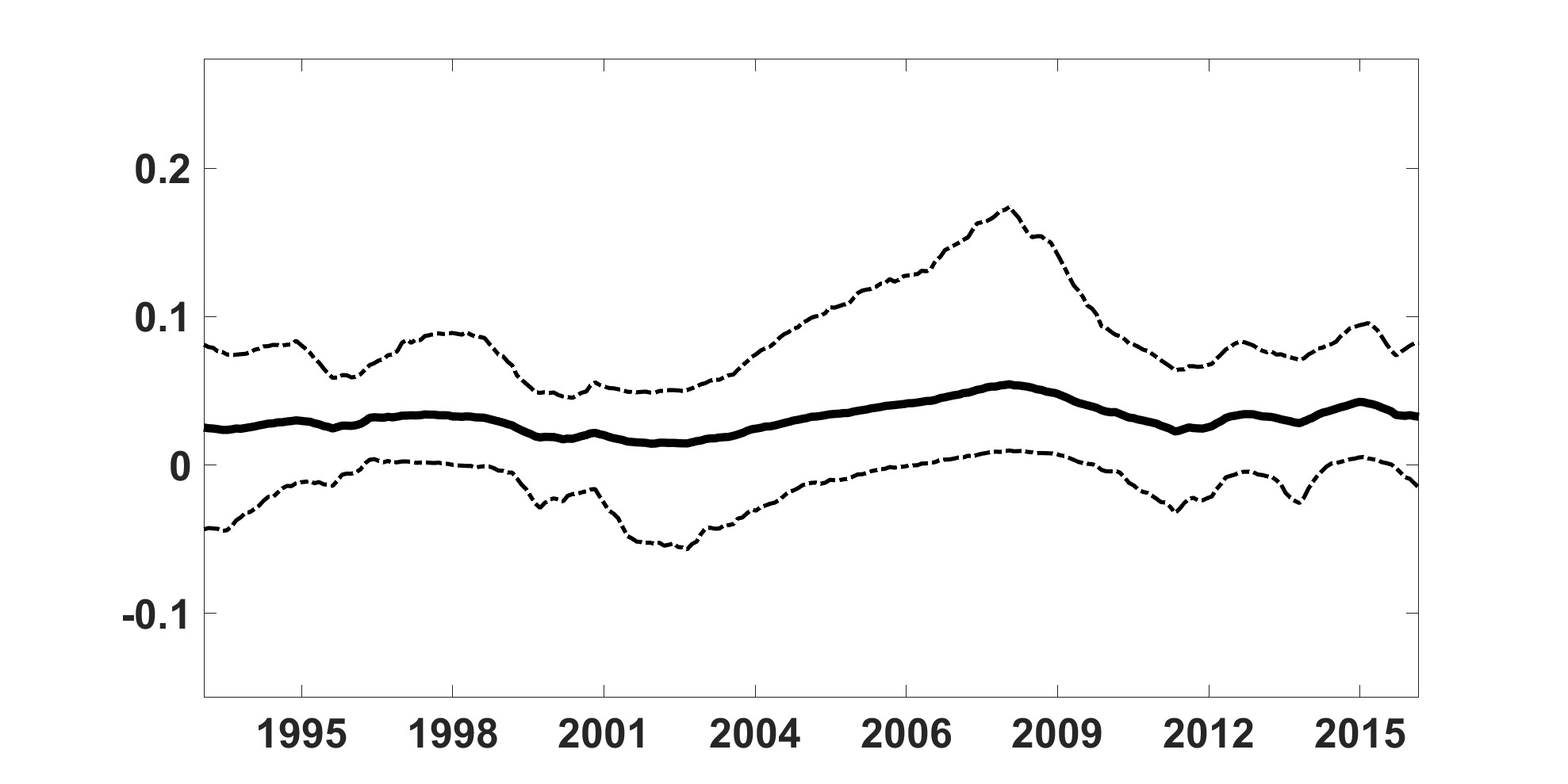}}\\
		\subfloat[SELO -SMB]{\includegraphics[width=6.5cm,height=4.5cm]{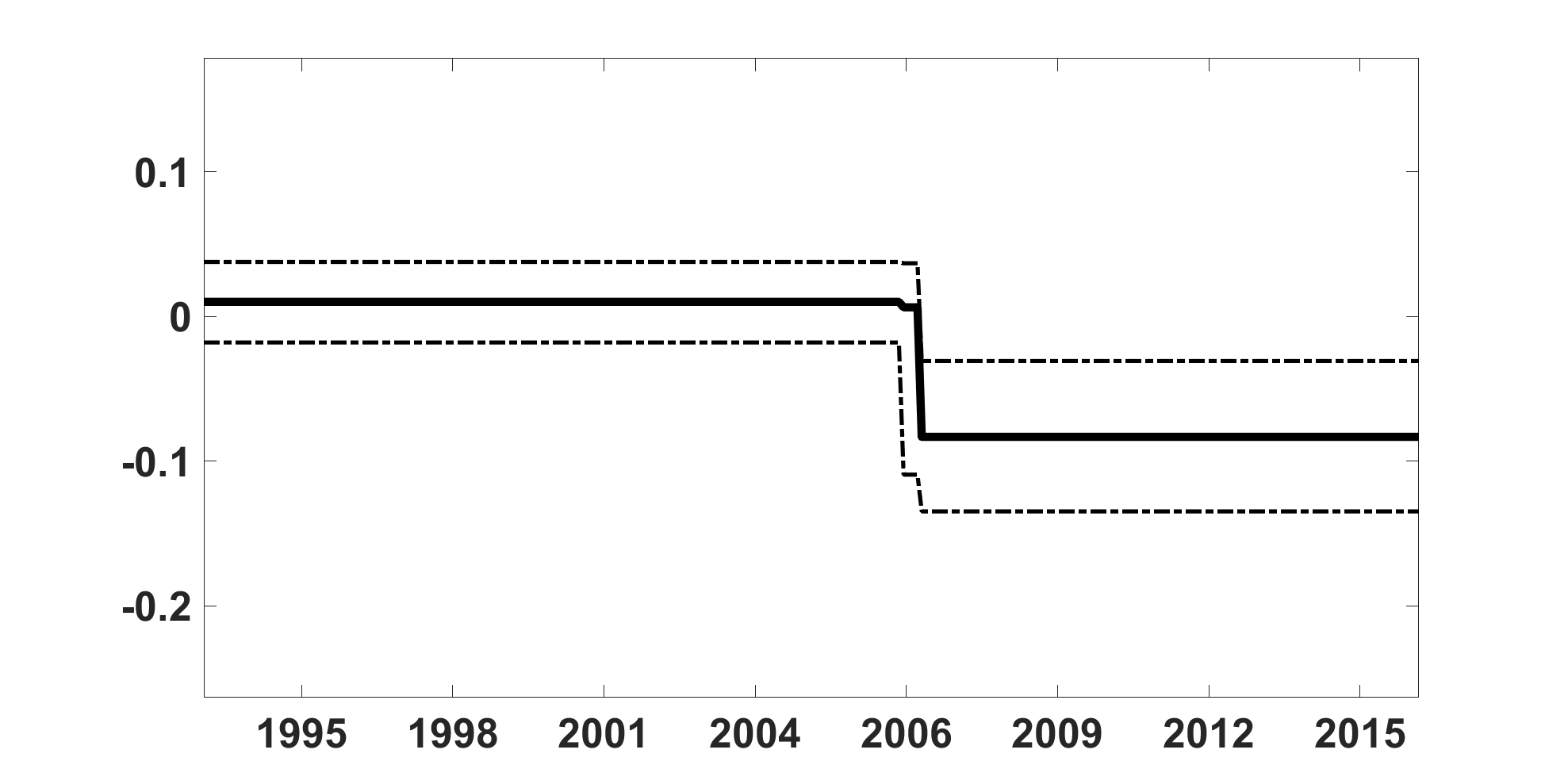}}
		\subfloat[TVP - SMB]{\includegraphics[width=6.5cm,height=4.5cm]{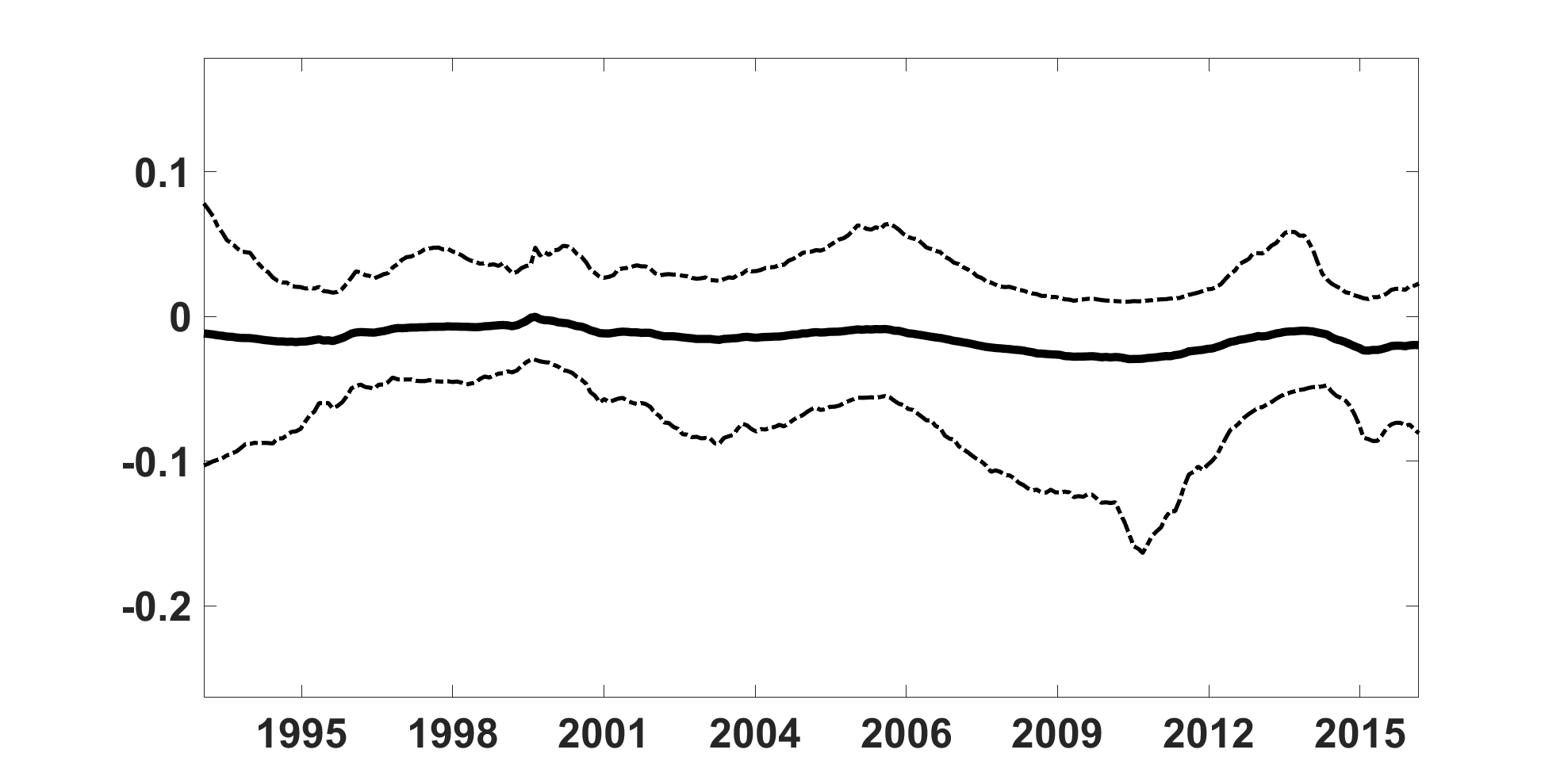}}\\
		\subfloat[SELO - TERM]{\includegraphics[width=6.5cm,height=4.5cm]{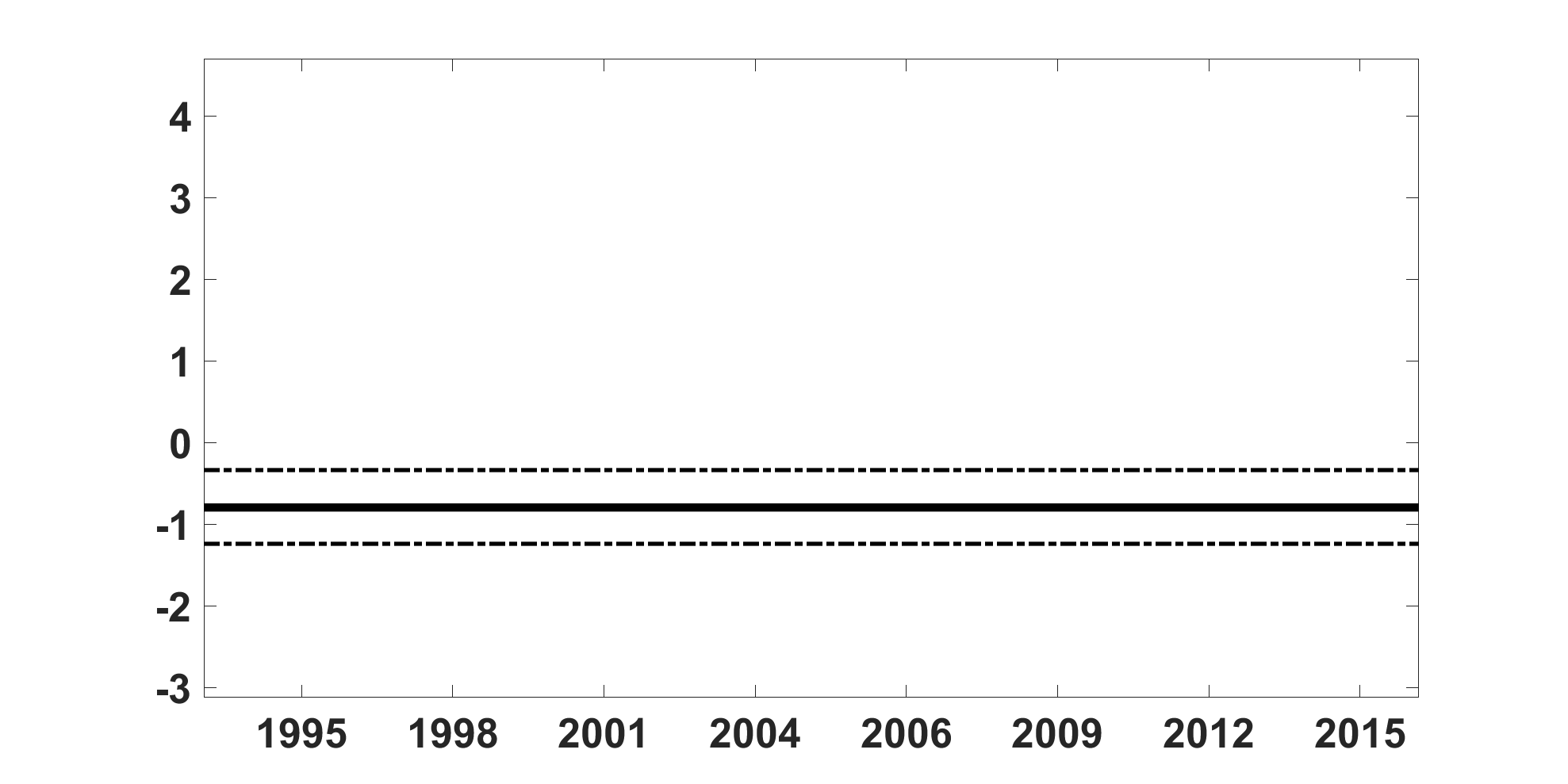}}
		\subfloat[TVP - TERM]{\includegraphics[width=6.5cm,height=4.5cm]{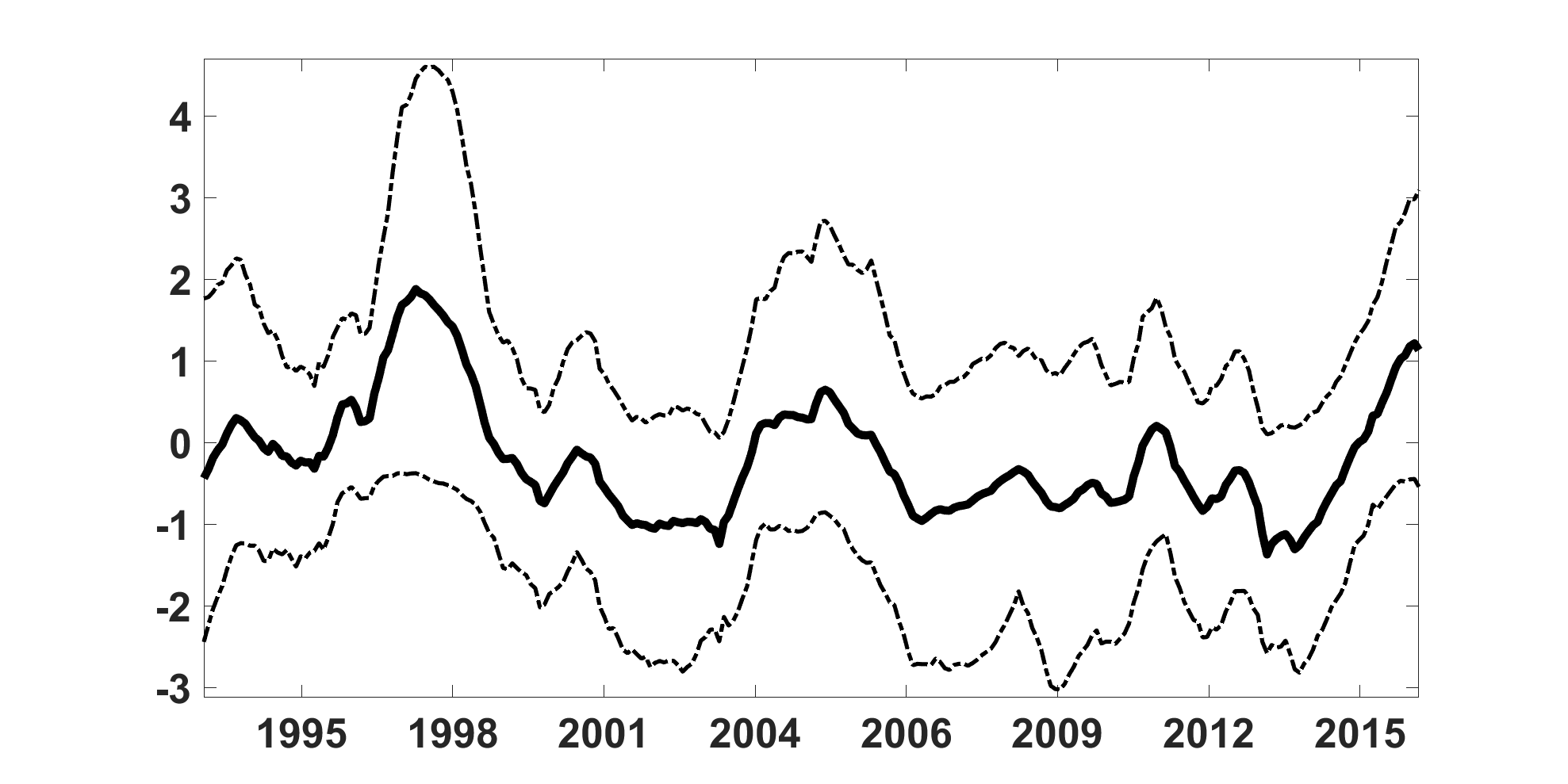}}\\
			\subfloat[SELO - DEF]{\includegraphics[width=6.5cm,height=4.5cm]{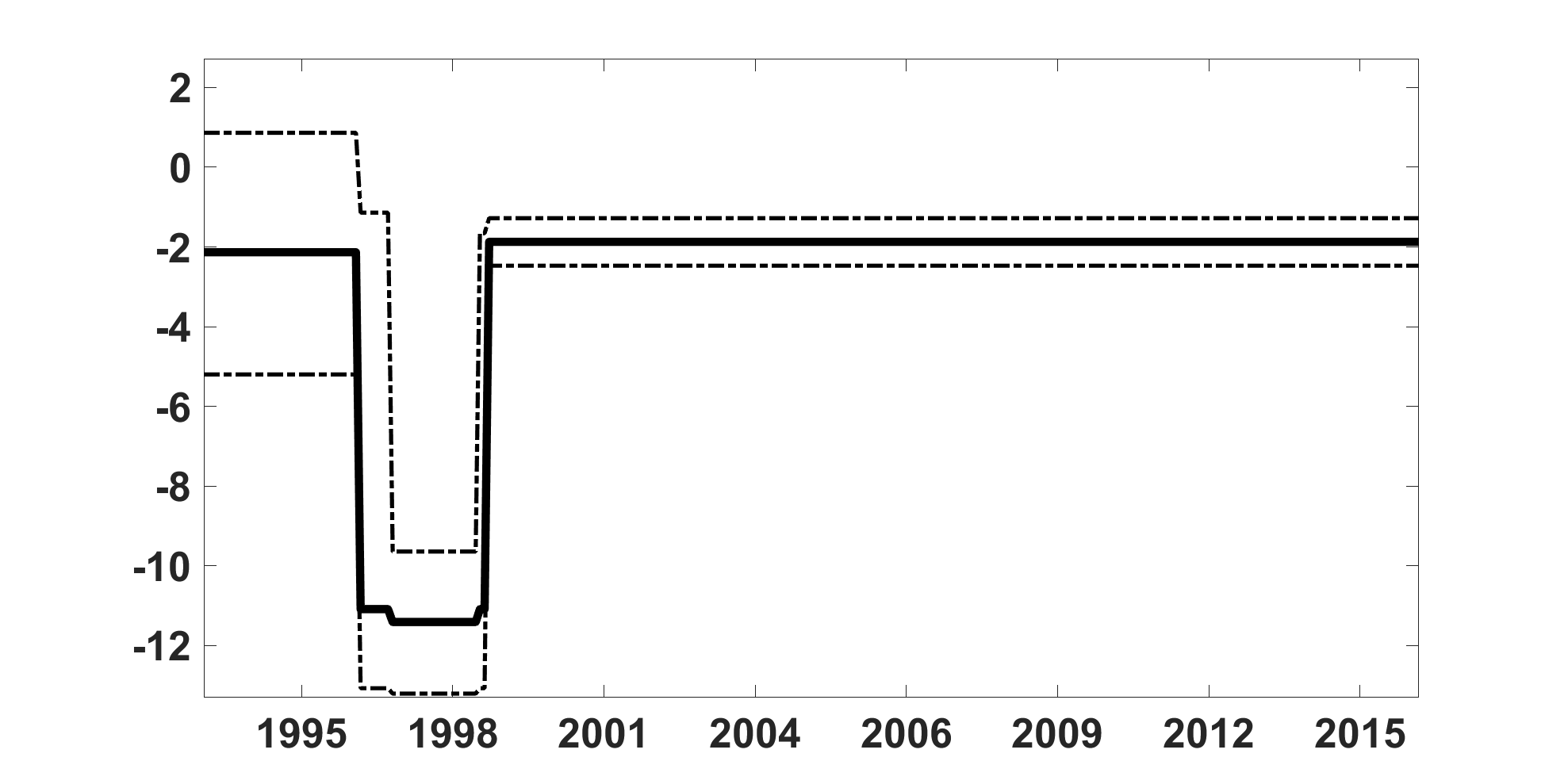}}
	\subfloat[TVP - DEF]{\includegraphics[width=6.5cm,height=4.5cm]{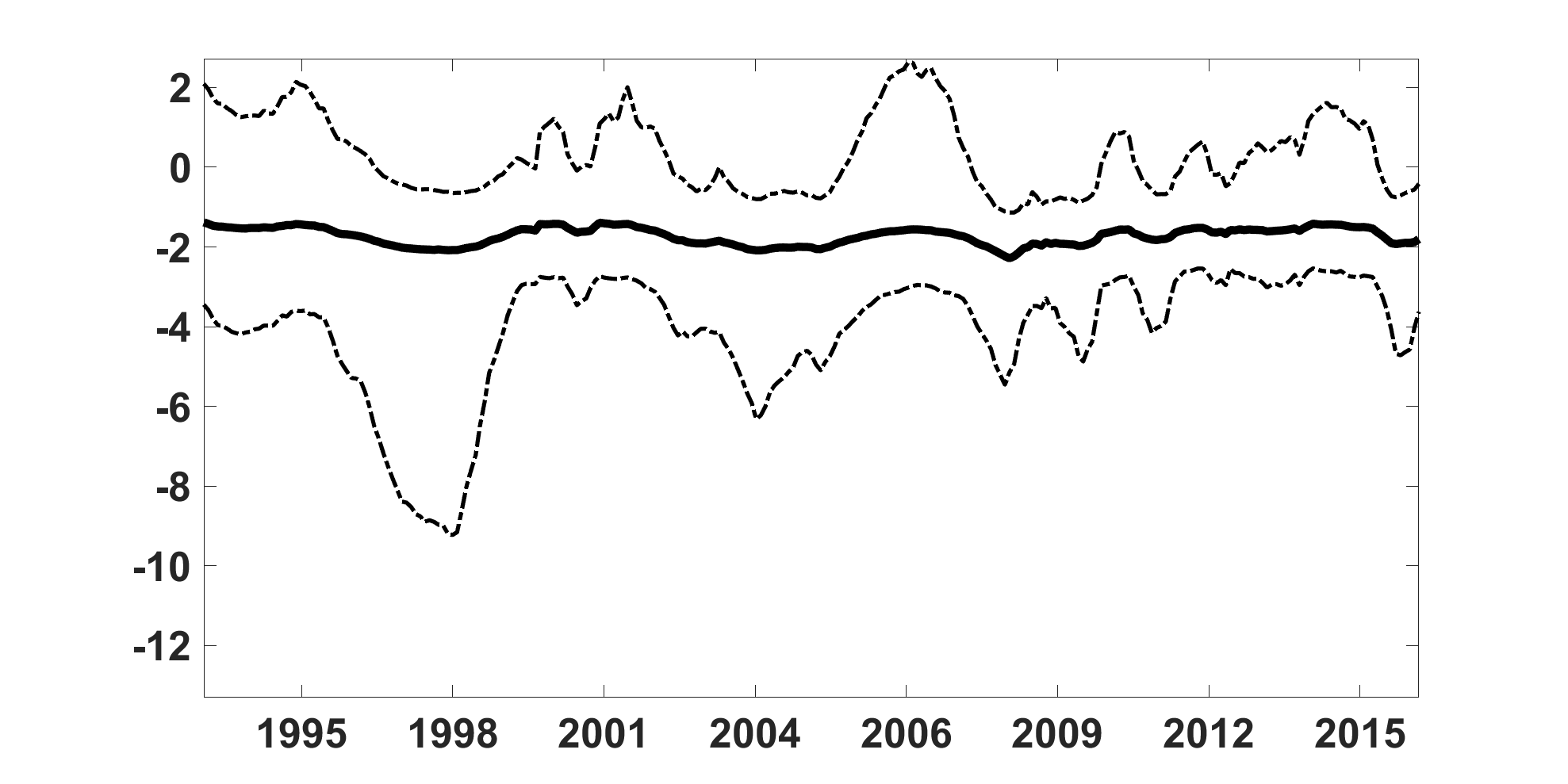}}\\
		\caption{\justifying \textbf{FIA returns - Selective segmentation (SELO) model and Time-varying parameter (TVP) model.} Posterior medians (black) and the 90\% credible intervals (dotted black lines) of the model parameters over time. For the SELO method, we take the break uncertainty into account using the MCMC algorithm presented in Section \ref{sec:breakuncertainty}. \label{fig:strat10}}
	\end{center}
\end{figure}
\renewcommand{\baselinestretch}{1.5}

\renewcommand{\baselinestretch}{1}
\begin{figure}[h!]
	\begin{center}
		\subfloat[SELO - PTFSBD]{\includegraphics[width=6.5cm,height=4.5cm]{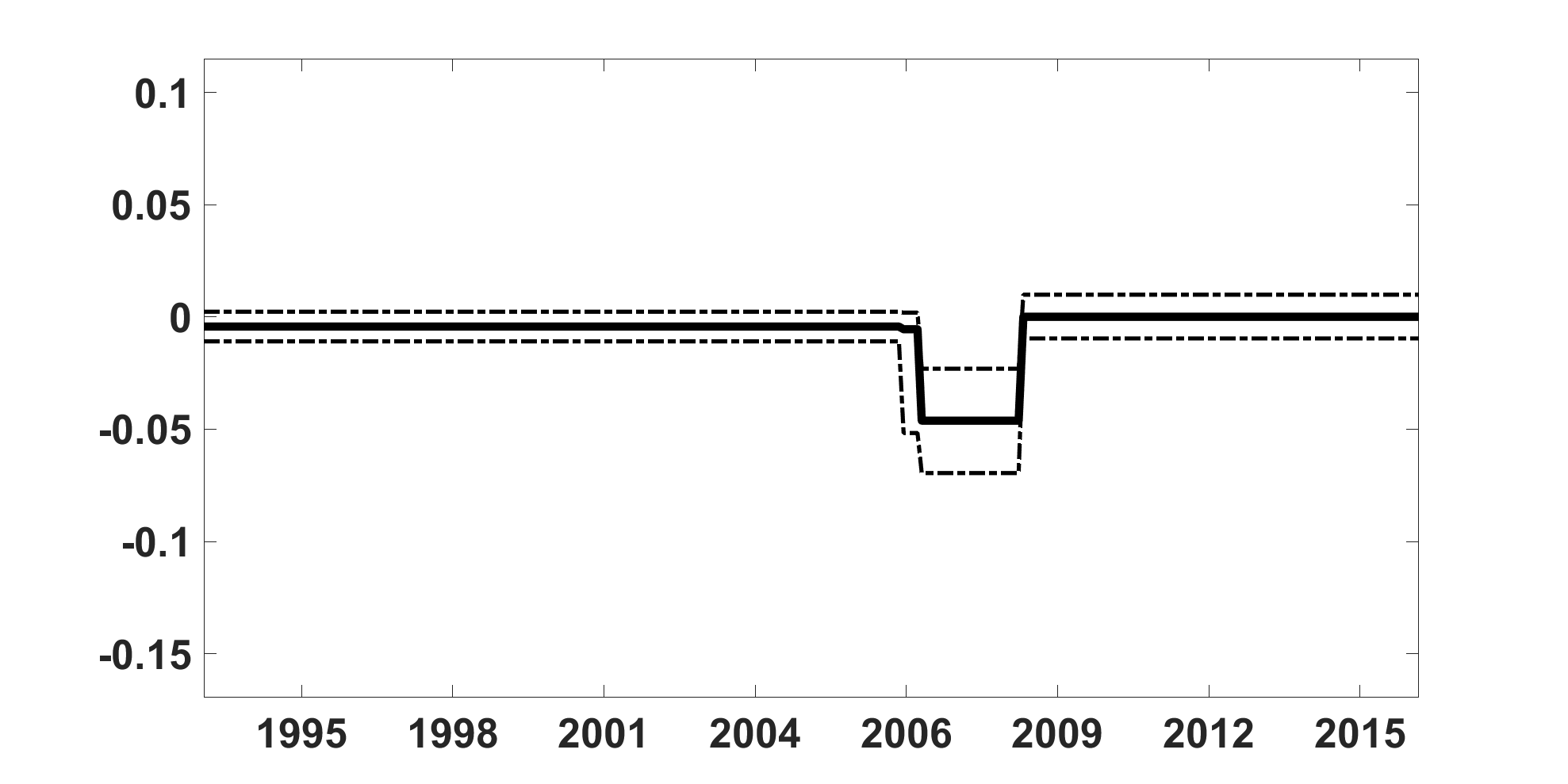}}
		\subfloat[TVP - PTFSBD]{\includegraphics[width=6.5cm,height=4.5cm]{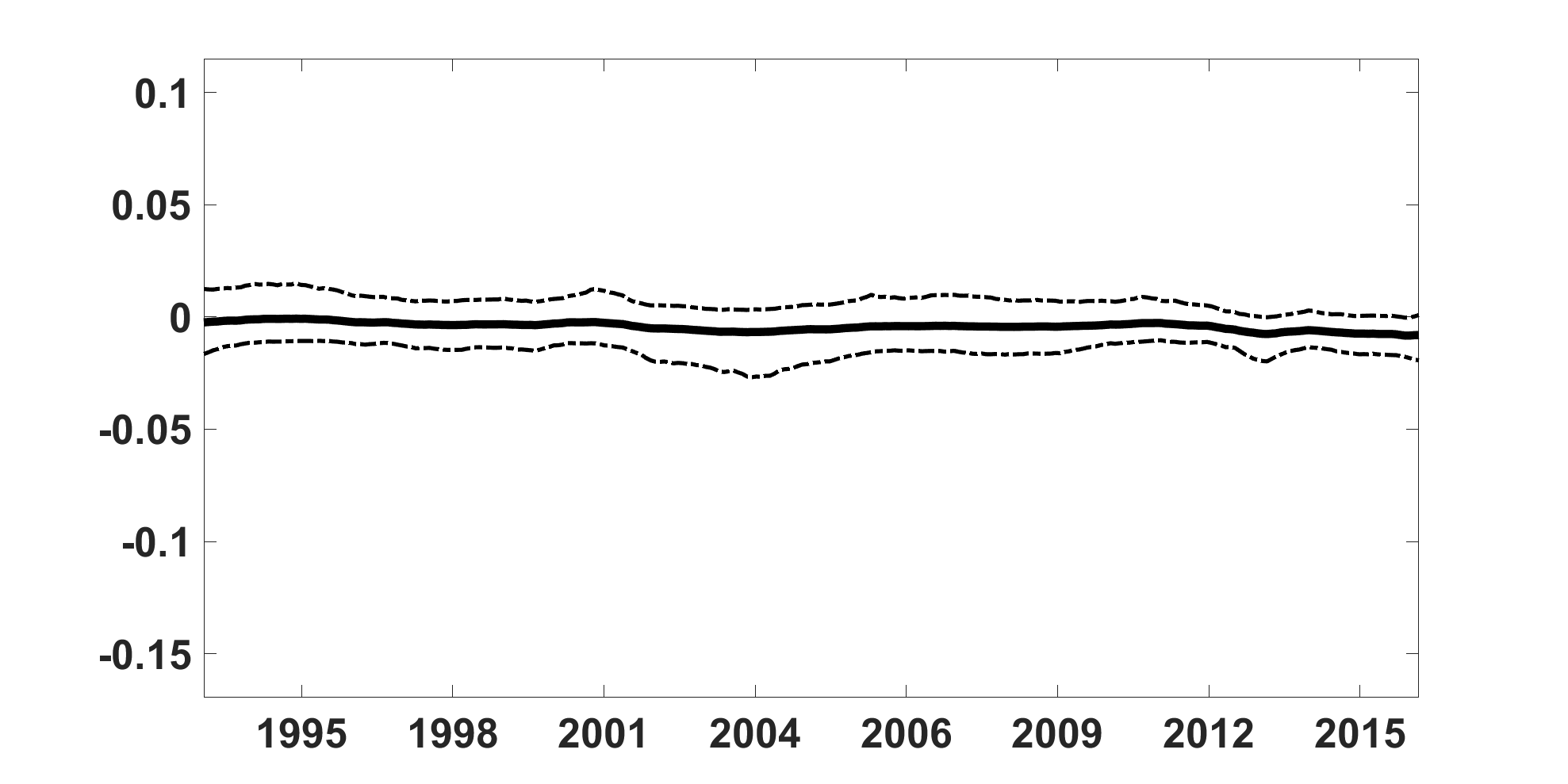}}\\
		\subfloat[SELO - PTFSFX]{\includegraphics[width=6.5cm,height=4.5cm]{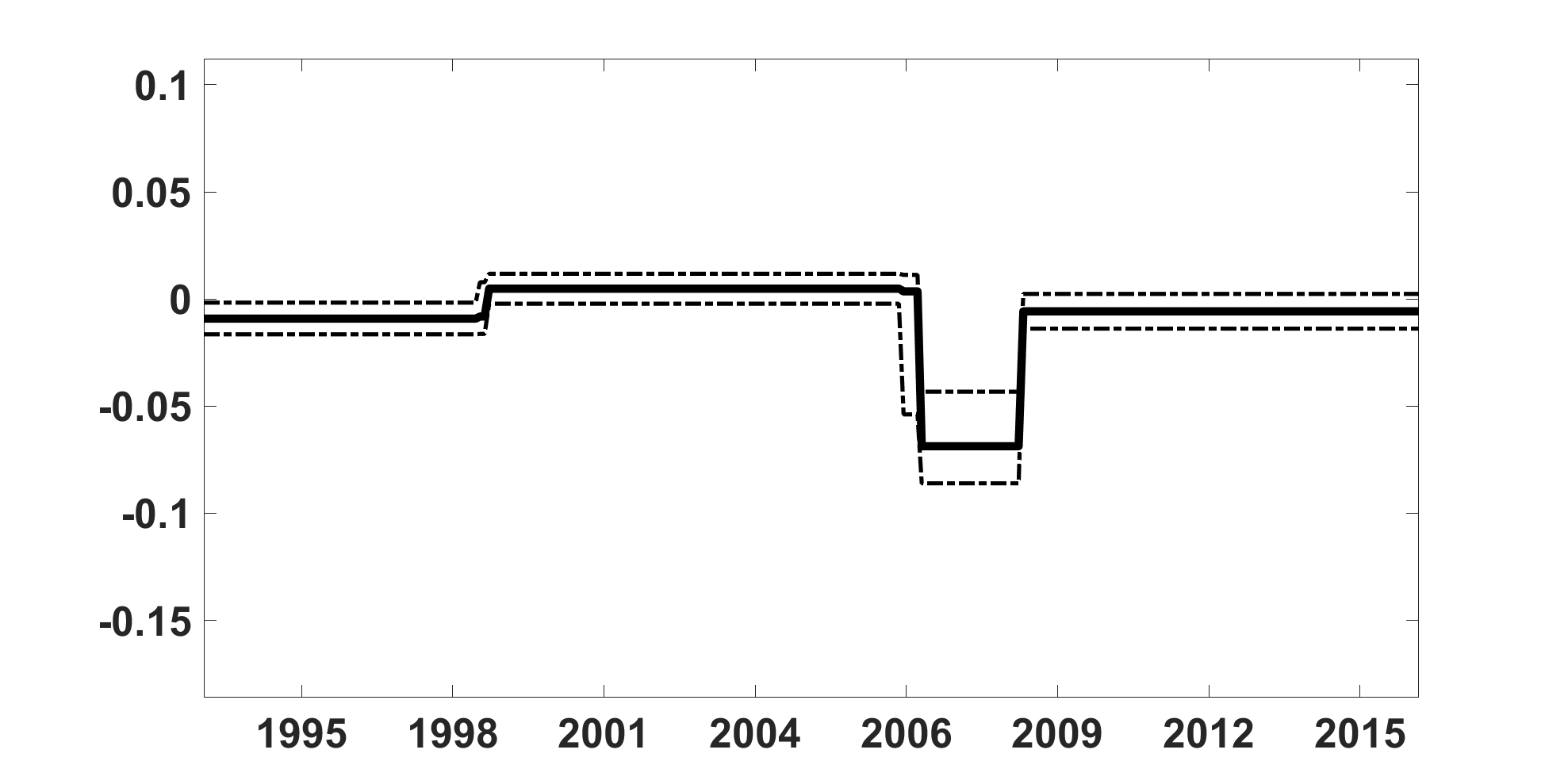}}
		\subfloat[TVP - PTFSFX]{\includegraphics[width=6.5cm,height=4.5cm]{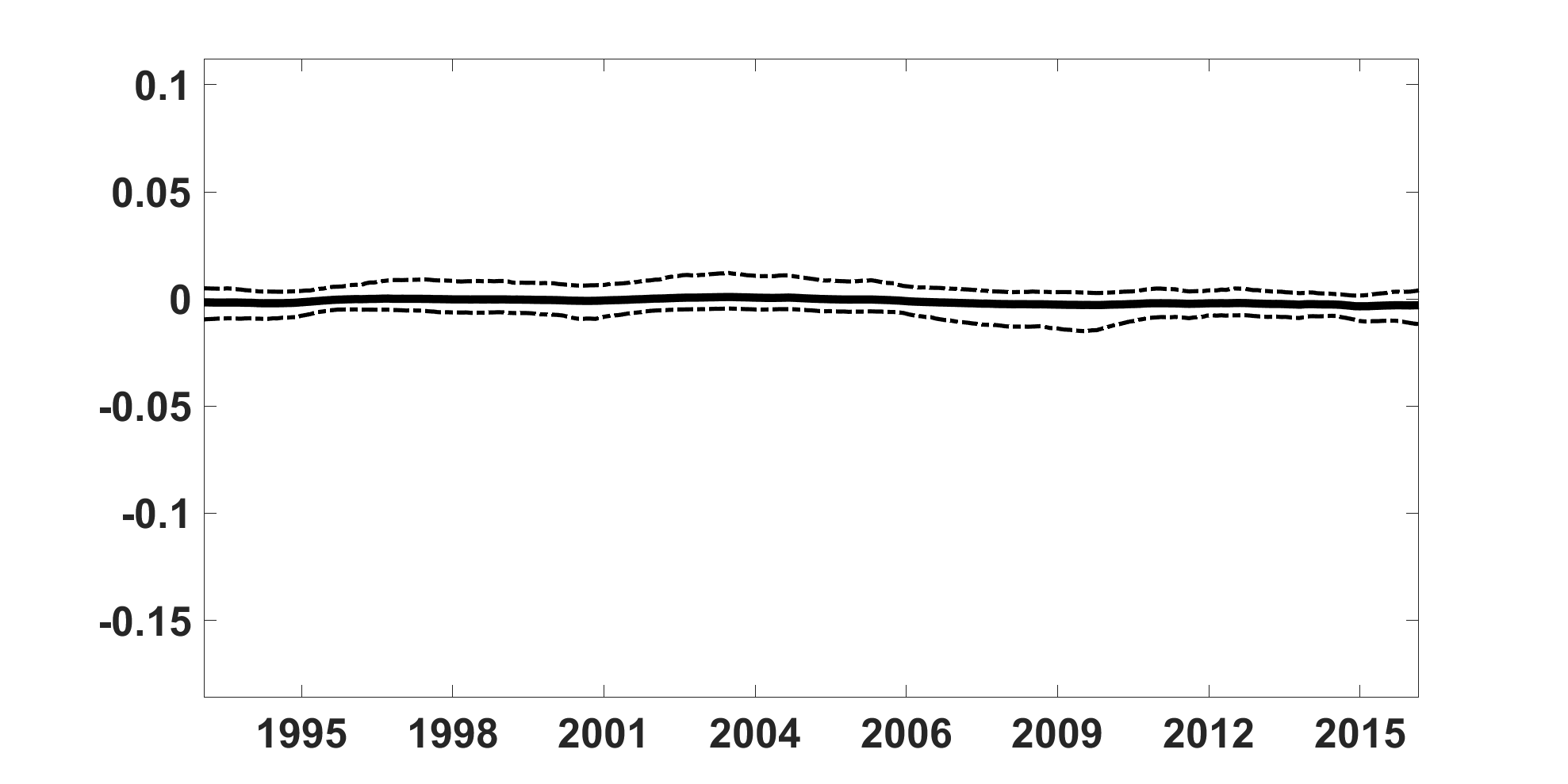}}\\
		\subfloat[SELO - PTFSCOM]{\includegraphics[width=6.5cm,height=4.5cm]{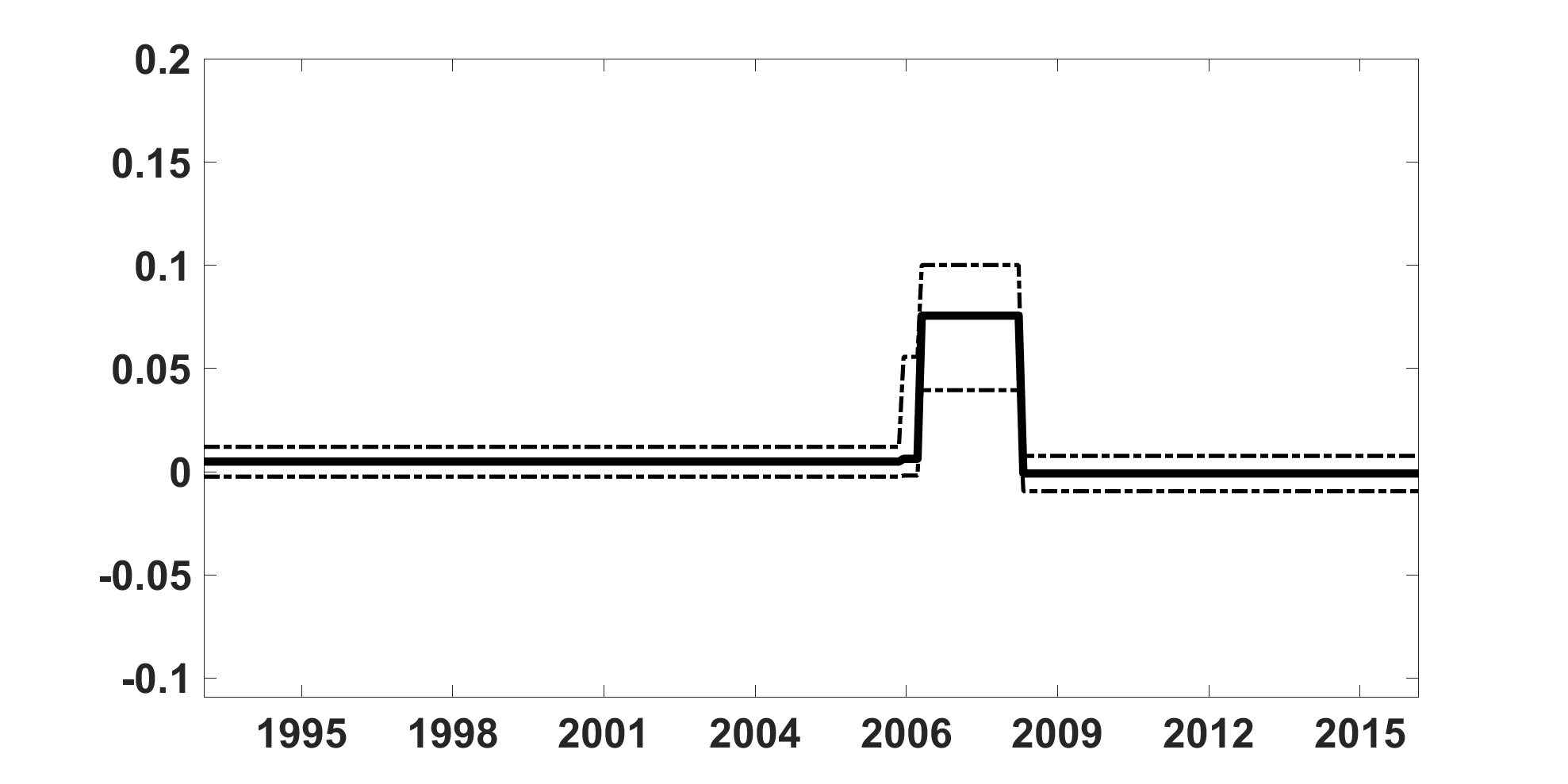}}
		\subfloat[TVP - PTFSCOM]{\includegraphics[width=6.5cm,height=4.5cm]{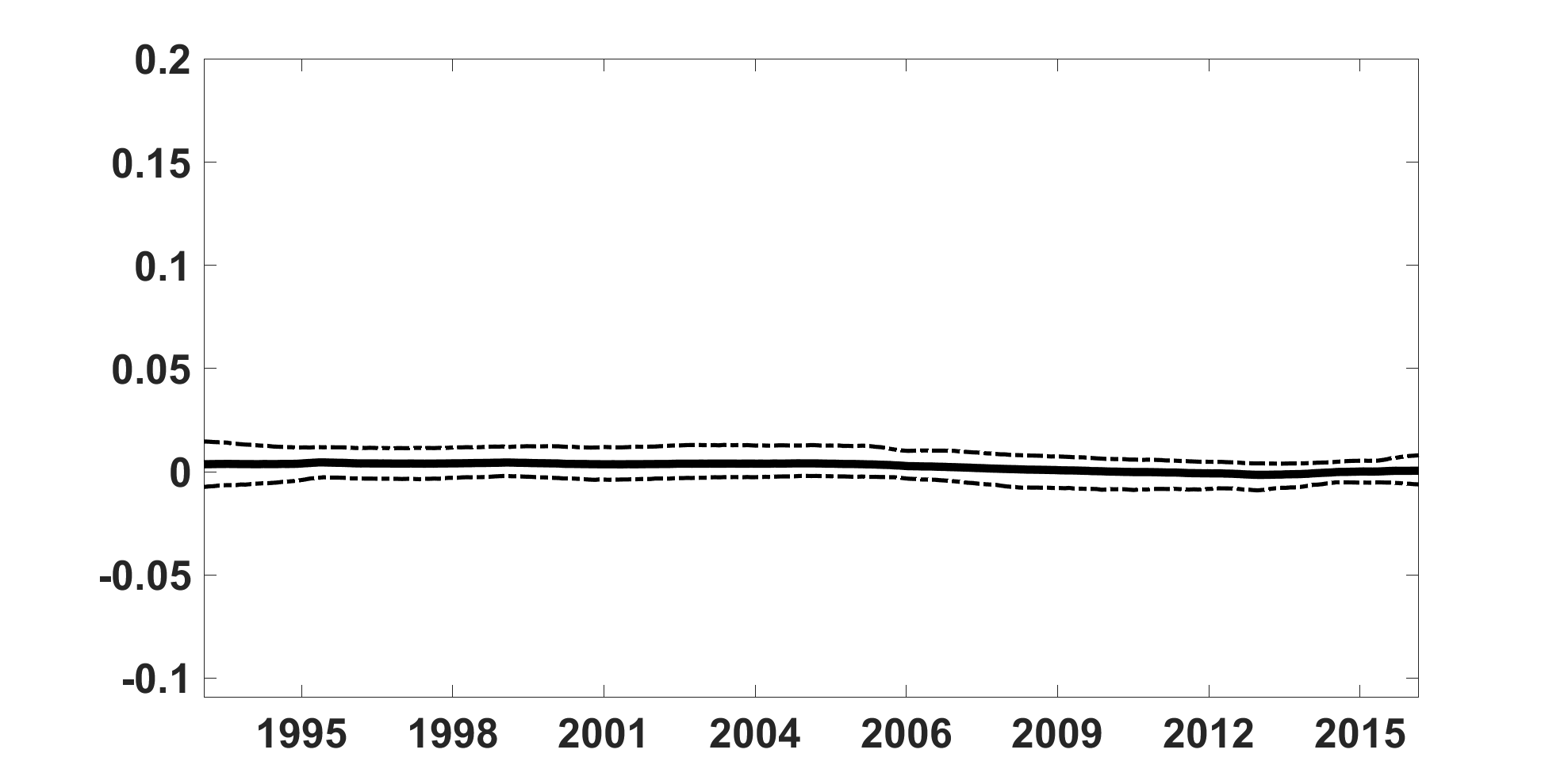}}\\
		\subfloat[SELO - UMD]{\includegraphics[width=6.5cm,height=4.5cm]{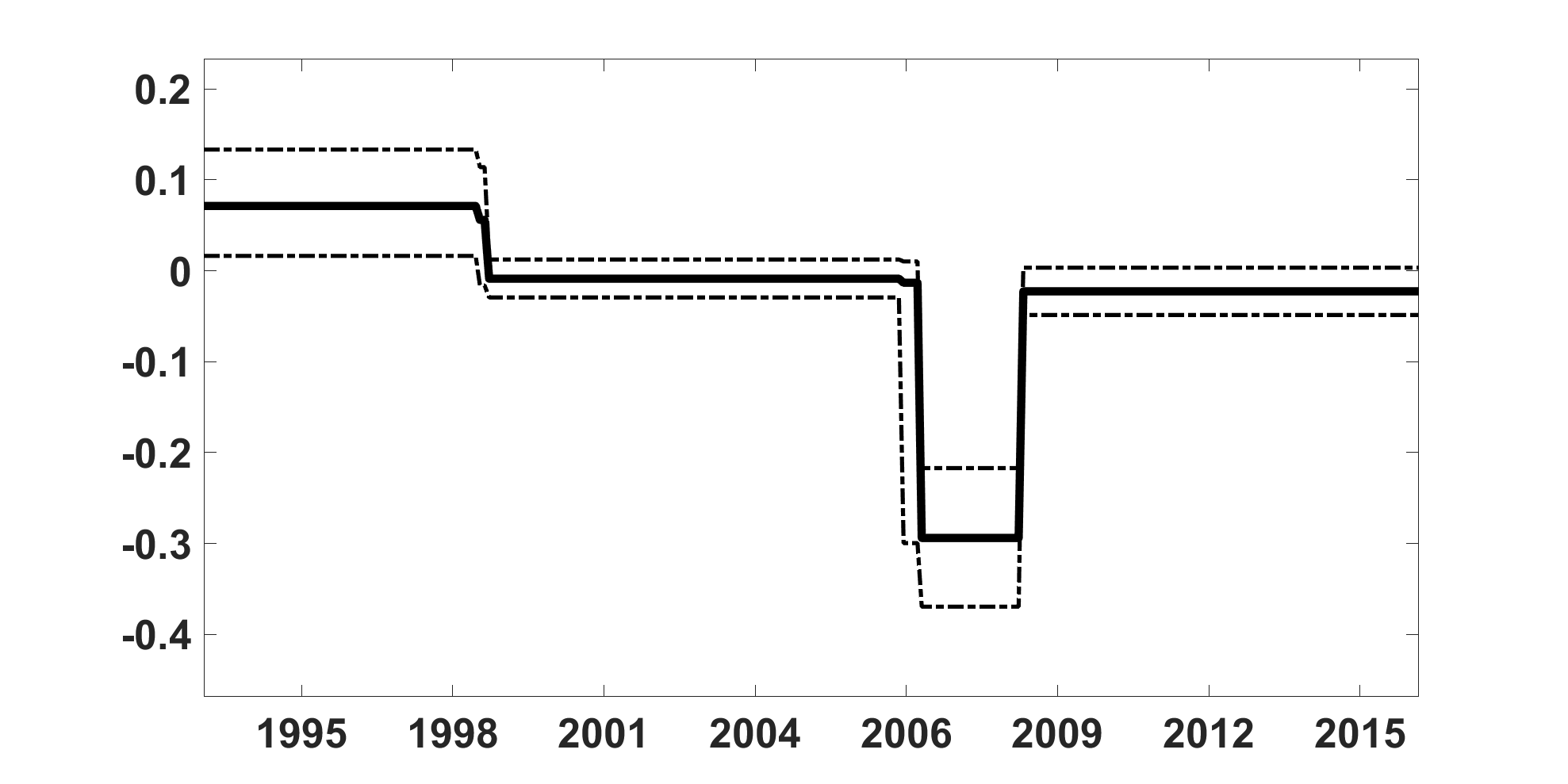}}
	\subfloat[TVP - UMD]{\includegraphics[width=6.5cm,height=4.5cm]{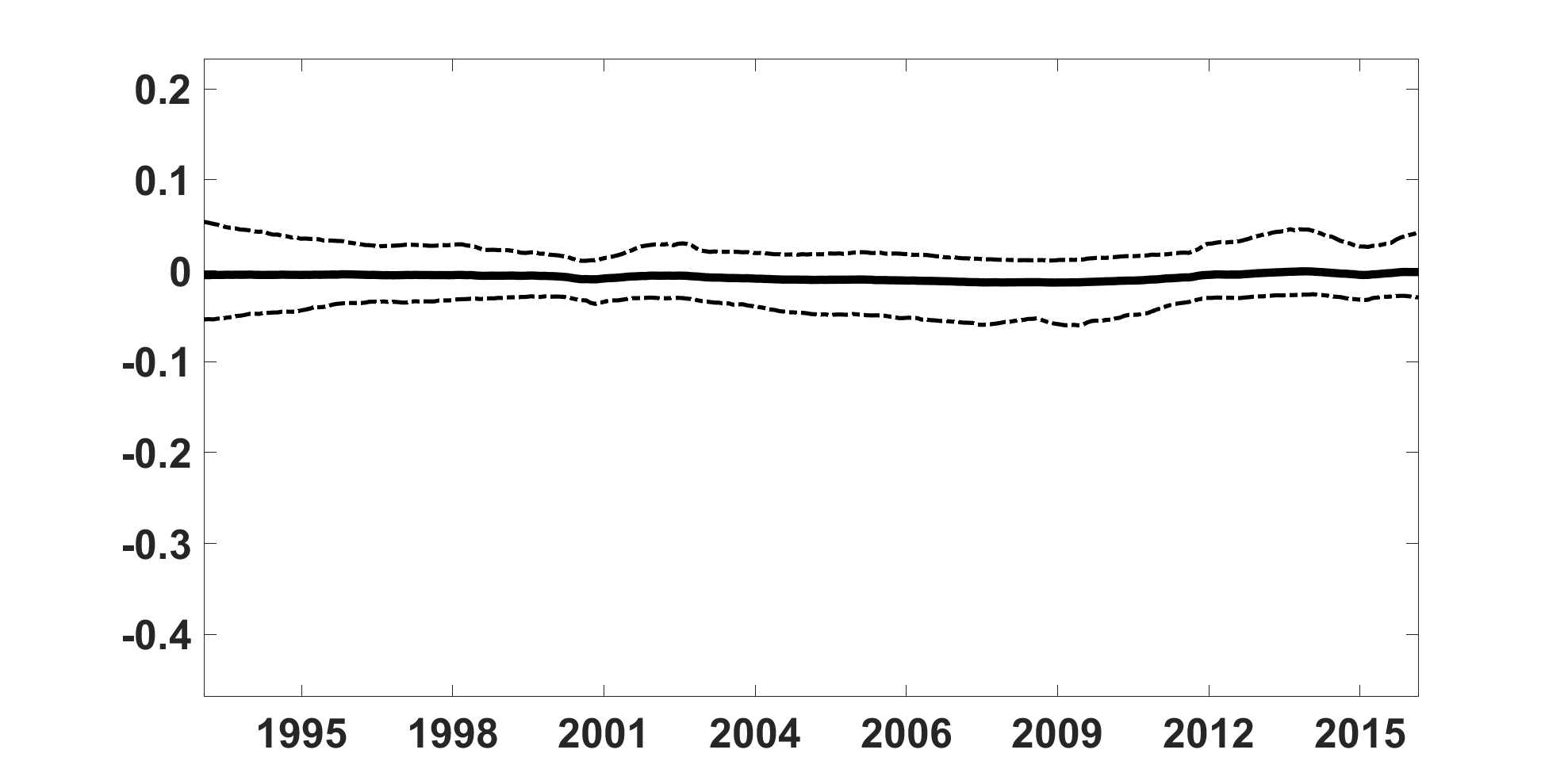}}\\
		\caption{\justifying \textbf{FIA returns - Selective segmentation (SELO) model and Time-varying parameter (TVP) model.} Posterior medians (black) and the 90\% credible intervals (dotted black lines) of the model parameters over time. For the SELO method, we take the break uncertainty into account using the MCMC algorithm presented in Section \ref{sec:breakuncertainty}. \label{fig:strat10_2}}
	\end{center}
\end{figure}
\renewcommand{\baselinestretch}{1.5}

\renewcommand{\baselinestretch}{1}
\begin{figure}[h!]
	\begin{center}
		\subfloat[SELO - PTFSIR]{\includegraphics[width=6.5cm,height=4.5cm]{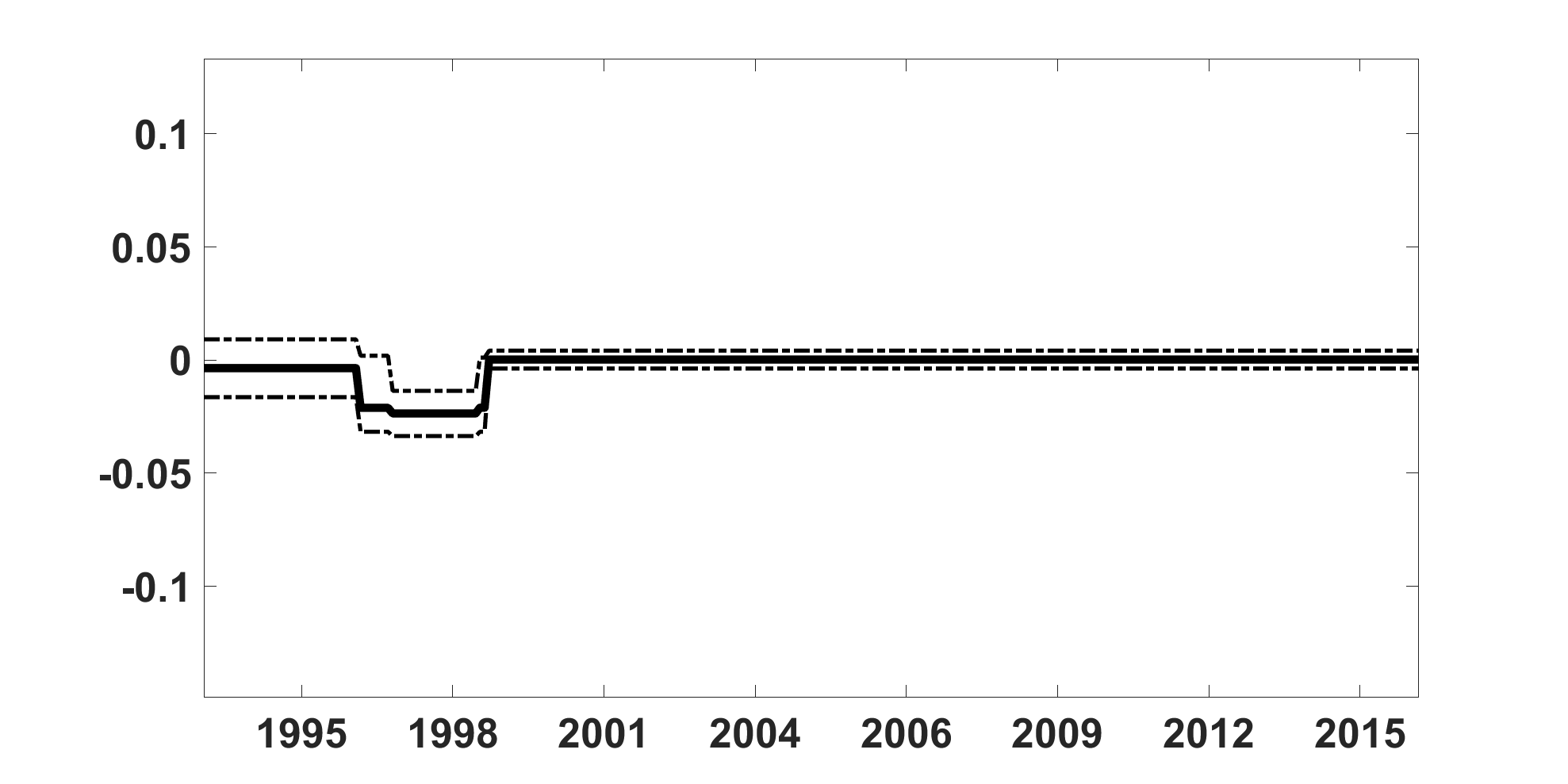}}
		\subfloat[TVP - PTFSIR]{\includegraphics[width=6.5cm,height=4.5cm]{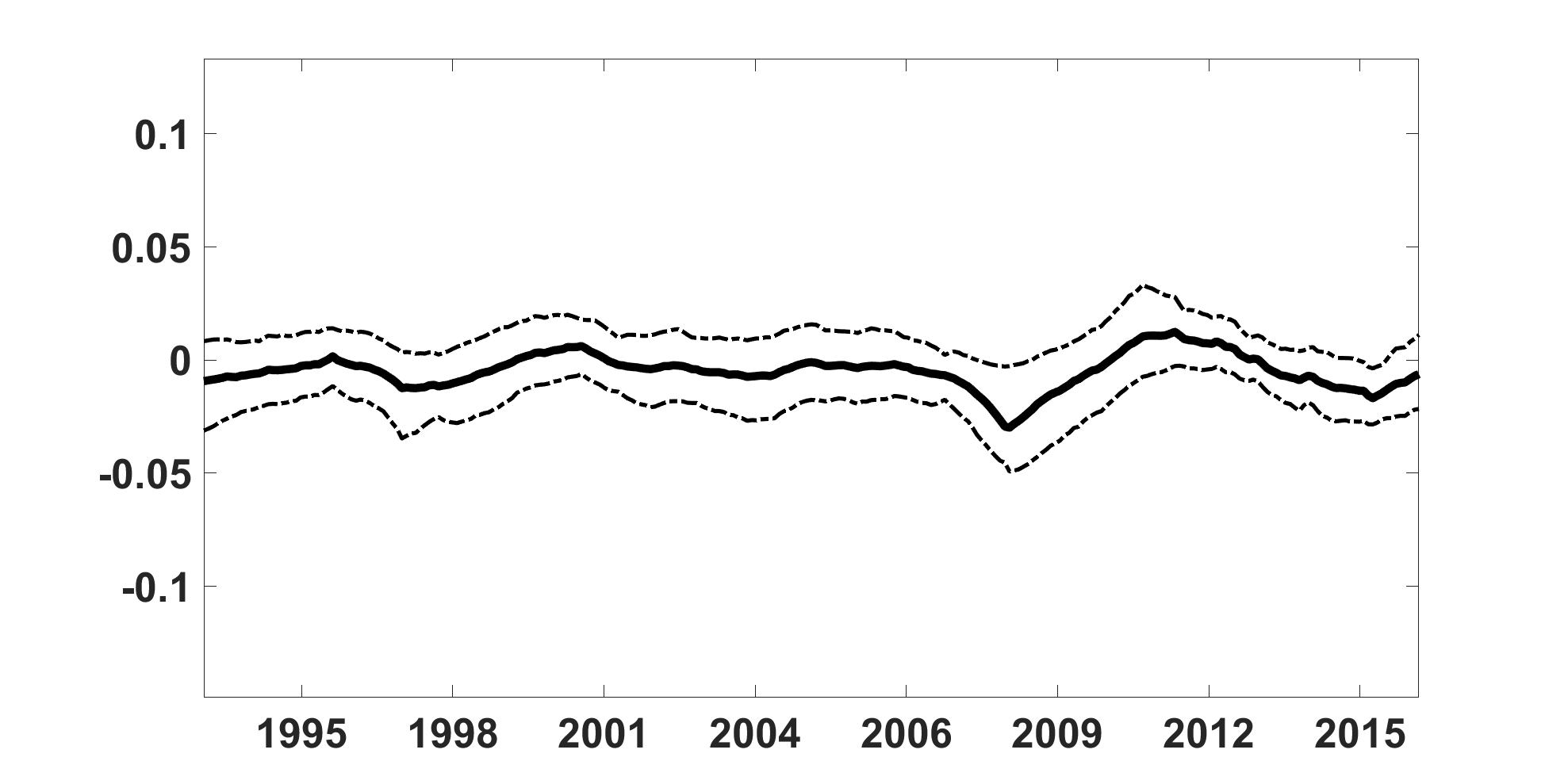}}\\
		\subfloat[SELO - PTFSSTK]{\includegraphics[width=6.5cm,height=4.5cm]{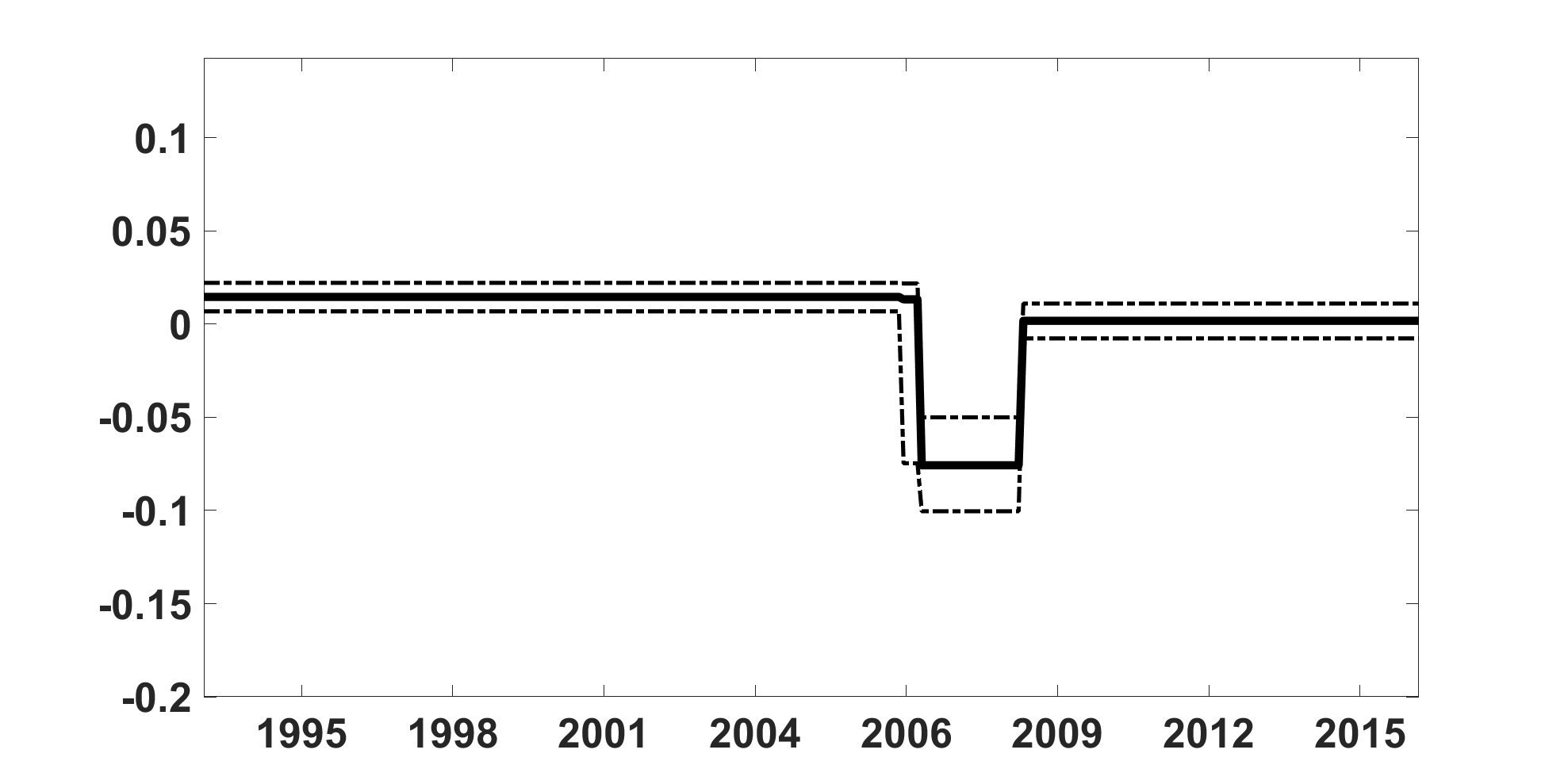}}
		\subfloat[TVP - PTFSSTK]{\includegraphics[width=6.5cm,height=4.5cm]{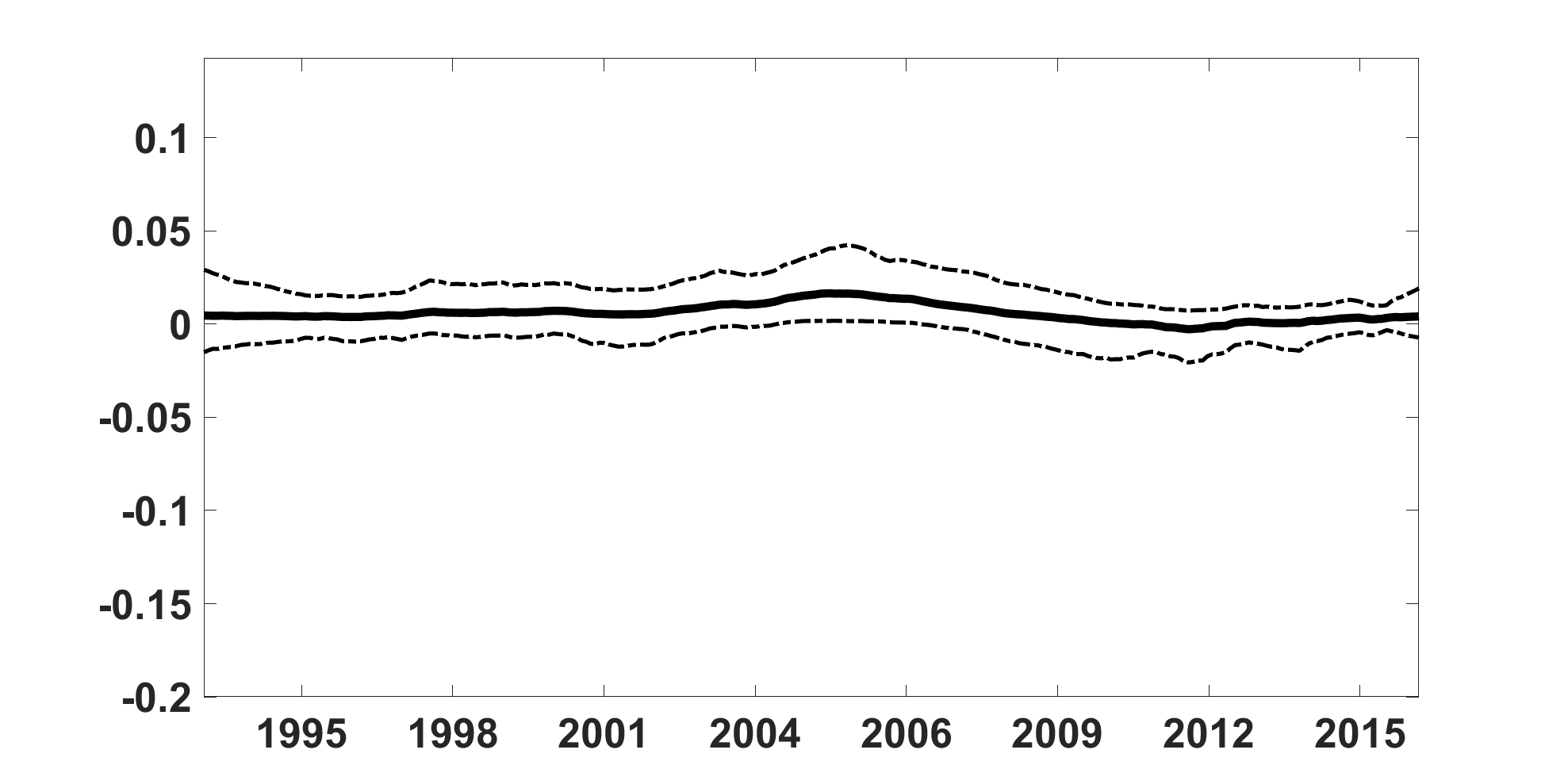}}\\
		\subfloat[SELO - CPI]{\includegraphics[width=6.5cm,height=4.5cm]{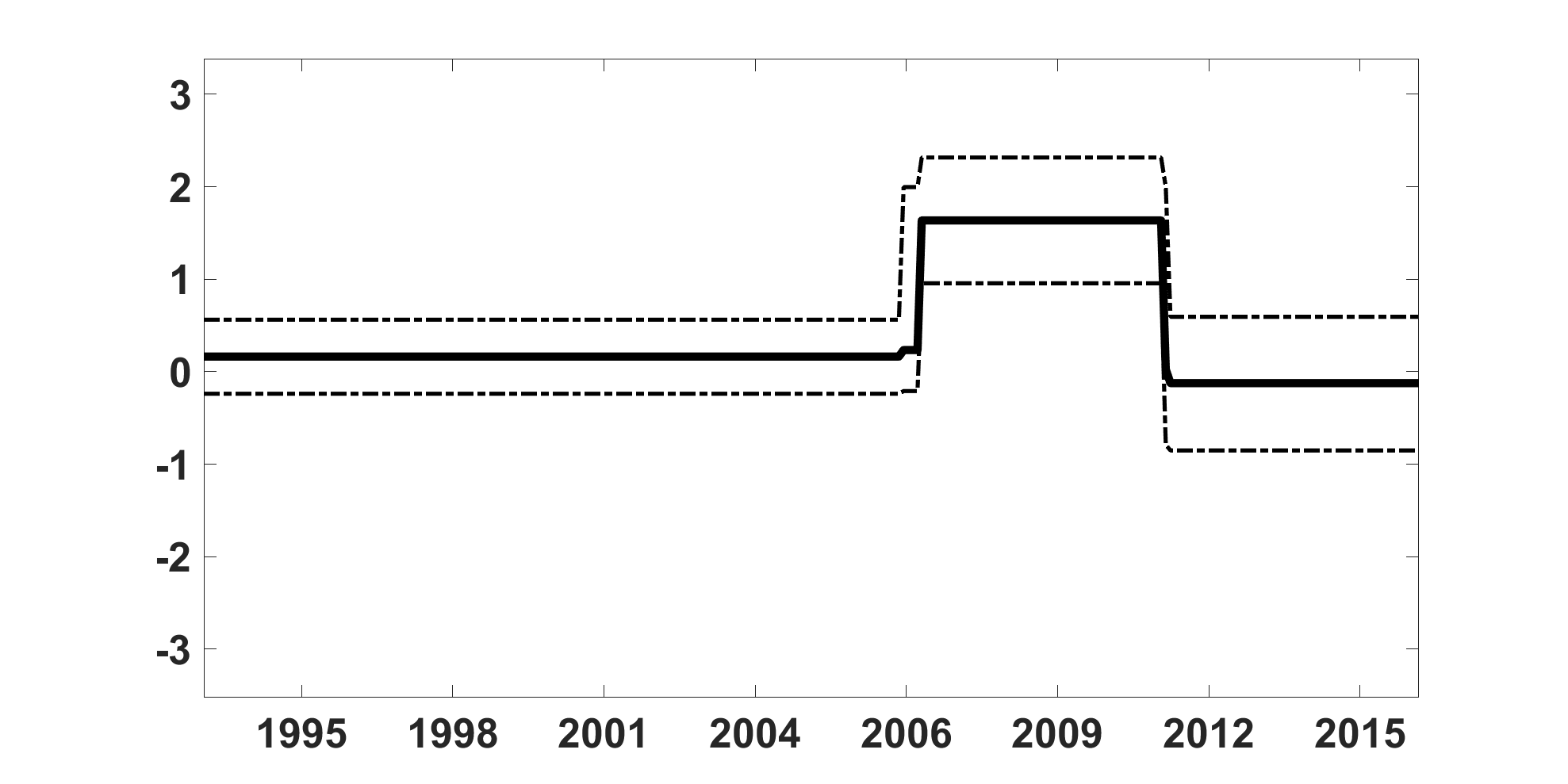}}
		\subfloat[TVP - CPI]{\includegraphics[width=6.5cm,height=4.5cm]{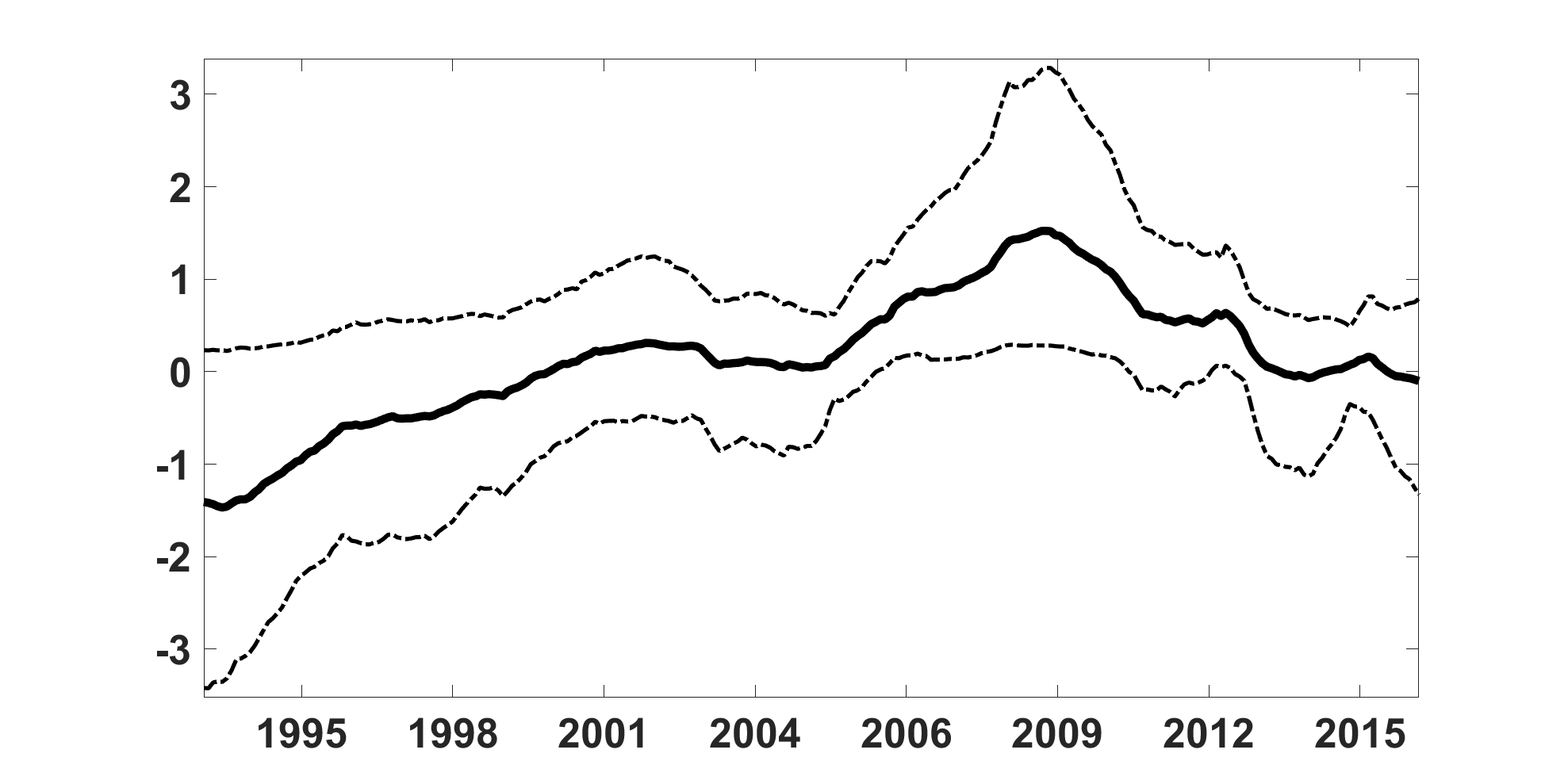}}\\
		\subfloat[SELO - NAREIT]{\includegraphics[width=6.5cm,height=4.5cm]{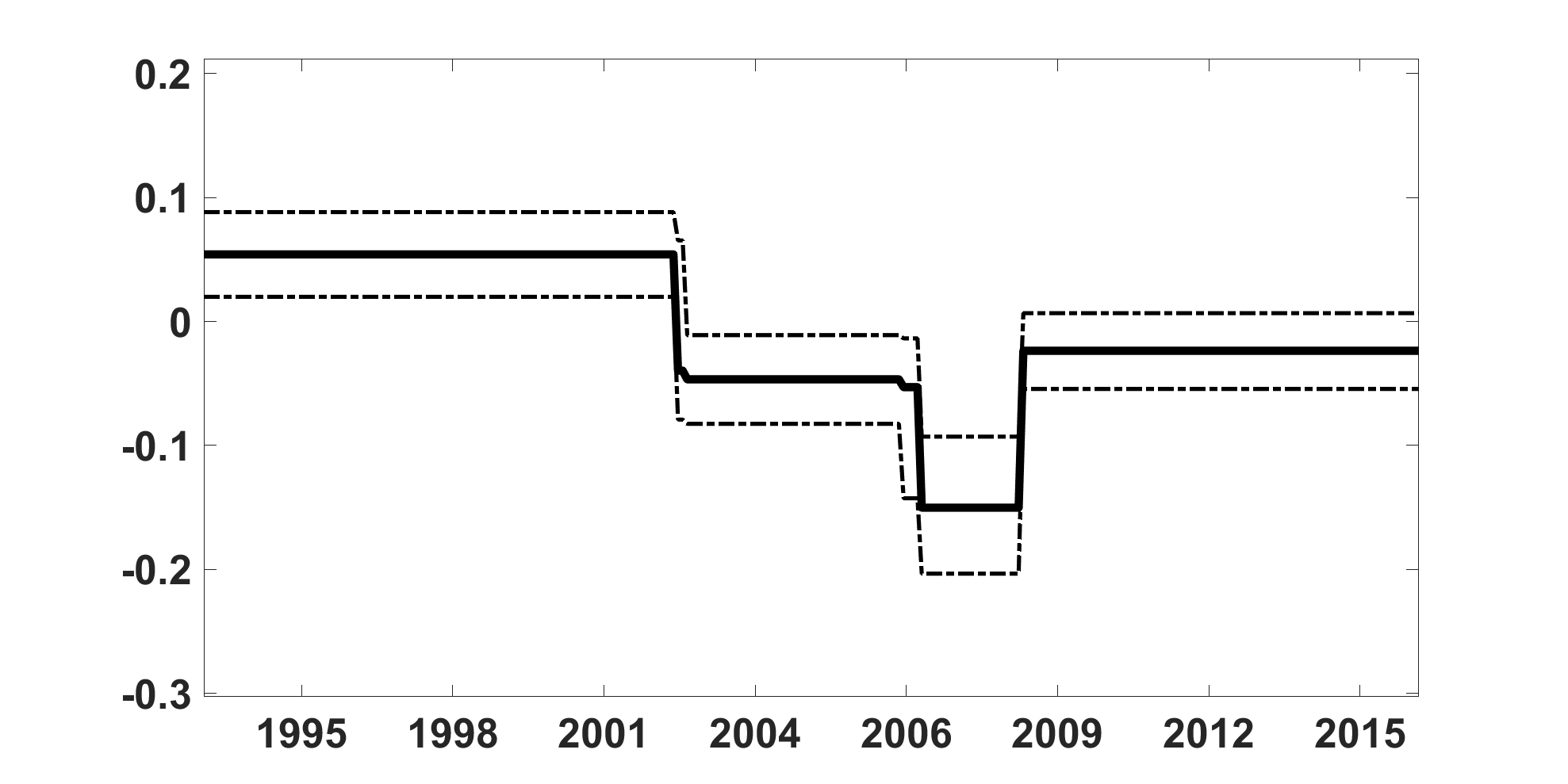}}
		\subfloat[TVP - NAREIT]{\includegraphics[width=6.5cm,height=4.5cm]{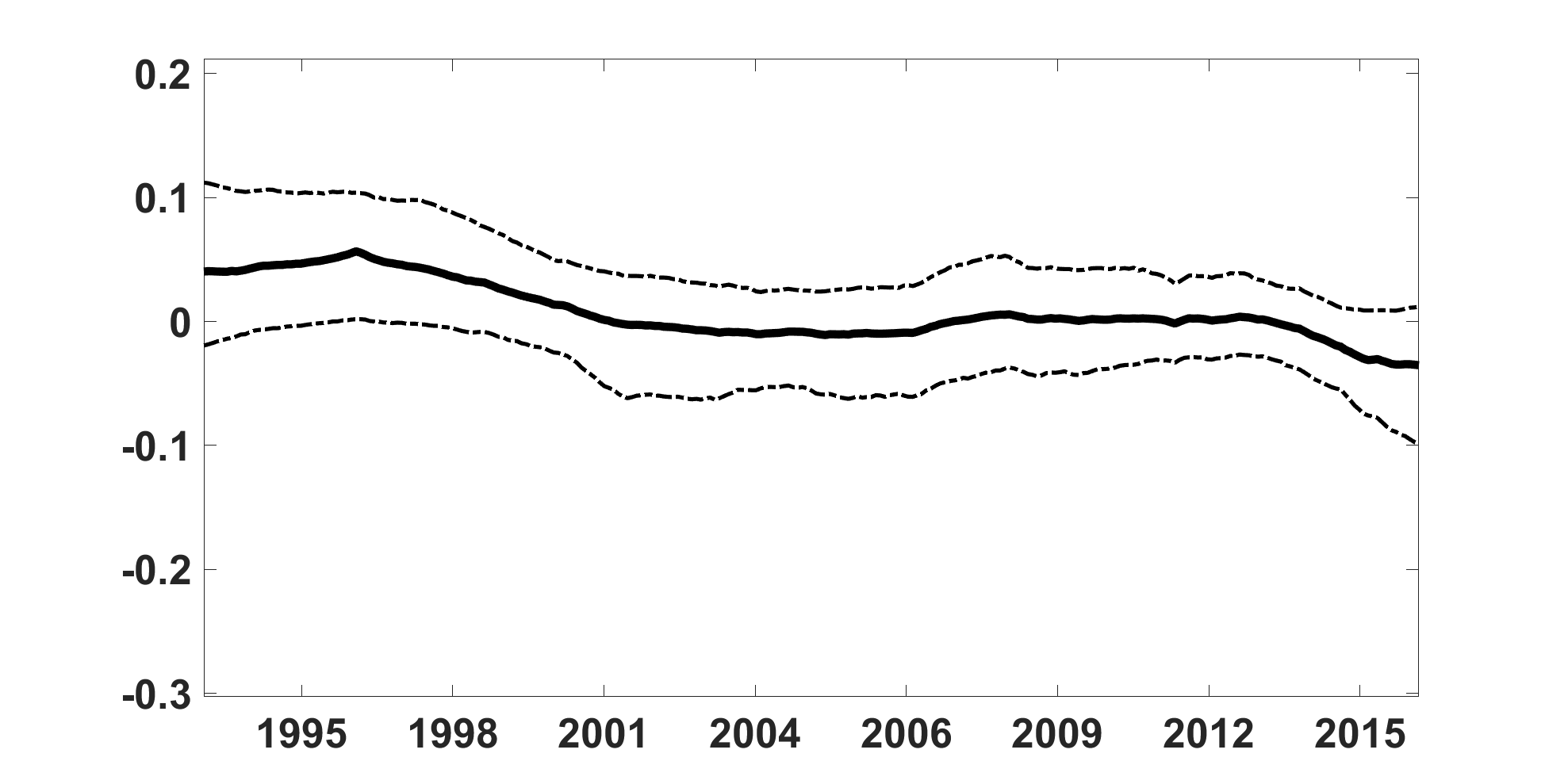}}\\
		\caption{\justifying \textbf{FIA returns - Selective segmentation (SELO) model and Time-varying parameter (TVP) model.} Posterior medians (black) and the 90\% credible intervals (dotted black lines) of the model parameters over time. For the SELO method, we take the break uncertainty into account using the MCMC algorithm presented in Section \ref{sec:breakuncertainty}. \label{fig:strat10_3}}
	\end{center}
\end{figure}
\renewcommand{\baselinestretch}{1.5}

\clearpage
\section{Bayesian alternatives to the selective segmentation method \label{App:Bayesian}}
The DGP J exhibits 100 explanatory variables and one CP in ten parameters. To capture which model parameters are experiencing a break, we are aware of four Bayesian alternatives that are \cite{GiordaniKohn2008}, \cite{Eo2012}, \cite{huber2019should} and \cite{dufays2019relevant}. We now discuss why these alternatives do not work when 100 explanatory variables are involved. 

\subsection*{\cite{GiordaniKohn2008}}
\vspace{-0.3cm}
The model of \cite{GiordaniKohn2008} stands for a particular case of the general mixture state space model developed in \cite{gerlach2000efficient}. They explain how to do  inference for a Gaussian state space model with a latent variable $K_t$ that determines the state of the model parameters. 
Modelling breaks in the mean parameters, their approach allows  estimating the following state space model:
\begin{eqnarray}
y_t & = & \beta_{t,1} + \beta_{t,2} x_{t,2} + \ldots + \beta_{t,N} x_{t,N} + \sigma \eta_t,  \label{eq:obs}\\
\beta_{t,i} & = & \beta_{t-1,i} + \gamma_{i,K_{t,i}} \nu_{t,i}, ~\text{for } i=1,\ldots,N, \text{ and } t>1, \label{eq:TVP}
\end{eqnarray}
in which $N$ is the number of explanatory variables, $\eta_t  \sim N(0,1)$, $K_t=\{K_{t,1},\ldots,K_{t,N}\}$ and $\nu_t = (\nu_{t,1},\ldots,\nu_{t,N})' \sim N(0,I_{N})$ (with $I_N$, the identity matrix of dimension $N$). To capture breakpoints, \cite{GiordaniKohn2008} suggest to set the states of the latent variable $K_{t,i}$ to $\{0,1\}$ such that we have $\gamma_{i,0}=0$ (i.e., no break when $K_{t,i}=0$) and $\gamma_{i,1}\in \Re^{+}$ (i.e., break in the $i$th parameter when $K_{t,i}=1$). The model parameters are given by $\btheta=\{\gamma_{1,1},\ldots,\gamma_{N,1},\beta_{1,1},\ldots,\beta_{1,N},\sigma\}$.

To efficiently estimate the model, \cite{GiordaniKohn2008} relies on the algorithm of \cite{gerlach2000efficient}. The main contribution of \cite{gerlach2000efficient} is to marginalize out the mean parameters $\beta_{1:T,1:N}$ and to provide an analytical formula for the latent variable distribution $f(K_t|y_{1:T},K_{\neq t},\btheta)$ from which $K_t$ is sampled in the MCMC algorithm. To normalize the posterior distribution $f(K_t|y_{1:T},K_{\neq t},\btheta)$, it requires to sum over all the possible values of $K_t$. Because $K_t=\{K_{t,1},\ldots,K_{t,N}\}$ and $K_{t,i}=\{0,1\}~\forall i\in[1,N]$, the number of possible values for $K_t$ amounts to $2^{N}$. Consequently, it increases geometrically with the number of explanatory variables. This is why \cite{chan2012time} on page 9 argue that the number of explanatory variables should be small (i.e., at least below fourteen) otherwise some structure on the break dynamic should be accounted for. With 100 regressors involved in DGP J, it is infeasible to compute the latent variable distribution $f(K_t|y_{1:T},K_{\neq t},\btheta)$ because its normalization requires to sum over $2^{100}$ values.

\subsection*{\cite{Eo2012}}
\vspace{-0.3cm}
\cite{Eo2012} relies on the approach of \cite{Chib98} for finding which parameters are experiencing a break. The method consists of estimating all the possible models given several number of breakpoints. Then, the best specification is selected by maximizing the marginal likelihood that is computed, for each model, using the method of \cite{chib95}. Considering DGP J and its 100 exogenous variables, the number of models to estimate reaches $\sum_{i=0}^{\bar{m}} 2^{100 i}$ in which $\bar{m}$ is the maximum number of breaks that can experience a parameter. In our context, the approach is computationally infeasible even when the upper bound of the number of break is equal to 1.

\subsection*{\cite{huber2019should}}
\vspace{-0.3cm}
\cite{huber2019should} propose a threshold approach to approximate the MCMC inference of mixture state space models. It generalizes the method of \cite{GiordaniKohn2008} because it is not limited by the number of explanatory variables. As illustrated in Appendix D of \cite{DufaysVAR20}, the approximation makes the MCMC inference depending on the starting value and the estimated breakpoints are unstable from one estimation to another. Consequently, in-sample results and forecasting exercises will also depend on starting values. Because the question on how to choose the starting values is not addressed in the paper,  the method does not provide reproducible results. However, it could be useful for exploring the space in order to find a promising starting value to be used in the MCMC algorithm of \cite{GiordaniKohn2008}. 

\subsection*{\cite{dufays2019relevant}}
\vspace{-0.3cm}
\cite{dufays2019relevant} rely on the standard CP model (see, e.g., \cite{Chib98}, \cite{PPT06} or \cite{maheu2014new}) to capture which parameters are time-varying when a break is detected. They specify the model parameters in first-difference with respect to the previous regime. By doing so, shrinkage priors can be used to infer which parameters are time-varying. The two main contributions of the paper are  i) the introduction of a shrinkage prior that is a 2-component mixture of Uniform distributions (hereafter, 2MU) and ii) a method that operates for models exhibiting the path dependence issue such as ARMA and GARCH processes.\\
The new shrinkage prior mimics the standard information criteria such as the AIC and the BIC because one hyper-parameter of the 2MU distribution acts like a penalty on the log-likelihood. Consequently, the 2MU prior can be seen as a Bayesian alternative to the popular $L_0$ penalty functions used in classical statistics. However, the 2MU prior is not suited for high-dimensional regressions because it is not continuous. To mitigate this problem, \cite{dufays2019relevant} propose a sequential Monte Carlo algorithm, which is known to explore multi-modal distributions more efficiently than MCMC algorithms based on a single chain \citep[see, e.g.,][]{herbst2014sequential}. Unfortunately, this algorithm is computationally intensive. While it takes around 10 minutes on a 6-CORE i5-8400 (2.8 Ghz) for estimating a series from DGPs A to F of our paper, it runs for 2.5 hours on the same computer for estimating one series similar to DGP J but with only 20 explanatory variables. Consequently, it is computationally infeasible to estimate the model of \cite{dufays2019relevant} on 100 series from DGP J as we do with the selective segmentation approach. As a comparison, the selective segmentation method requires ??? minutes on the same computer for detecting which parameters are time-varying in one simulated series from DGP J.

\end{document}